\documentclass[11pt,onecolumn,letterpaper]{IEEEtran}

\usepackage[cmex10]{amsmath} 
\makeatletter
\NewDocumentCommand{\seteqnum}{o}{%
  \IfValueTF{#1}
    {\textup{\tagform@{#1}}}
    {\incr@eqnum \print@eqnum}
}
\NewCommandCopy{\ltxlabel}{\ltx@label}
\makeatother
\usepackage{amsfonts}
\usepackage{stmaryrd}
\usepackage{acronym}
\usepackage{amssymb}
\usepackage{amsthm}
\usepackage{dsfont}
\usepackage{graphicx}
\usepackage{paralist}
\usepackage{mdframed}
\usepackage{url}
\usepackage{cite}
\usepackage{hyperref}
\usepackage{flushend}
\usepackage{enumerate}
\usepackage{pdflscape}
\usepackage[dvipsnames]{xcolor}
\usepackage{mathrsfs}
\usepackage{lipsum}
\usepackage{diagbox}
\usepackage{soul}
\usepackage{kantlipsum}

\usepackage{color}
\usepackage{framed}
\usepackage{float}

\usepackage[bbgreekl]{mathbbol}
\DeclareSymbolFontAlphabet{\mathbbm}{bbold}
\DeclareSymbolFontAlphabet{\mathbb}{AMSb}

\graphicspath{{Graphics/}}

\makeatletter
\def\blfootnote{\xdef\@thefnmark{}\@footnotetext}
\makeatother
\newcommand{\indi}[1]{\ensuremath{\mathds{1}}}

\allowdisplaybreaks
\allowbreak

\newcommand{\calA}{\mathcal{A}}

\newcommand{\calB}{\mathcal{B}}

\newcommand{\calC}{\mathcal{C}}

\newcommand{\calD}{\mathcal{D}}
\newcommand{\bbD}{\mathbb{D}}

\newcommand{\calE}{\mathcal{E}}
\newcommand{\bbE}{\mathbb{E}}

\newcommand{\calF}{\mathcal{F}}

\newcommand{\calG}{\mathcal{G}}

\newcommand{\calH}{\mathcal{H}}
\newcommand{\bbH}{\mathbb{H}}

\newcommand{\calI}{\mathcal{I}}
\newcommand{\bbI}{\mathbb{I}}

\newcommand{\calJ}{\mathcal{J}}

\newcommand{\calK}{\mathcal{K}}

\newcommand{\calL}{\mathcal{L}}

\newcommand{\calM}{\mathcal{M}}

\newcommand{\calN}{\mathcal{N}}
\newcommand{\bbN}{\mathbb{N}}

\newcommand{\calO}{\mathcal{O}}

\newcommand{\bbP}{\mathbb{P}}

\newcommand{\calQ}{\mathcal{Q}}

\newcommand{\bbR}{\mathbb{R}}

\newcommand{\calS}{\mathcal{S}}

\newcommand{\calT}{\mathcal{T}}

\newcommand{\calU}{\mathcal{U}}

\newcommand{\calV}{\mathcal{V}}
\newcommand{\bbV}{\mathbb{V}}

\newcommand{\calX}{\mathcal{X}}

\newcommand{\calY}{\mathcal{Y}}

\newcommand{\calZ}{\mathcal{Z}}

\newcommand{\expec}{\mathbb{E}}

\newcommand{\dent}{\mathds{h}}

\DeclareMathOperator*{\argmax}{argmax}

\newcommand{\ToV}{\mathbb{V}}
\newcommand{\Squad}{\hspace{0.5em}}
\newcommand{\brk}[1]{\ensuremath{\big[{#1}\big]}}
\newcommand{\kd}[2]{\ensuremath{\bbD\hspace{-0.7mm}\left({#1}{#2}\right)}}
\newcommand{\sbra}[2]{\ensuremath{\left[{#1}{\,,\,}{#2}\right]}}%
\newcommand{\sbr}[1]{\ensuremath{\left[{#1}\right]}}%
\newtheorem{theorem}{Theorem}
\newtheorem{corollary}{Corollary}
\newtheorem{definition}{Definition}

\newtheorem{remark}{Remark}

\newtheorem{lemma}{Lemma}[]

\newcommand{\indic}[1]{\ensuremath{\mathds{1}}}
\newcommand{\card}[1]{\ensuremath{\left|{#1}\right|}}   
\newcommand{\abs}[1]{\ensuremath{\left|#1\right|}}   

\newcommand{\cc}{\text{\textnormal{c}}}
\newcommand{\tck}{\text{\textnormal{tck}}}
\newcommand{\tcsk}{\text{\textnormal{tcsk}}}
\newcommand{\scn}{\text{\textnormal{sc}}}
\newcommand{\A}{\text{\textnormal{A}}}
\newcommand{\B}{\text{\textnormal{B}}}
\newcommand{\Ci}{\text{\textnormal{C}}}
\renewcommand{\leq}{\leqslant} 
\renewcommand{\geq}{\geqslant} 

\acrodef{ACDIS}[ACDIS]{Adaptive Communication Decision and Information Systems}
\acrodef{AEP}{Asymptotic Equipartition Property}
\acrodef{AoA}{Angle of Arrival}
\acrodef{AWGN}{Additive White Gaussian Noise}
\acrodef{AVC}[AVC]{Arbitrarily Varying Channel}
\acrodefplural{AVC}{Arbitrarily Varying Channels}
\acrodef{PIR-PNSI}{Private Information Retrieval with Private Noisy Side Information}
\acrodef{BER}{Bit-Error-Rate}
\acrodef{BEC}{Binary Erasure Channel}
\acrodefplural{BEC}{Binary Erasure Channels}
\acrodef{BSC}{Binary Symmetric Channel}
\acrodefplural{BSC}{Binary Symmetric Channels}
\acrodef{BSCO}{Binary Symmetric Channel with Additional ``off'' Symbol}
\acrodefplural{BSCO}{Binary Symmetric Channels with Additional ``off'' Symbols}
\acrodef{FDG}{Functional Dependence Graph}
\acrodef{BPSK}{Binary Phase-Shift Keying}
\acrodef{BICM}[BICM]{Bit-Interleaved Coded-Modulation}
\acrodef{CDF}[CDF]{Cumulative Distribution Function}
\acrodef{CGF}[CGF]{Cumulant Generating Function}
\acrodef{CLT}[CLT]{Central Limit Theorem}
\acrodef{CSI}[CSI]{Channel State Information}
\acrodef{DMC}[DMC]{Discrete Memoryless Channel}
\acrodefplural{DMC}{Discrete Memoryless Channels}
\acrodef{DMS}[DMS]{Discrete Memoryless Source}
\acrodef{ERM}[ERM]{Empirical Risk Minimization}
\acrodef{FER}[FER]{Frame Error Rate}
\acrodef{ICA}[ICA]{Independent Component Analysis}
\acrodef{iid}[i.i.d.]{independent and identically distributed}
\acrodef{IoT}[IoT]{Internet of Things}
\acrodef{KKT}[KKT]{Karush-Kuhn Tucker}
\acrodef{LASSO}[LASSO]{Least Absolute Shrinkage and Selection Operator}
\acrodef{LPD}[LPD]{Low Probability of Detection}
\acrodef{LDPC}[LDPC]{Low-Density Parity-Check}
\acrodef{LLMS}[LLMS]{Linear Least Mean Square}
\acrodef{LMS}[LMS]{Least Mean Square}
\acrodef{MAC}[MAC]{Multiple-Access Channel}
\acrodef{ADSI}[ADSI]{Action-Dependent State Information}
\acrodef{MGF}[MGF]{Moment Generating Function}
\acrodef{MLC}[MLC]{Multi-Level Coding}
\acrodef{MLE}[MLE]{Maximum Likelihood Estimate}
\acrodef{MIMO}[MIMO]{Multiple-Input Multiple-Output}
\acrodef{MISO}{Multiple-Input Single-Output}
\acrodef{MSD}[MSD]{Multi-Stage Decoding}
\acrodef{MMSE}[MMSE]{Minimum Mean-Square Error}
\acrodef{PAC}[PAC]{Probably Approximately Correct}
\acrodef{PCA}[PCA]{Principal Component Analysis}
\acrodef{PDF}[PDF]{Probability Density Function}
\acrodefplural{PDF}{Probability Density Functions}
\acrodef{PMF}[PMF]{Probability Mass Function}
\acrodefplural{PMF}{Probability Mass Functions}
\acrodef{PPM}[PPM]{Pulse Position Modulation}
\acrodef{PSD}{Power Spectral Density}
\acrodef{PSK}{Phase Shift Keying}
\acrodef{QKD}{Quantum Key Distribution}
\acrodef{ROC}{Receiver Operating Characteristic}
\acrodef{CVQKD}{Continuous-Variable \ac{QKD}}
\acrodef{QPSK}{Quadrature Phase-Shift Keying}
\acrodef{RV}{random variable}
\acrodefplural{RV}{random variables}
\acrodef{SIMO}{Single-Input Multiple-Output}
\acrodef{SNR}{Signal-to-Noise Ratio}
\acrodef{SVM}[SVM]{Support Vector Machine}
\acrodef{TPCP}{Trace-Preserving Completely-Positive}
\acrodef{wrt}[w.r.t.]{with respect to}
\acrodef{WSS}{Wide Sense Stationary}
\acrodef{RHS}{Right Hand Side}
\acrodef{LHS}{Left Hand Side}
\acrodef{PIR}{Private Information Retrieval}
\acrodef{MDS}{Maximum Distance Separable}
\acrodef{LLN}{Law of Large Numbers}
\acrodef{DFRC}{Dual-Function Radar Communication}
\acrodef{ISAC}{Integrated Sensing and Communication}
\acrodef{RadCom}{Joint Radar and Communicatins}

\begin{document}

\title{Covert Communication via Action-Dependent States}

\author{
\IEEEauthorblockN{Hassan ZivariFard and Xiaodong Wang}\\
\thanks{The authors are with the Department of Electrical Engineering, Columbia University, New York, NY 10027. This work is supported in part by the U.S. Office of Naval Research (ONR) under grant N000142412212. E-mails: \{hz2863, xw2008\}@columbia.edu. Part of this work is presented at the 2023 IEEE International Symposium on Information Theory~\cite{ISIT23}.}
}
\maketitle
\date{}

\begin{abstract}
\label{sec:Abstract}
This paper studies covert communication over channels with \ac{ADSI} when the state is available either non-causally or causally at the transmitter. Covert communication refers to reliable communication between a transmitter and a receiver while ensuring a low probability of detection by an adversary, which we refer to as ``warden''. It is well known that in a point-to-point \ac{DMC}, it is possible to communicate on the order of $\sqrt{N}$ bits reliably and covertly over $N$ channel uses while the transmitter and the receiver are required to share a secret key on the order of $\sqrt{N}$ bits.  
This paper studies achieving reliable and covert communication of positive rate, i.e., reliable and covert communication on the order of $N$ bits in $N$ channel uses, over a channel with \ac{ADSI} while the transmitter has non-causal or causal access to the \ac{ADSI}, and the transmitter and the receiver share a secret key of negligible rate. We derive achievable rates for both the non-causal and causal scenarios by using block-Markov encoding and secret key generation from the \ac{ADSI}, which subsumes the best achievable rates for channels with random states. We also derive upper bounds, for both non-causal and causal scenarios, that meet our achievable rates for some special cases. As an application of our problem setup, we study covert communication over channels with rewrite options, which are closely related to recording covert information on memory, and show that a positive covert rate can be achieved in such channels. As a special case of our problem, we study the \ac{AWGN} channels and provide lower and upper bounds on the covert capacity that meet when the transmitter and the receiver share a secret key of sufficient rate and when the warden's channel is noisier than the legitimate receiver channel. As another application of our problem setup, we show that cooperation can lead to a positive covert rate in Gaussian channels. A few other examples are also worked out in detail.
\end{abstract}

\section{Introduction}
\label{sec:Intro}
The inherent broadcast nature of communication networks, while beneficial, also facilitates easier interference or tampering with sensitive information by malicious users \cite{BlochBarros}. Motivated by this challenge, covert communication offers a solution to preserve the users' privacy \cite{Bash13,CheISIT13,Bloch16,Wang16}. 
The present paper studies covert communication over channels with action-dependent states. Our motivation for studying this problem stems from the importance of privacy and security in communication systems and the implications and practical use of the \textit{action} in modeling important problems in communications, including communication with control over state, feedback, and data \cite{Dai20,Information_Embedding,Feedbak_or_Not}, recording on magnetic memory \cite{Weissman10,MultiStage_Writing}, recording on computer memories with defects \cite{Heegard83,Weissman10}, and channels with input cost constraints \cite{Weissman10}.

Covert communication refers to a problem where a transmitter wishes to communicate reliably with a receiver over a channel while ensuring that the distribution of channel observations seen by an adversary remains indistinguishable from that observed when the transmitter sends an innocent symbol~\cite{Bash13,CheISIT13,Bloch16,Wang16,Mehrdad19}. In a point-to-point \ac{DMC}, it is well known that a transmitter can communicate at most on the order of $\sqrt{N}$, i.e., $\calO(\sqrt{N})$, covert and reliable bits over $N$ channel uses with a receiver~\cite{Bloch16,Wang16}, while the transmitter and the receiver are required to share $\calO(\sqrt{N})$ secret key bits~\cite{Bloch16}. In a point-to-point \ac{DMC}, it is possible to achieve a positive covert and reliable communication rate, i.e., $\calO(N)$ covert and reliable bits in $N$ channel uses, only when the innocent symbol $x_0\in\calX$ is \textit{redundant}, which is the case where the distribution induced on the channel observation of the warden by the innocent symbol $x_0\in\calX$ can be written as a convex combination of the distribution induced on the channel observation of the warden by $\{x\in\calX:x\ne x_0\}$~\cite{Wang16}. 

Besides the special case described above, there are a few scenarios where it is possible to go beyond the square root law and communicate with a positive covert rate, which we summarize below.
\begin{enumerate}[i)]
    \item When the warden has uncertainty about the statistical characterization of its channel~\cite{Lee15,Deniable_ITW14}: the uncertainty in what the warden expects to observe when the transmitter is not communicating with the receiver can be leveraged to obtain positive covert and reliable communication rate. We note that \cite{Lee15} assumes that there is a secret key of infinite rate between the transmitter and the receiver, and \cite{Deniable_ITW14} assumes that the warden's channel is degraded \ac{wrt} the legitimate receiver's channel.
    \item When a friendly jammer is present~\cite{Sobers17,Shahzad18,Shmuel19,ISIT21,ISIT22,MyDissertation}: The randomness introduced by the friendly jammer increases the warden's uncertainty about the statistics of its channel, allowing the transmitter to go beyond the square root law and transmit covert information on the order of $\calO(N)$ bits. We note that \cite{Sobers17,Shahzad18,Shmuel19} assume that there is a secret key of infinite rate between the transmitter and the receiver.
    \item When the \ac{CSI} \cite{StateDepChan} is known by the transmitter~\cite{LeeWang18,Keyless22}: this scenario is similar to the previous scenario with a friendly jammer. The randomness introduced by nature increases the warden's uncertainty about its channel statistics, while the transmitter's knowledge of \ac{CSI} enables it to hide information in the \ac{CSI}. This allows the transmitter to overcome the square root law and covertly transmit on the order of $\calO(N)$ bits. Note that the \ac{CSI} is \ac{iid} and is chosen by the nature.
\end{enumerate}

Of particular relevance, \cite{LeeWang18} studies covert communication over a state-dependent channel when \ac{CSI} is available either causally or non-causally at the transmitter and the transmitter and the receiver share a secret key. The authors provide achievable rate regions for this problem and show that their scheme is optimal when the transmitter and the receiver share a secret key of infinite rate. In \cite{Keyless22}, the authors study covert communication over the state-dependent channels when the \ac{CSI} is available either non-causally, causally, or strictly causally, either at the transmitter alone or at both the transmitter and the receiver. When the \ac{CSI} is available at both the transmitter and receiver, the authors derive the covert capacity by simultaneously using the \ac{CSI} for communication and secret key generation. When the \ac{CSI} is available only at the transmitter the authors derive lower and upper bounds on the covert capacity, which improves upon the achievable rates in \cite{LeeWang18}, by using the \ac{CSI} for both communication and secret key generation simultaneously. Another relevant extension of the state-dependent channels is studied in \cite{Steinberg09} in the context of information embedding, where apart from decoding the embedded message, the decoder has an additional constraint on reconstructing the channel input signal reliably.
    
Secure communication over state-dependent channels is studied in \cite{ChenVinck08,ChiaElGamal12,HanSasaki19,ZivBC17,ZivCSIT20}. The trade-off between secret message and secret key rates, simultaneously achievable over a state-dependent wiretap channel with causal and non-causal \ac{CSI} at the encoder is studied in \cite{Bunin20,HanSasaki21}. Secure communication over \ac{MAC} is studied in \cite{GMAWCJamming,YassaeeMAWC,Frey18}, and secure communication over \ac{MAC} with cribbing is studied in \cite{Helal20}. Channels with action-dependent states are first introduced in~\cite{Weissman10}, where the author derives the capacity when \ac{ADSI} is available either causally or non-causally at the encoder and studies the channels with rewrite option as a special case of the channels with \ac{ADSI}. This problem is extended to various scenarios in~\cite{Gaussian_ADSI,Information_Embedding,Feedbak_or_Not,Secure_Source_Coding_Action}. \ac{MAC} with \ac{ADSI} at one encoder is studied in~\cite{Dikstein15}, where the authors derive the capacity region for this problem. Also, secure communication over channels with action-dependent states is studied in~\cite{Dai20}.  

In this paper, we study covert communication over a \ac{DMC} with action-dependent states, when the \ac{ADSI} is available either non-causally or causally at the transmitter, as illustrated in Fig.~\ref{fig:System_Model}. 
When the \ac{ADSI} is available non-causally at the transmitter, we derive an achievable covert rate and an upper bound on the covert capacity. 
As an example, we study this problem in the context of \ac{AWGN} channels, and provide lower and upper bounds on the covert capacity which meet when the transmitter and the receiver share a secret key of sufficient rate and when the warden's channel output observation is ``noisier" than the legitimate receiver's channel output observation. 
When the \ac{ADSI} is available causally at the transmitter, we provide lower and upper bounds on the covert capacity, which meet when the legitimate receiver's channel output observation is less noisy \ac{wrt} the warden's channel output observation and when the receiver and the warden observe the same channel output. This problem could be considered as a generalization of \cite{Keyless22}, in which the \ac{CSI} is generated by nature, whereas in this paper the \ac{CSI} is partially controlled by the transmitter. Compared to the previous works \cite{LeeWang18,Keyless22}, in which \ac{CSI} is solely determined by the environment, the additional degree of freedom in designing the action sequence leads to an advantage in overcoming the square-root law by hiding more information in states. 

As the first application of our problem setup and our results, we study covert communications over channels with a rewrite option, which we refer to as \textit{writing a clean memory}, this is because this problem recovers as a special case, the problem of recording on a computer memory while ensuring that the distribution of the outcome of the recording process is almost identical to the distribution of an erased (clean) computer memory. We consider various scenarios of this problem and visualize our results with some numerical examples.  
As another application of our setting, we show that cooperation between two users can lead to a positive covert rate. We also show that our results can be used to study covert communication over a state-dependent \ac{MAC} with a common message and \ac{CSI} known at one transmitter as studied in \cite{Cooperative_MAC}. 

Our coding scheme combines different code constructions including channel resolvability for the covertness analysis, randomness extraction from \ac{ADSI} for secret key generation, Gel'fand-Pinsker encoding for message transmission, Block-Markov encoding for the dual use of the \ac{ADSI}, and rate splitting for the efficient use of the generated secret key. The key technical challenge consists of properly combining these mechanisms to ensure the overall reliability and covertness of the transmission through block-Markov chaining schemes. Compared with the achievability schemes in \cite[Theorem~4 and 7]{Keyless22} the main challenge in this paper is dealing with the non \ac{iid} nature of the \ac{ADSI}, which complicates the covert analysis, and compared with the achievability scheme in \cite{Weissman10} the main challenge is exploiting the \ac{ADSI} for both message transmission and secret key generation through a block Markov encoding scheme and inducing the distribution corresponding to the no-communication mode on warden's channel output observation. 

\textit{Notation:} Let $\bbN_*$ be the set of positive natural numbers, and $\bbR$ be the set of real numbers. We define $\bbR_+\triangleq\{x\in\bbR|x\ge0\}$. For any $j\in\bbN_*$, $[j]$ denotes the set $\{1,2,\dots,j\}$. Here, \acp{RV} are denoted by capital letters and their realizations by lowercase letters. The expectation \ac{wrt} the \ac{RV} $X$ is denoted by $\bbE_{X}[\cdot]$ and the indicator function is denoted by $\indi{1}_{\{\cdot\}}$. For a set of \acp{RV} $\{X_i\}_{i\in\calI}$, indexed over a countable set $\calI$, $\bbE_{\backslash j}[\cdot]$ denotes the expectation \ac{wrt} all $X_i$, for $i\in\calI$, excluding $X_j$, i.e., $\bbE_{X_1,X_2,\dots,X_{j-1},X_{j+1},\dots,X_{\abs{\calI}}}[\cdot]$. Superscripts indicate the dimension of a vector, e.g., $X^N$. $X_{\sim i}^N$ denotes $(X_1,\dots,X_{i-1},X_{i+1},\dots,X_N)$, and $X_i^j$ denotes $(X_i,X_{i+1},\dots,X_j)$. The cardinality of a set is denoted by $|\cdot|$. For a countable set $\calX$, the KL-Divergence between the two distribution $P_X$ and $Q_X$ is denoted $\bbD(P_X||Q_X)=\sum_{x\in\calX}P_X(x)\log\frac{P_X(x)}{Q_X(x)}$ and the total variation distance between the two distribution $P_X$ and $Q_X$ is denoted by  $\bbV(P_X,Q_X)=\frac{1}{2}\sum_{x\in\calX}\left|P_X(x)-Q_X(x)\right|$. The product distribution is denoted by $P_X^{\otimes N}\triangleq\prod_{i=1}^NP_X(x_i)$ and the $\log$ functions are assumed to be base 2 logarithm in this paper. We use $\calT_{\epsilon}^{(N)}(P_{XY})$ to denote the $\epsilon$-strongly jointly typical sequences of length $N$ \ac{wrt} $P_{XY}$ and $\epsilon>0$. The set of $\epsilon$-strongly typical sequences of length $N$ \ac{wrt} $P_X$ is defined as,
    \begin{align}
        &\calT_\epsilon^{(N)}(P_X)\triangleq\left\{x^N\in\calX^N:\left|\frac{\textup{Nu}\left(\bar{x}|x^N\right)}{N}-P_X(\bar{x})\right|\le\epsilon P_X(\bar{x}),\forall\bar{x}\in\calX\right\},\nonumber
    \end{align}where $\textup{Nu}\left(\bar{x}|x^N\right)=\sum_{i=1}^N\indi{1}_{\{x_i=\bar{x}\}}$.
\begin{figure}[t]
\centering
\includegraphics[width=4.5in]{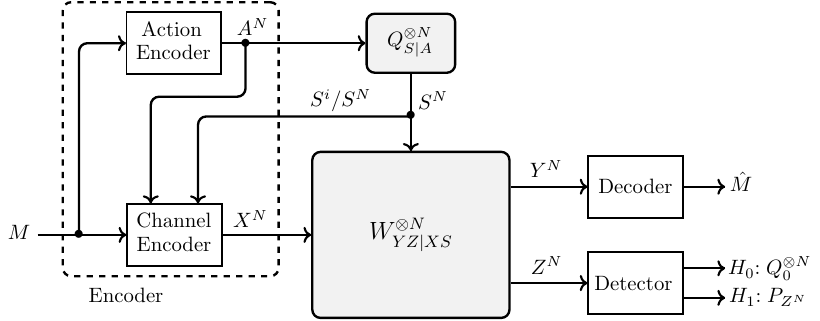}
\caption{Covert communication over channels with action-dependent states}
\label{fig:System_Model}
\end{figure}

The remainder of this paper is organized as follows. The system model is given in Section~\ref{sec:Problem}. The main results are provided in Section~\ref{sec:Main_Results}. Applications of our problem setup to channels with a rewrite option are presented in Section~\ref{sec:rewrite}. The \ac{AWGN} channels and cooperative Gaussian channels are studied in Section~\ref{sec:Gaussian}. Finally, the results are summarized in Section~\ref{sec:Conclusion}.
\section{System Model}
\label{sec:Problem}
Consider a \ac{DMC} $\big(\calA,\calS,\calX,\calY,\calZ,Q_{S|A},W_{YZ|XS}\big)$, as depicted in Fig.~\ref{fig:System_Model}, where $\calA$ is the action alphabet, $\calS$ is the \ac{ADSI} alphabet, $\calX$ is the channel input alphabet, $\calY$ and $\calZ$ are the channel output alphabets for the legitimate receiver and the warden, respectively, $Q_{S|A}$ is the \ac{ADSI} \ac{PMF}, and $W_{YZ|XS}$ is the channel law. We define the innocent symbol for the channel input as $x_0\in\calX$ and the innocent symbol for the action as $a_0\in\calA$, both representing the no-communication mode.  Therefore, the distribution induced at the warden's channel output when communication does not take place is $Q_0^{\otimes N}\triangleq\prod\nolimits_{i = 1}^N Q_0$, where
\begin{align}
    Q_0(\cdot)\triangleq\sum_{s\in\calS} Q_{S|A}(s|a_0)W_{Z|XS}(\cdot|x_0,s).\label{eq:q0}
\end{align}The model described above is essentially a \ac{DMC} with two-stage coding, in which what we are referring to as ``actions" is nothing but the channel input of the first stage of the encoding. Following the first stage of the encoding is the standard problem of encoding over a state-dependent channel. The \ac{ADSI} is assumed to be known either non-causally or causally to the transmitter, but unknown to the receiver and the warden.  
A code is formally defined as follows.
\begin{definition}
\label{defi:Code}
A $(2^{NR},N)$ code $\calC_N$ with the \ac{ADSI} available non-causally or causally at the encoder consists of:
\begin{itemize}
    \item a message set $\calM\triangleq\sbr{2^{NR}}$, a secret key set $\calK\triangleq\sbr{2^{\bar{R}_K^{(N)}}}$ of negligible rate, i.e., $\lim_{N\to\infty}\bar{R}_K^{(N)}\to0$, and local randomness sets $\calQ_1$ and $\calQ_2$;
    \item a stochastic action encoder at the transmitter $\calE_{1,t}:\calM\times\calK\times\calQ_1\times\calA^{t-1}\mapsto \calA_t$, for each time slot $t\in\sbr{N}$, that maps the message $m\in\calM$, the secret shared key $k\in\calK$, a realization of the local randomness $q_1\in\calQ_1$, and the past symbols of the action encoder output $a^{t-1}\in\calA^{t-1}$, to the action symbol $a_t\in\calA_t$\footnote{In this paper, we do not account for the rate of local randomness the transmitter needs, and by convention we assume that $a^0=x^0=0$.};
    \item when the transmitter has non-causal (or causal) access to the \ac{ADSI}, a stochastic encoder at the transmitter $\calE_{2,t}:\calM\times\calK\times\calQ_2\times\calA^N\times\calX^{t-1}\times\calS^N(\text{or}\,\,\calS^t)\mapsto\calX_t$, for each time slot $t\in\sbr{N}$, that maps the message $m\in\calM$, the secret shared key $k\in\calK$, a realization of the local randomness $q_2\in\calQ_2$, the action sequence $a^n\in\calA^N$, the past symbols of the encoder output $x^{t-1}\in\calX^{t-1}$, and the \ac{ADSI} sequence $s^N\in\calS^N$ (or the \ac{ADSI} sub-sequence $s^t\in\calS^t$) to a channel input symbol $x_t\in\calX_t$;
    \item a decoding function $\calD: \calY^N\times\calK\mapsto\calM\cup\{\mathfrak{e}\}$, that maps the channel output $y^N\in\calY^N$ and the secret shared key $k\in\calK$ to an estimate of the message $\hat{m}\in\calM$ or an error $\mathfrak{e}$.
\end{itemize}
\end{definition}
The main reason that we use a stochastic encoder instead of a deterministic encoder is that our achievability scheme is based on the Likelihood encoder \cite{Cuff13,Yassaee13,Watanabe15}, which is a stochastic encoder and helps approximate the desired distribution of the warden's channel observations. 
We assume that the code is public knowledge and is known by all the terminals, including the warden. The transmitter and the receiver aim to design a code that is both reliable and covert. The code is reliable if the probability of error $P_e^{(N)}\triangleq\bbP(\hat{M}\ne M)$ vanishes when $N$ grows. 
The code is covert if the warden cannot distinguish the two hypotheses corresponding to the situation in which that communication is happening (hypothesis $H_1$) or communication is not happening (hypothesis $H_0$). The probabilities of false alarm and missed detection are denoted by $\alpha_N$ and $\beta_N$, respectively. A blind random decision for the warden, where he ignores his channel observations, satisfies $\alpha_N+\beta_N=1$. The optimal hypothesis test by the warden satisfies $\alpha_N+\beta_N\ge 1-\sqrt{\bbD\left(P_{Z^N}||Q_0^{\otimes N}\right)}$ \cite{HypothesesTesting}, where $P_{Z^N}$ denote the distribution induced on the warden's observation when the transmitter is communicating with the receiver. Hence, to prove that the communication is covert we aim to show that $\bbD\left(P_{Z^N}||Q_0^{\otimes N}\right)\to0$, where we assume that $\text{supp}(W_{Z|X}(\cdot|x))\subseteq\text{supp}(Q_0)$, for $x\in\calX$, since otherwise $\bbD\left(P_{Z^N}||Q_0^{\otimes N}\right)\to\infty$. Therefore, we aim to design a code such that
\begin{subequations}
\begin{align}
    &P_{e}^{(N)}\xrightarrow[]{N\to\infty}0,\label{eq:PE}\\
    &\bbD\left(P_{Z^N}||Q_0^{\otimes N}\right)\xrightarrow[]{N\to\infty}0.\label{eq:Covertness}
\end{align}
\end{subequations}We define the covert capacity as the supremum of all achievable rates and denote it as $\textup{C}_{\mbox{\scriptsize\rm AD-NC}}$ and $\textup{C}_{\mbox{\scriptsize\rm AD-C}}$, when the \ac{ADSI} is available non-causally and causally at the encoder, respectively.

The following lemma is frequently used in this paper and relates the KL-Divergence with the total variation distance.  
\begin{lemma}
    \label{lemma:KLD_TV}
    According to Pinsker's inequality, for two distributions $P$ and $Q$ defined over the alphabet set $\calX$ we have,
    \begin{subequations}
    \begin{align}
        \bbV(P,Q)\le\sqrt{\frac{1}{2}\bbD(P||Q)}.\label{eq:Pinsker_Ineq}
    \end{align}There is no reverse Pinsker inequality.
    However, when the alphabet set $\calX$ is finite and $P_n$ is absolutely continuous \ac{wrt} $Q^{\otimes n}$ we have~\cite[Eq.~(30)]{Cuff13},
    \begin{align}
        \bbD(P_n||Q^{\otimes n})\in\calO\left(\left(n+\log\frac{1}{\bbV(P_n,Q^{\otimes n})}\right)\bbV(P_n,Q^{\otimes n})\right),\label{eq:Reverse_Pinsker}
    \end{align}where $f(n)\in\calO(g(n))$ means that $f(n)\le kg(n)$, for some $k$ independent of $n$ and sufficiently large $n$. In particular, \eqref{eq:Reverse_Pinsker} implies that an exponential decay of the total variation distance in $n$ produces an exponential decay of the relative entropy with the same exponent.
    \end{subequations}
\end{lemma}
\section{Main Results}
\label{sec:Main_Results}
In this section, we provide the main results of the paper, which include a lower and an upper bound on the covert capacity when the \ac{ADSI} is available non-causally or causally at the encoder. We compare our results with the previous works in the literature and provide some examples and insights.
\subsection{\texorpdfstring{\ac{ADSI}}{ADSI} Available Non-Causally at the Transmitter}
\label{sec:NC}
The following theorem presents a lower bound on the covert capacity when the \ac{ADSI} is available non-causally at the encoder.
\begin{theorem}
\label{thm:Acievability_KG}
Let
\begin{subequations}\label{eq:Achievability_AD_NC}
\begin{align}
  \calF_{\text{L-NC}} =\left\{
    \begin{aligned}
&R\geq 0: \exists P_{ASUVXYZ}\in\calG_{\text{L-NC}}: \\
  &R< \bbI(A,U;Y) - \bbI(U;S|A)& \quad\seteqnum[a]\\
  &R< \bbI(A,U,V;Y) -  \bbI(U,V;S|A)& \quad\seteqnum[b]
\end{aligned}
\right\},\label{eq:Achievability_A_NC}
\end{align}
where 
\begin{align}
  \calG_{\text{L-NC}}\triangleq \left\{
    \begin{aligned}&P_{ASUVXYZ}:& \\
&P_{ASUVXYZ}=P_AP_{U|A}P_{V}Q_{S|AUV}P_{X|US}W_{YZ|XS}&\\
&Q_{S|A}(\cdot|a)=\sum_{v\in\calV}\sum_{u\in\calU}P_V(v)P_{U|A}(u|a)Q_{S|AUV}(s|a,u,v)&\quad\seteqnum[c]\\
&\bbI(V;Y|A,U)>\bbI(V;Z)& \quad\seteqnum[d]\\
&\bbI(A,U,V;Y)\ge\bbI(A,U,V;Z)& \quad\seteqnum[e]\\
&\bbI(A,U,V;Y)\ge\bbI(A,V;Z)+\bbI(U;S|A)& \quad\seteqnum[f]\\
&\bbI(A,U,V;Y) + \bbI(V;Y|A,U)\ge\bbI(A,V;Z) + \bbI(U,V;S|A)& \quad\seteqnum[g]\\
&P_Z=Q_0& \quad\seteqnum[h]\\
&\card{\calU}\leq\card{\calA}\card{\calS}\card{\calX}+4,\,\card{\calV}\leq\left(\card{\calA}\card{\calS}\card{\calX}+4\right)^2& \quad\seteqnum[i]
\end{aligned}\right\}.\label{eq:Achievability_D_NC}
\end{align}
The covert capacity of the \acp{DMC} with action-dependent states, depicted in Fig.~\ref{fig:System_Model}, when the \ac{ADSI} is available non-causally at the transmitter, is lower-bounded as
\begin{align*}
\textup{C}_{\mbox{\scriptsize\rm AD-NC}} \ge\mbox{\rm sup}\{R:R\in\calF_{\text{L-NC}}\}.
\end{align*}
\end{subequations}
\end{theorem}
The proof of Theorem~\ref{thm:Acievability_KG} is given in Appendix~\ref{proof:thm:Acievability_KG}, and it is based on block Markov encoding by using the \ac{ADSI} for two different purposes simultaneously: first, Gel'fand-Pinsker encoding for transmitting the message by using the \ac{ADSI} \cite{StateDepChan}, and second, Wyner-Ziv coding for generating a secret key from the \ac{ADSI} \cite{Source_Coding_SI}. Note that we generate the secret key from a description of the \ac{ADSI}; this helps not overgenerate secret keys when secret keys do not help achieve a positive covert communication rate. This also helps exploit the correlation between the \ac{ADSI} $S^N$ and the channel output $Y^N$ to generate the secret key more efficiently, by adapting a Wyner-Ziv encoding scheme. 
Our block-Markov encoding scheme is depicted in Fig.~\ref{fig:Scheme_NC} for some block $b$. As seen in this figure we use a likelihood encoder to perform the Gel'fand-Pinsker encoding and to compute a description of the \ac{ADSI} which will be used for the secret key generation in the next block. After decoding the reconciliation information of the description of the \ac{ADSI} of the previous block, the decoder uses a Wyner-Ziv decoder to reconstruct the description of the \ac{ADSI} of the previous block, therefore, will be able to generate a secret key with the transmitter. 
We also note that the cardinality bounds in this paper are based on the standard techniques in \cite{ElGamalKim}.
\begin{remark}[Optimal Distributions]
\label{rem:Optimal_Dists}
    Since in the mutual information expressions in Theorem~\ref{thm:Acievability_KG}, the auxiliary random variable $U$ appears either alongside $A$ or conditioned on it, we may replace $(A,U)$ with $\tilde{U}\triangleq(A,U)$ without altering the achievable rate or the underlying joint distribution.
\end{remark}
To gain insight into the structure of the achievable rate in Theorem~\ref{thm:Acievability_KG}, notice that $U$ in Theorem~\ref{thm:Acievability_KG} corresponds to the auxiliary \ac{RV} for the Gel'fand-Pinsker encoding and represents the message, and the auxiliary \ac{RV} $V$ is a description of the \ac{ADSI} which is only used for key generation. 
\begin{figure*}[t]
\centering
\includegraphics[width=17cm]{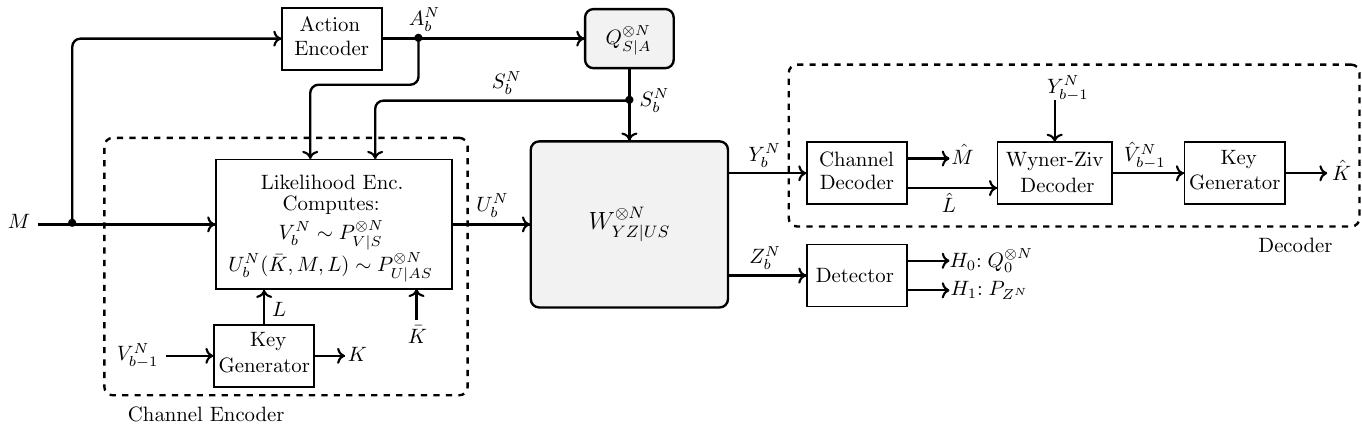}
\caption{The encoding and decoding scheme for block $b$, assuming that the encoder and the decoder have generated a secret key $\bar{K}$ in the previous block. The encoder first generates the reconciliation index $L$ and the secret key $K$ from a description of the \ac{ADSI}, i.e., $V_{b-1}^N$, generated in the previous block. To transmit the message $M$ and the reconciliation index $L$, the likelihood encoder performs Gel'fand-Pinsker encoding. It computes the sequence $U_b^N$ according to the indices $(\bar{K},M,L)$ such that $U_b^N$ is \ac{iid} with the \ac{ADSI}. Also, the likelihood encoder computes the sequence $V_b^N$ that is \ac{iid} with the \ac{ADSI}, which will be used in the next block for the secret key generation. By decoding the indices $M$ and $L$ and accessing $Y_{b-1}^N$ the decoder will be able to reconstruct the sequence $V_{b-1}^N$ by using a Wyner-Ziv decoder, and therefore it can generate the secret key $K$.}
\label{fig:Scheme_NC}
\end{figure*}
We now provide interpretations of the achievable rate in Theorem~\ref{thm:Acievability_KG}. The rate constraints in \eqref{eq:Achievability_A_NC} correspond to two rates that our code design can achieve:
$r_\cc\triangleq\bbI(A,U;Y)-\bbI(U;S|A)$ is the total communication rate. Since the auxiliary \ac{RV} $U$ represents the message, and the auxiliary \ac{RV} $V$ is only used for key generation from the \ac{ADSI}, and since the sufficient condition for transmitting the message and reconstructing $V$ is $r_{\tck}\triangleq\bbI(A,U;Y)+\bbI(V;Y|A,U)-\bbI(U,V;S|A)\mathop=\bbI(A,U,V;Y)-\bbI(U,V;S|A)$ we interpret $r_{\tck}$ as the total communication rate and key rate (secured and unsecured) \cite[Theorem~1]{Weissman10}. 
Moreover, since the \ac{RV} $V$ is used for the secret key generation,  $\bar{r}_{\scn}\triangleq\bbI(V;Y|A,U)-\bbI(V;Z)$ corresponding to inequality $(d)$ in \eqref{eq:Achievability_D_NC} is the secret key rate that the key generation codebook can produce. Therefore, the condition $\bar{r}_{\text{sc}}>0$ in \eqref{eq:Achievability_D_NC} means that the secret key generation over the channel $W_{YZ|XS}$ with action-dependent states $S$ should be feasible. The inequality $(e)$ in \eqref{eq:Achievability_D_NC} corresponds to $r_{\tcsk}\triangleq\bbI(A,U,V;Y)-\bbI(A,U,V;Z)\ge0$, where $\bbI(A,U,V;Y)-\bbI(A,U,V;Z)$ is the total covert communication and secret key rate, and the inequality $(f)$ corresponds to $0\le\bbI(A,U,V;Y) - \bbI(A,V;Z) - \bbI(U;S|A) 
= \bbI(A,U;Y) + \bbI(V;Y|U,A) - \bbI(A,V;Z) - \bbI(U;S|A) =r_\cc+r_{\scn}$, where $r_{\scn}\triangleq\bbI(V;Y|A,U)-\bbI(A,V;Z)$ is the secret key rate generated by our code design. Note that, since the action $A$ carries information about both the message and the \ac{ADSI} it may leak information about the secret key, and therefore the secret key rate that our code design, which includes all the codebooks, can produce is less than the secret key rate that our secret key generation codebook alone can produce, i.e., $r_{\scn}\le\bar{r}_{\scn}$.   
The inequality $(g)$ in \eqref{eq:Achievability_D_NC} corresponds to $r_{\tck}+r_{\scn}\ge0$. 
Also, note that the condition $(c)$ in \eqref{eq:Achievability_D_NC} is because the test channel $Q_{S|A}$ is fixed and chosen by nature. 
Theorem~\ref{thm:Acievability_KG} subsumes some existing results in the literature, as elaborated by the following two remarks.
\begin{remark}[Comparison with Channels with Action-Dependent States Without Covert Constraint]
Setting $Z=\emptyset$ and $V=\emptyset$, and removing the covert constraint $P_Z=Q_0$, the achievable rate in Theorem~\ref{thm:Acievability_KG} recovers the capacity of the channels with action-dependent states in \cite[Theorem~1]{Weissman10}, which is $r_c=I(U;Y)-I(U;S|A)=I(A,U;Y)-I(U;S|A)$.
\end{remark}
\begin{remark}[Comparison with Covert Communications Over Channels with States]
By setting $A=\emptyset$, $Q_{S|A}=Q_S$, and $P_{U|A}=P_U$ the lower bound in Theorem~\ref{thm:Acievability_KG} recovers the lower bound on the covert capacity of state-dependent channels in \cite[Theorem~4]{Keyless22}. 
Note that since the main idea to achieve a positive covert rate in both the state-dependent channels and in the channels with action-dependent states is to ``hide" information in the state; therefore the more random the state is the better. In the state-dependent channels, the \ac{CSI} is controlled by nature whereas the \ac{ADSI} can be partially controlled by the transmitter, therefore we expect to achieve a higher covert rate in channels with action-dependent states. Also, the \ac{ADSI} carries information about the message, this provides an additional resource for the receiver to decode the transmitted message reliably. This can be seen in the expression of Theorem~\ref{thm:Acievability_KG}. 
\end{remark}
\begin{remark}[The Effect of the Secret Shared Key of Negligible Rate]
    To illustrate the effect of the secret key of negligible rate we use non-strict inequalities for the rate constraints that will be affected by a shared secret key between the legitimate terminals in \eqref{eq:Achievability_D_NC}, i.e., the second and the third inequalities. The effect of the secret shared key can also be seen in state-dependent channels \cite[Theorem~2 and Theorem~4]{LeeWang18}.
\end{remark}

When we only use the \ac{ADSI} for the message transmission and not for the key generation, the following covert rate can be achieved by choosing $V=\emptyset$ in Theorem~\ref{thm:Acievability_KG} and considering Remark~\ref{rem:Optimal_Dists}, which will be used to study the Gaussian channels and channels with rewrite option in Section~\ref{sec:rewrite} and Section~\ref{sec:Gaussian}, respectively.
\begin{corollary}
\label{cor:Ach_NC_Simple}
Let
\begin{subequations}
\begin{align}
  \calF_{\text{L-NC}}^* = \left.\begin{cases}R\geq 0: \exists P_{ASUXYZ}\in\calG_{\text{L-NC}}^*:\\
  R< \bbI(U;Y)-\bbI(U;S|A)\\
\end{cases}\right\},\label{eq:Ach_NC_Simple_A}
\end{align}
where
\begin{align}
  \calG_{\text{L-NC}}^*\triangleq \left.\begin{cases}P_{ASUXYZ}:\\
P_{ASUXYZ}=P_AQ_{S|A}P_{U|SA}P_{X|US}W_{YZ|XS}\\
\bbI(U;Y)\ge\bbI(U;Z)\\
\bbI(U;Y)\ge\bbI(A;Z) + \bbI(U;S|A)\\
P_Z=Q_0\\
\card{\calU}\leq\card{\calA}\card{\calS}\card{\calX}+2
\end{cases}\right\}.\label{eq:Ach_NC_Simple_D}
\end{align}
The covert capacity of the \acp{DMC} with action-dependent states, depicted in Fig.~\ref{fig:System_Model}, when the \ac{ADSI} is available non-causally at the transmitter, is lower-bounded as
\begin{align*}
\textup{C}_{\mbox{\scriptsize\rm AD-NC}} \ge\mbox{\rm sup}\{R:R\in\calF_{\text{L-NC}}^*\}.
\end{align*}
\end{subequations}
\end{corollary}
We note that since $U$ and $S$ are correlated in Theorem~\ref{thm:Acievability_KG}, we can treat $(A,U,V)$ as a single \ac{RV} in the constraint $\bbI(A,U,V;Y)\ge\bbI(A,U,V;Z)$ and therefore the constraint $\bbI(A,U,V;Y)\ge\bbI(A,U,V;Z)$ in Theorem~\ref{thm:Acievability_KG} is equivalent to the constraint $\bbI(U;Y)\ge\bbI(U;Z)$ in Corollary~\ref{cor:Ach_NC_Simple}. However, by choosing a proper realization for the auxiliary \ac{RV} $V$, the constraints $\bbI(A,U,V;Y)\ge\bbI(A,V;Z)+\bbI(U;S|A)$, $\bbI(V;Y|A,U)>\bbI(V;Z)$, and $\bbI(A,U,V;Y)+\bbI(V;Y|A,U)\ge\bbI(A,V;Z)+\bbI(U,V;S|A)$ in Theorem~\ref{thm:Acievability_KG} may relax the constraint $\bbI(U;Y)\ge\bbI(A;Z)+\bbI(U;S|A)$ in Corollary~\ref{cor:Ach_NC_Simple}. As a result, Theorem~\ref{thm:Acievability_KG} may lead to a higher covert rate compared to Corollary~\ref{cor:Ach_NC_Simple}.

We now provide an upper bound on the covert capacity when the \ac{ADSI} is available non-causally at the encoder. 
\begin{theorem}
\label{thm:Converse_NC}
Let
\begin{subequations}\label{eq:Converse_AD_NC}
\begin{align}
  \calF_{\text{U-NC}} = \left.\begin{cases}R\geq 0: \exists P_{ASUVXYZ}\in\calG_{\text{U-NC}}:\\
  R\le\bbI(U;Y) - \bbI(U;S|A)\\
  R\le\bbI(U,V;Y)-\bbI(U;S|V,A)\\
  R\le\bbI(V;A,U,Y)\\
\end{cases}\right\},\label{eq:Converse_A_NC}
\end{align}
where
\begin{align}
  &\calG_{\text{U-NC}}\triangleq\left.\begin{cases}P_{ASUVXYZ}:\\
P_{ASUVXYZ}=P_AQ_{S|A}P_{U|AS} P_{V|US}P_{X|US}W_{YZ|XS}\\
\bbI(U;Y) - \bbI(U;S|A)\ge\bbI(V;Z)-\bbI(V;S)\\
\bbI(U,V;Y)-\bbI(U;S|V,A)
\ge\bbI(V;Z)-\bbI(V;S)\\
\bbI(V;A,U,Y)\ge\bbI(V;Z)-\bbI(V;S)\\
P_Z=Q_0\\
\card{\calU}\leq\card{\calA}\card{\calS}\card{\calX}+4\\
\card{\calV}\leq\left(\card{\calA}\card{\calS}\card{\calX}+4\right)^2
\end{cases}\right\}.\label{eq:Converse_D_NC}
\end{align}
The covert capacity of the \acp{DMC} with action-dependent states, depicted in Fig.~\ref{fig:System_Model}, when the \ac{ADSI} is available non-causally at the transmitter, is upper-bounded as
\begin{align}
\textup{C}_{\mbox{\scriptsize\rm AD-NC}} \le\mbox{\rm sup}\{R:R\in\calF_{\text{U-NC}}\}.
\label{eq:Converse_NC}
\end{align}
\end{subequations}
\end{theorem}The proof of Theorem~\ref{thm:Converse_NC} is available in Appendix~\ref{proof:thm:Converse_NC}. 
We note that the underlying joint distribution of the involved \acp{RV} in Theorem~\ref{thm:Converse_NC} is more general compared to that in Theorem~\ref{thm:Acievability_KG}, since the auxiliary \ac{RV} $V$ in Theorem~\ref{thm:Converse_NC} is generated according to the \acp{RV} $U$ and $S$, whereas in Theorem~\ref{thm:Acievability_KG} this \ac{RV} is only correlated with the \ac{ADSI}~$S$.
\begin{remark}
When the transmitter and the receiver share a secret key of sufficient rate the conditions $\bbI(U;Y)\ge\max\big\{\bbI(U;Z),\bbI(A;Z) + \bbI(U;S|A)\big\}$ in the achievable rate in Corollary~\ref{cor:Ach_NC_Simple} will be satisfied and therefore this achievable rate meet a looser version of the upper bound in Theorem~\ref{thm:Converse_NC}.
\end{remark}
\begin{figure*}[t]
\centering
\includegraphics[width=17cm]{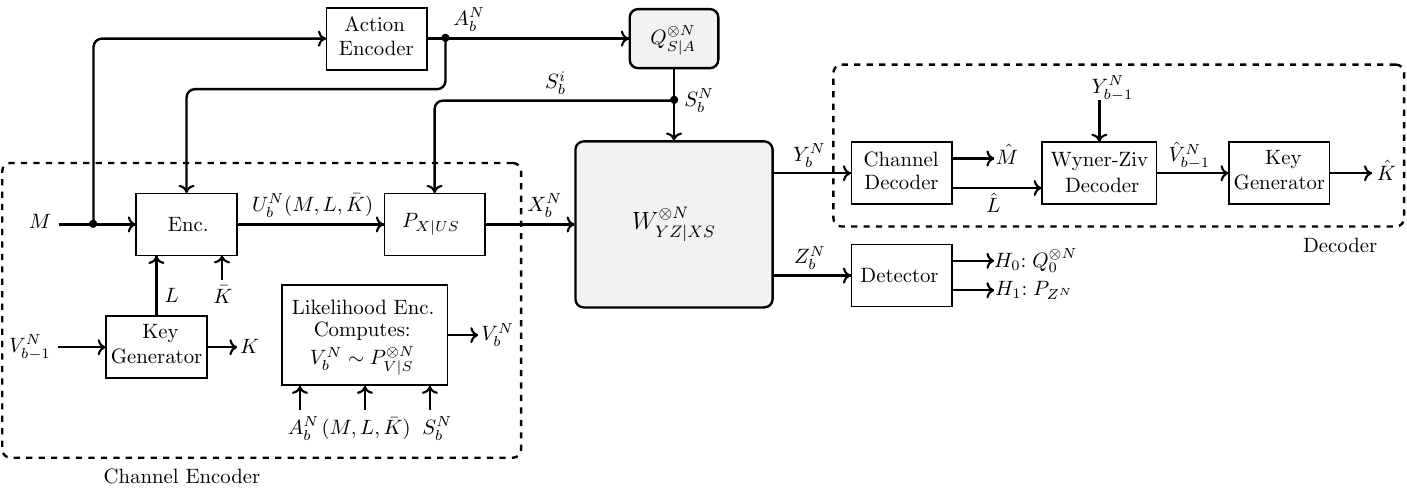}
\caption{The encoding and decoding scheme for block $b$, assuming that the encoder and the decoder have generated a secret key $\bar{K}$ in the previous block. The encoder first generates the reconciliation index $L$ and the secret key $K$ from a description of the \ac{ADSI}, i.e., $V_{b-1}^N$, generated in the previous block. To transmit the message $M$ and the reconciliation index $L$, according to the secret key $\bar{K}$, the encoder computes $U^N_b(M,L,\bar{K})$ and transmits $X_b^N$, where $X_{b,i}$ is generated by passing $U_{b,i}$ and $S_{b,i}$ through the test channel $P_{X|US}$. At the end of block $b$ accessing the action $A_b^N$, the \ac{ADSI} $S_b^N$, and the indices $(M,L,\bar{K})$ the likelihood encoder computes the sequence $V_b^N$ that is \ac{iid} with the \ac{ADSI} $S_b^N$, which will be used in the next block for the secret key generation.  
By decoding the indices $M$ and $L$ and accessing $Y_{b-1}^N$ the decoder will be able to reconstruct the sequence $V_{b-1}^N$ by using a Wyner-Ziv decoder, and therefore it can generate the secret key $K$.}
\label{fig:Scheme_C}
\end{figure*}
\subsection{\texorpdfstring{\ac{ADSI}}{ADSI} Available Causally at the Transmitter}
\label{sec:C}
We now present our results for the case of causal \ac{ADSI} at the encoder and provide some numerical examples.
\begin{theorem}
\label{thm:Acievability_KG_C}
Let
\begin{subequations}\label{eq:Achievability_AD_C}
\begin{align}
  \calF_{\text{L-C}}=\left\{
    \begin{aligned}
&R\geq 0: \exists P_{ASUVXYZ}\in\calG_{\text{L-C}}:\\
  &R< \bbI(A,U;Y) & \quad\seteqnum[a]\\
  &R< \bbI(A,U,V;Y) -  \bbI(V;S|A)& \quad\seteqnum[b]
\end{aligned}
\right\},\label{eq:Achievability_A_C}
\end{align}
where 
\begin{align}
  \calG_{\text{L-C}}\triangleq\left\{
    \begin{aligned}
&P_{ASUVXYZ}:&\quad\\
&P_{ASUVXYZ}=P_AP_{U|A}P_VQ_{S|AV}P_{X|US}W_{YZ|XS}&\quad\\
&Q_{S|A}(\cdot|a)=\sum_{v\in\calV}P_V(v)Q_{S|AV}(s|a,v)&\quad\seteqnum[c]\\
&\bbI(V;Y|A,U)>\bbI(V;Z) & \quad\seteqnum[d]\\
&\bbI(A,U,V;Y)\ge\bbI(A,U,V;Z)& \quad\seteqnum[e]\\
&\bbI(A,U,V;Y) + \bbI(V;Y|A,U)\ge\bbI(A,V;Z) + \bbI(V;S|A)& \quad\seteqnum[f]\\
&P_Z=Q_0& \quad\seteqnum[g]\\
&\card{\calU}\leq\card{\calA}\card{\calS}\card{\calX}+3& \quad\seteqnum[h]\\
&\card{\calV}\leq\left(\card{\calA}\card{\calS}\card{\calX}+3\right)\left(\card{\calA}\card{\calS}\card{\calX}+4\right)& \quad\seteqnum[i]\\
\end{aligned}
\right\}\label{eq:Achievability_D_C}
\end{align}
The covert capacity of the \acp{DMC} with action-dependent states, depicted in Fig.~\ref{fig:System_Model}, when the \ac{ADSI} is available causally at the transmitter, is lower-bounded as
\begin{align}
\textup{C}_{\mbox{\scriptsize\rm AD-C}} \ge\mbox{\rm sup}\{R:R\in\calF_{\text{L-C}}\}.
\label{eq:Achievability_C}
\end{align}
\end{subequations}
\end{theorem}
The proof of Theorem~\ref{thm:Acievability_KG_C} uses a coding scheme similar to that of Theorem~\ref{thm:Acievability_KG}, the main difference is that instead of Gel'fand-Pinsker type of encoding for transmitting the message according to the \ac{ADSI} we use a Shannon strategy \cite{StateDepChan,Weissman10}. Theorem~\ref{thm:Acievability_KG_C} is proved in Appendix~\ref{proof_Acive_Causal}. To be more specific, the proof of Theorem~\ref{thm:Acievability_KG_C} is based on block-Markov encoding scheme as depicted in Fig.~\ref{fig:Scheme_C} for some block $b$. As seen in this figure, we generate the codebook independently of the \ac{ADSI} and we then generate the channel input based on the codeword and the \ac{ADSI}. We also use a likelihood encoder to compute a description of the \ac{ADSI}, which will be used for the secret key generation in the next block. After decoding the reconciliation information of the description of the \ac{ADSI} of the previous block, the decoder uses a Wyner-Ziv decoder to reconstruct the description of the \ac{ADSI} of the previous block, therefore, will be able to generate a secret key with the transmitter.  

Note that the difference between  Theorem~\ref{thm:Acievability_KG_C} and Theorem~\ref{thm:Acievability_KG} is that the supremum for the achievable rates in Theorem~\ref{thm:Acievability_KG} is over a larger set of distributions. Specifically, for the causal case, $U-A-S$ forms a Markov chain, and therefore $\bbI(U;S|A)=0$. 
If we choose the \ac{ADSI} to be independent of the \ac{RV} $U$ given the action $A$, i.e., $Q_{S|AUV}=Q_{S|AV}$, the achievable rate in Theorem~\ref{thm:Acievability_KG} recovers Theorem~\ref{thm:Acievability_KG_C}. Additionally, it is worth noting that the same observation as outlined in Remark~\ref{rem:Optimal_Dists} applies to Theorem~\ref{thm:Acievability_KG_C}.

Theorem~\ref{thm:Acievability_KG_C} subsumes some existing results in the literature, as elaborated by the following two remarks.
\begin{remark}[Comparison with Channels with Action-Dependent States Without Covert Constraint]
Setting $Z=\emptyset$ and $V=\emptyset$, and removing the covert constraint $P_Z=Q_0$, the achievable rate in Theorem~\ref{thm:Acievability_KG_C} recovers the capacity of the channels with action-dependent states in \cite[Theorem~2]{Weissman10}.
\end{remark}
\begin{remark}[Comparison with Covert Communications Over Channels with States]
Setting $A=\emptyset$, $Q_{S|A}=Q_S$, and $P_{U|A}=P_U$ in Theorem~\ref{thm:Acievability_KG_C} results in the achievable rate for covert communications over state-dependent channels in \cite[Theorem~7]{Keyless22}. 
\end{remark}

When we use the \ac{ADSI} only for the message transmission and not for the key generation, the following rate can be achieved by choosing $V=\emptyset$ in Theorem~\ref{thm:Acievability_KG_C} and considering Remark~\ref{rem:Optimal_Dists}. 
\begin{corollary}
\label{cor:Ach_C_Simple}
Let
\begin{subequations}
\begin{align}
  \calF_{\text{L-C}}^*= \left.\begin{cases}R\geq 0: \exists P_{ASUXYZ}\in\calG_{\text{L-C}}^*:\\
  R< \bbI(U;Y)\\
\end{cases}\right\},\label{eq:Ach_C_Simple_A}
\end{align}
where
\begin{align}
  \calG_{\text{L-C}}^*\triangleq \left.\begin{cases}P_{ASUXYZ}:\\
P_{ASUXYZ}=P_AP_{U|A}Q_{S|A}P_{X|US}W_{YZ|XS}\\
\bbI(U;Y)\ge\bbI(U;Z)\\
P_Z=Q_0\\
\card{\calU}\leq\card{\calA}\card{\calS}\card{\calX}+1
\end{cases}\right\}.\label{eq:Ach_C_Simple_D}
\end{align}
\end{subequations}
The covert capacity of the \acp{DMC} with action-dependent states, depicted in Fig.~\ref{fig:System_Model}, when the \ac{ADSI} is available causally at the encoder, is lower-bounded as
\begin{align}
\textup{C}_{\mbox{\scriptsize\rm AD-C}} \ge\mbox{\rm sup}\{R:R\in\calF_{\text{L-C}}^*\}.
\label{eq:Simple_Achi_C}
\end{align}
\end{corollary} 

Similarly as before, Corollary~\ref{cor:Ach_C_Simple} can be recovered by Corollary~\ref{cor:Ach_NC_Simple}, by considering Remark~\ref{rem:Optimal_Dists} and since $\bbI(U;S|A)=0$. Corollary~\ref{cor:Ach_C_Simple} is used to study the channels with a rewrite option in Section~\ref{sec:rewrite}. 

Note that the feasible set of joint distributions in Corollary~\ref{cor:Ach_C_Simple} is a subset of the feasible set of joint distributions in Theorem~\ref{thm:Acievability_KG_C}. This is because the condition $\bbI(U;Y)\ge\bbI(U;Z)$ in \eqref{eq:Ach_C_Simple_D} is only satisfied when the legitimate receiver's channel is less noisy \ac{wrt} the warden's channel; whereas, intuitively, since the dominant condition in \eqref{eq:Achievability_D_C}, considering Remark~\ref{rem:Optimal_Dists}, is $\bbI(U,V;Y)\ge\bbI(U,V;Z)$, which might be satisfied even when the legitimate receiver's channel is \textit{not} less noisy \ac{wrt} the warden's channel; since choosing the auxiliary \ac{RV} $V$, which is correlated with the \ac{ADSI} $S$, provides some flexibility for the encoder to satisfy the constraints on the feasible set of joint distributions. When Corollary~\ref{cor:Ach_C_Simple} leads to a positive covert rate for some channels $W_{YZ|XS}$, which means that there is at least one joint distribution $P_{ASUXYZ}$ for which the constraint in \eqref{eq:Ach_C_Simple_D} is satisfied, by setting $V=\emptyset$ Theorem~\ref{thm:Acievability_KG_C} can achieve the same covert rate as Corollary~\ref{cor:Ach_C_Simple}. Setting $V=\emptyset$ means we do not use the \ac{ADSI} for the secret key generation. On the other hand, when the warden's channel is less noisy \ac{wrt} the legitimate receiver's channel, Corollary~\ref{cor:Ach_C_Simple} does not lead to a positive covert rate; however, there might be some joint distribution $P_{ASUVXYZ}$ for which the constraints in \eqref{eq:Achievability_D_C} are satisfied and as a result Theorem~\ref{thm:Acievability_KG_C} may lead to a positive covert rate. In this case, we may need to allocate part of the rate for the secret key generation scheme. In Section~\ref{sec:Binary_DFRC} we present an example in which Corollary~\ref{cor:Ach_C_Simple} fails to achieve a positive covert rate but Theorem~\ref{thm:Acievability_KG_C} leads to achieve a positive covert rate.

We now provide an upper bound on the covert capacity when the \ac{ADSI} is available causally at the encoder.
\begin{theorem}
\label{thm:Converse_C}
Let
\begin{subequations}\label{eq:Converse_AD_C}
\begin{align}\label{eq:Converse_A_C}
  \calF_{\text{U-C}}= \left.\begin{cases}R\geq 0: \exists P_{ASUVXYZ}\in\calG_{\text{U-C}}:\\
  R\le \bbI(U;Y)\\
\end{cases}\right\},
\end{align}
where
\begin{align}\label{eq:Converse_D_C}
  \calG_{\text{U-C}}\triangleq \left.\begin{cases}P_{ASUVXYZ}:\\
P_{ASUVXYZ}=P_AQ_{S|A}P_{U|A}P_{V|US} P_{X|US}W_{YZ|XS}\\
\bbI(U;Y)\ge\bbI(V;Z)\\
P_Z=Q_0\\
\card{\calU},\card{\calV}\leq\card{\calY}\\
\end{cases}\right\}.
\end{align}
\end{subequations}
The covert capacity of the \acp{DMC} with action-dependent states, depicted in Fig.~\ref{fig:System_Model}, when the \ac{ADSI} is available causally at the encoder, is upper-bounded as
\begin{align}
\textup{C}_{\mbox{\scriptsize\rm AD-C}} \le\mbox{\rm sup}\{R:R\in\calF_{\text{U-C}}\}.
\label{eq:Converse_Causal}
\end{align}
\end{theorem}
The proof of Theorem~\ref{thm:Converse_C} is available in Appendix~\ref{proof:thm:Converse_C}. Similar to the non-causal case, the auxiliary \ac{RV} $V$ for the upper bound in Theorem~\ref{thm:Converse_C} depends on $U$ and $S$, but for the lower bound in Theorem~\ref{thm:Acievability_KG_C} this \ac{RV} is only correlated with \ac{ADSI} $S$.
\begin{remark}
When the legitimate receiver's channel is less noisy \ac{wrt} the warden's channel, i.e., $\bbI(U;Y)\ge\bbI(U;Z)$, the achievable rate in Corollary~\ref{cor:Ach_C_Simple} and a looser version of the upper in Theorem~\ref{thm:Converse_C}, by removing the condition $\bbI(U;Y)\ge\bbI(V;Z)$, meet. Also, when the transmitter and the receiver share a secret key of sufficient rate the condition $\bbI(U;Y)\ge\bbI(U;Z)$ in Corollary~\ref{cor:Ach_C_Simple} will be satisfied and Corollary~\ref{cor:Ach_C_Simple} meets a looser version of the upper bound in Theorem~\ref{thm:Converse_C}.
\end{remark}

When the warden and the legitimate receiver observe the same channel output, i.e., $Y=Z$, our lower bound in Corollary~\ref{cor:Ach_C_Simple} meets a looser version of Theorem~\ref{thm:Converse_C}, by ignoring the constraint $\bbI(U;Y)\ge\bbI(V;Z)$, and hence we obtain the covert capacity, as provided in the sequel.
\begin{corollary}
\label{cor:Y=Z_C}
Let
\begin{subequations}\label{eq:Converse_AD_C_Y=Z}
\begin{align}\label{eq:Converse_A_C_Y=Z}
  \calF_{\text{O-C}} = \left.\begin{cases}R\geq 0: \exists P_{ASUXY}\in\calG_{\text{O-C}}:\\
  R\le \bbI(U;Y)\\
\end{cases}\right\},
\end{align}
where
\begin{align}\label{eq:Converse_D_C_Y=Z}
  \calG_{\text{O-C}}\triangleq \left.\begin{cases}P_{ASUXY}:\\
P_{ASUXY}=P_AP_{U|A}Q_{S|A}P_{X|US}W_{Y|XS}\\
P_Y=Q_0\\
\card{\calU}\leq\card{\calY}
\end{cases}\right\}.
\end{align}
\end{subequations}
The covert capacity of the \acp{DMC} with action-dependent states, depicted in Fig.~\ref{fig:System_Model}, when the \ac{ADSI} is available causally at the encoder and $Y=Z$, is
\begin{align}
\textup{C}_{\mbox{\scriptsize\rm AD-C}} =\mbox{\rm max}\{R:R\in\calF_{\text{O-C}}\}.
\label{eq:Capacity_Causal_Y=Z}
\end{align}
\end{corollary}
We note that from Remark~\ref{remark:Causal_General_Converse}, the covert capacity in Corollary~\ref{cor:Y=Z_C} can also be applied to the general channels of the form $W_{YZ|ASX}$. Also note that our inner and outer bounds in Theorem~\ref{thm:Acievability_KG} and Theorem~\ref{thm:Converse_NC} do not meet when the \ac{ADSI} is available non-causally at the encoder and $Y=Z$. 
\subsection{Intuitions and Numerical Examples}
In this section, we begin with providing some intuition on why the knowledge of the \ac{ADSI} helps go beyond the square root law and then provide an example to demonstrate the advantages of our scheme with secret key generation from the \ac{ADSI}, followed by another example showing that the \ac{ADSI} results in a higher covert rate compared to state-dependent channels. 
\subsubsection{Why does ADSI help go beyond the square root law?}
In all of our achievability results discussed above, the constraints over the joint distributions could be circumvented by an external secret key shared between the legitimate terminals, except for the covertness constraint $P_Z=Q_0$. 
This section provides some insights into why it is possible to go beyond the square root law when the \ac{ADSI} is available at the transmitter. For simplicity, we ignore the role of action input here and consider the state-dependent channels, but our argument can be extended to the channels with \ac{ADSI}. Consider a \ac{DMC} $W_{Z|X,S}$, when the \ac{CSI} is {\em not} available at the transmitter the distribution induced at the output of the warden for the no-communication mode is
\begin{align}
    Q_0&\triangleq W_{Z|X=x_0}\nonumber\\
    &=\sum_{j\in[\abs{\calS}]}Q_S(s_j)W_{Z|S,X=x_0}(\cdot|s_j,x_0).\label{eq:q0_average}
\end{align}Also, the distribution induced at the output of the warden by all the channel input symbols except for the innocent symbol $x_0$ is
\begin{align}
    P_Z&\triangleq \sum_{\substack{i\in[\abs{\calX}]}}P_X(x_i)W_{Z|X}(\cdot|x_i),\label{eq:pz_p2p}
\end{align}where $W_{Z|X}(\cdot|x_i)=\sum_{j\in[\abs{\calS}]}Q_S(s_j)W_{Z|S,X}(\cdot|s_j,x_i)$, for $x_i\in\calX$, and $P_X(x_0)=0$. Therefore, $P_Z$ is in the convex hull of the set $\calH\triangleq\left\{W_{Z|X=x_i}\right\}_{i\in[|\calX|],i\ne0}$. It is possible to achieve a positive covert rate for the channel $W_{Z|X}$ when $Q_0=P_Z$, which means $Q_0$ should be in the convex hull of the set $\calH$. 

When the \ac{CSI} is available non-causally at the transmitter, the distribution induced at the output of the warden by all the channel input symbols except for the innocent symbol $x_0$ is
\begin{align}
    \bar{P}_Z&\triangleq \sum_{\substack{i\in[\abs{\calX}]}}\sum_{j\in[\abs{\calS}]}Q_S(s_j)P_{X|S}(x_i|s_j)W_{Z|SX}(\cdot|s_j,x_i),\label{eq:pz_csi}
\end{align}where $P_{X|S}(x_0|s_j)=0$ for all $j\in[\abs{\calS}]$. Therefore $\bar{P}_Z$ is in the convex hull of the set $\bar{\calH}\triangleq\left\{W_{Z|X=x_i,S=s_j}\right\}_{i\in[|\calX|],i\ne0,j\in[|\calS|]}$ and the covertness constraint $\bar{P}_Z=Q_0$ means $Q_0$ must be in the convex hull of the set $\bar{\calH}$. 
Also, since the \acp{RV} $X$ and $S$ in \eqref{eq:pz_csi} are correlated, we can rewrite \eqref{eq:pz_csi} as
\begin{align}
    \bar{P}_Z
    &=\sum_{\substack{i\in[\abs{\calX}]}}\sum_{j\in[\abs{\calS}]}P_X(x_i)P_{S|X}(s_j|x_i)W_{Z|SX}(\cdot|s_j,x_i),\label{eq:pz_csi2}
\end{align}where $P_{S|X}(x_i|s_j)=\frac{Q_S(s_j)P_{X|S}(x_i|s_j)}{\sum_{i\in[\abs{\calX}-1]}Q_S(s_j)P_{X|S}(x_i|s_j)}$. Therefore, for each $x_i\in\calX$, $i\in[|\calX|]$ and $x_i\ne x_0$, we have 
\begin{align}
    W_{Z|X=x_i}(\cdot|x_i)
    &=\sum_{j\in[\abs{\calS}]}P_{S|X}(s_j|x_i)W_{Z|SX}(\cdot|s_j,x_i).\label{eq:pz_csi3}
\end{align}This shows that $W_{Z|X=x_i}$, for $i\in[|\calX|]$ and $i\ne0$, is in the convex hull of the set $\bar{\calH}$. Therefore, the set $\calH$ is in the convex hull of the set $\bar{\calH}$, as a result, $Q_0$ might be in the convex hull of the set $\bar{\calH}$, but not the set $\calH$, potentially allowing for a positive covert rate for a larger class of channels. Similarly, one can show that the set $\bar{\calH}$ is a subset of the set of distributions induced at the output of the warden in channels with \ac{ADSI}.

\subsubsection{Example (Theorem~\ref{thm:Acievability_KG_C} improves the simple achievable rate in Corollary~\ref{cor:Ach_C_Simple})}
\label{sec:Reversely_degraded}
To show the advantages of the key generation scheme used in this paper we provide an example in which the simple scheme in Corollary~\ref{cor:Ach_C_Simple} does not lead to a positive covert rate while the secret key generation scheme in  Theorem~\ref{thm:Acievability_KG_C} leads to a positive covert rate. Consider a channel in which the warden's and the receiver's channels are \acp{BSCO}, as illustrated in Fig~\ref{fig:Example}. Let the \ac{ADSI} be $S=(S_Y,S_Z)$, where $S_Y$ is a Bernoulli \ac{RV} with parameter $\frac{1}{2}$ and independent of the action $A$, and $S_Z=A\oplus N_Z$ with $A$ as the action, which is a Bernoulli \ac{RV} with parameter $\alpha$, and $N_Z$ is a Bernoulli \ac{RV} with parameter $\lambda$. We assume that $A$, $N_Z$, $S_Y$ are mutually independent, $\oplus$ is the modulo two addition, and the innocent symbols are $a_0=x_0=0$. When $S_Z=0$ the warden's channel is a \ac{BSCO} with parameter $\epsilon_1$ and when $S_Z=1$ the warden's channel is another \ac{BSCO} with parameter $\epsilon_2$. Since the warden does not know the \ac{ADSI} $S_Z$ its channel is equivalent to a convex combination of these two \acp{BSCO} which is another \ac{BSCO} with some parameter $\epsilon^*$, as seen in Fig.~\ref{fig:Example}. Similarly, when $S_Y=0$ the receiver's channel is a \ac{BSCO} with parameter $\epsilon_3$ and when $S_Z=1$ the receiver's channel is another \ac{BSCO} with parameter $\epsilon_4$, as seen in Fig.~\ref{fig:Example}. We assume that $\epsilon_1<\epsilon_2<\epsilon_3<\epsilon_4\le\frac{1}{2}$ and therefore $\epsilon^*<\epsilon_3<\epsilon_4\le\frac{1}{2}$. Also, we assume that the receiver knows $S_Y$, i.e., $Y=(\tilde{Y},S_Y)$, where $\tilde{Y}$ is the output of the \ac{BSCO} corresponding to $S_Y$. 
Note that since $\epsilon^*<\epsilon_3<\epsilon_4\le\frac{1}{2}$, similar to the \acp{BSC}, one can show that the legitimate receiver's channel, which is a \ac{BSCO} with parameter $\epsilon_{3+S_Y}$, is degraded \ac{wrt} the warden's channel, which is a \ac{BSCO} with some parameter $\epsilon^*$. 
Therefore, the region in Corrolary~\ref{cor:Ach_C_Simple} does not lead to a positive covert rate since the condition $\bbI(U;Y)\ge\bbI(U;Z)$ is never satisfied. We now show that the region in Theorem~\ref{thm:Acievability_KG_C} leads to a positive covert rate. 
We first note that when the transmitter is not communicating with the receiver, that is the case when it is transmitting the innocent symbols $a_0=x_0=0$, the distribution observed by the warden is the uniform distribution over $\{1,-1\}$. One can show that the transmitter can induce the uniform distribution at the warden's channel output by choosing the channel input to be uniform over $\{1,-1\}$. Therefore the covert constraint $P_Z=Q_0$ is satisfied when the channel input $X$ is a uniform Bernoulli \ac{RV} over $\{1,-1\}$. 
We choose $V=S_Y$ and $U$ as a Bernoulli \ac{RV} with parameter $\frac{1}{2}$ over the set $\{1,-1\}$ independent of the action $A$ and the channel input to be $X=U$. 
Now we show that the constraints in \eqref{eq:Achievability_D_C} are satisfied. We have,
\begin{figure*}[t]
\centering
\includegraphics[width=13cm]{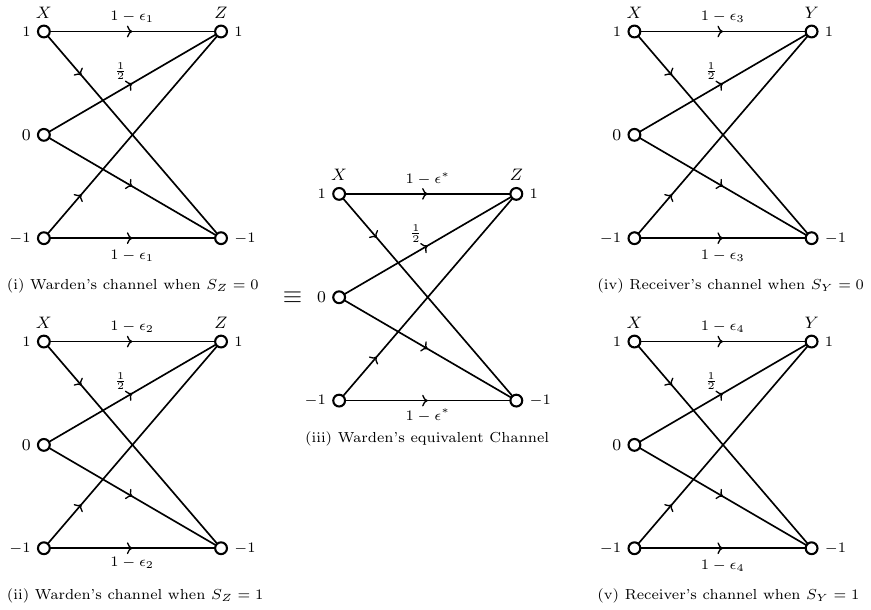}
\caption{\acp{BSCO} with \ac{ADSI}}
\label{fig:Example}
\end{figure*}
\begin{subequations}\label{eq:MI_Calculation}
    \begin{align}
        \bbI(V;Y|A,U)&=\bbI(S_Y;\tilde{Y},S_Y|A,U)\nonumber\\
        &=\bbH(S_Y|A,U)\nonumber\\
        &=\bbH(S_Y)=1,\label{eq:VYgAU}\\
        \bbI(V;Z)&=\bbI(S_Y;Z)=0,\label{eq:VZ}\\
        \bbI(V,U;Y)&=\bbI(S_Y,U;\tilde{Y},S_Y)\nonumber\\
        &=\bbI(S_Y,U;S_Y)+\bbI(U;\tilde{Y}|S_Y)\nonumber\\
        &=1+\bbI(U;\tilde{Y}|S_Y)\nonumber\\
        &=2-\bbH_b(\epsilon_{3+S_Y}),\label{eq:VUY}\\
        \bbI(V,U;Z)&=\bbI(S_Y,U;Z)=\bbI(U;Z),\label{eq:VUZ}\\
        \bbI(A,V;Z)&=\bbI(A,S_Y;Z)=\bbI(A;Z),\label{eq:AVZ}\\
        \bbI(V;S|A)&=\bbI(S_Y;S_Y,A\oplus N_Z|A)\nonumber\\
        &=\bbI(S_Y;S_Y,N_Z|A)\nonumber\\
        &=\bbH(S_Y)=1,\label{eq:VSgA}
    \end{align}where $\bbH_b(\epsilon)=-\epsilon\log(\epsilon)-(1-\epsilon)\log(1-\epsilon)$. \eqref{eq:VYgAU} and \eqref{eq:VZ} show that the condition $\bbI(V;Y|A,U)>\bbI(V;Z)$ is satisfied. Since $\bbI(U;Z)\le\bbH(U)=1$ and $\bbH_b(\epsilon_{3+S_Y})\le1$, \eqref{eq:VUY} and \eqref{eq:VUZ} show that the constraint $\bbI(V,U;Y)\ge\bbI(V,U;Z)$ is satisfied. Also, \eqref{eq:MI_Calculation} shows that the condition $\bbI(V,U;Y)+\bbI(V;Y|A,U)>\bbI(A,V;Z)+\bbI(V;S|A)$ is satisfied. We also have $\bbI(U;Y)=\bbI(U;\tilde{Y},S_Y)=\bbI(U;\tilde{Y}|S_Y)=1-\bbH_b(\epsilon_{3+S_Y})$. Therefore, considering \eqref{eq:MI_Calculation}, Theorem~\ref{thm:Acievability_KG_C} leads to achieving a positive covert rate $1-\bbH_b(\epsilon_{3+S_Y})$ bits.
\end{subequations} 
\subsubsection{Example (Achieving higher covert communication rate compared to the state-dependent channels)}
\label{sec:Binary_DFRC} 
In the state-dependent channels, the flexibility provided by the knowledge of the state $S^N$, which is generated by nature and is independent of the message $M$ \cite{LeeWang18,Keyless22}, leads to achieving a positive covert rate. However, in the channels with \ac{ADSI}, the state $S^N$ can be partially controlled by the transmitter and carries information about the message $M$. Therefore, we expect to achieve a higher covert rate in channels with \ac{ADSI} compared to the state-dependent channels. In the following, we provide an example to illustrate this intuition.

Let the channel law be,
\begin{align}
    Y&=S,\quad\quad
    Z=X\oplus S,\quad\text{and}\quad S=A\oplus N_S,\label{eq:system_Model_1}
\end{align}where $N_S$ is a Bernoulli \ac{RV} with parameter $\beta\in(0,0.5)$ and independent of the action $A$, all the \acp{RV} are binary, and the innocent symbols are $x_0=a_0=0$.

If the state $S$ does not carry information about the message, i.e., $A=\emptyset$, which corresponds to the state-dependent channels, the covert communication rate for this example is zero. Now, using the achievability result in Corollary~\ref{cor:Ach_C_Simple}, we show that it is possible to achieve a positive covert rate for channels with \ac{ADSI}. We choose the action $A$ to be a Bernoulli \ac{RV} with parameter $\alpha$ and choose the auxiliary \ac{RV} $U$ and the channel input to be equal to the action, i.e., $X=U=A$. Therefore, the warden's channel output when the transmitter is communicating with the receiver is $Z=N_S$, which is equal to the warden's channel output when the transmitter is not communicating with the receiver. Therefore, the covertness constraint is satisfied. Substituting these choices of \acp{RV} into Corollary~\ref{cor:Ach_C_Simple} leads to the communication rate $R<\dent_b(\lambda)-\dent_b(\beta)$, where $\lambda=\alpha\beta+(1-\alpha)(1-\beta)$ and $\dent_b(\lambda)=-\lambda\log(\lambda)-(1-\lambda)\log(1-\lambda)$. This is because
\begin{align*}
    \bbI(U;Y)&=\bbI(A;A\oplus N_S)\nonumber\\
        &=\bbH(A\oplus N_S)-\bbH(A\oplus N_S|A)\nonumber\\
        &=\bbH(A\oplus N_S)-\bbH(N_S)\nonumber\\
    &=\dent_b(\lambda)-\dent_b(\beta),\\
    \bbI(U;Z)&=\bbI(A;N_S)=0.
\end{align*}Note that choosing $\alpha=0.5$ leads to achieving $1-\dent_b(\beta)$ covert bits.
\subsection{More General Channels}
\label{sec:More_General_Channels}
The achievability results in Section~\ref{sec:NC} and Section~\ref{sec:C}, i.e., Theorem~\ref{thm:Acievability_KG}, Corollary~\ref{cor:Ach_NC_Simple}, Theorem~\ref{thm:Acievability_KG_C},  Corollary~\ref{cor:Ach_C_Simple}, and the optimal result in Corollary~\ref{cor:Ach_C_Simple} are applicable to general channels of the form $W_{YZ|XSA}$. In this case, all of our results provided above, except for Theorem~\ref{thm:Converse_NC} and Theorem~\ref{thm:Converse_C}, remain valid, the only change is that $X$ should be generated according to $P_{X|AUS}$ instead of $P_{X|US}$. This also follows by defining a new state $\tilde{S}=(A,S)$ and applying the characterization described above. 

As an application of this generalization, 
we show that our general problem setup can be used to study the problem of covert communication over a state-dependent \ac{MAC} channel, where the two transmitters try to transmit the same message, and states are known at one of the transmitters.  
This problem without covert constraint is motivated and studied in \cite{Cooperative_MAC}. In this problem, two transmitters aim to covertly transmit a common message $M$ to a receiver over a state-dependent channel while one of the transmitters knows the state sequence $S^N$ non-causally or causally. Denote the channel inputs of this problem by $X_1$ and $X_2$, and the distribution of the \ac{CSI} by $Q_S$, this problem can be considered as a special case of the general model depicted in Fig.~\ref{fig:System_Model} through the following associations: $A\to X_1$, $Q_{S|A}\to Q_S$, $X\to X_2$, and $W_{YZ|SXA}\to W_{YZ|SX_1X_2}$. Then Theorem~\ref{thm:Acievability_KG} and Corollary~\ref{cor:Ach_NC_Simple} in Section~\ref{sec:NC} and Theorem~\ref{thm:Acievability_KG_C},  Corollary~\ref{cor:Ach_C_Simple}, and the optimal result in Corollary~\ref{cor:Y=Z_C} in Section~\ref{sec:C} are directly applicable to this channel model. 
\section{Writing on Clean Memory}
\label{sec:rewrite}
In this section we study the problem of transmitting covert information over a channel with \textit{rewrite} option, in which the transmitter uses the channel once and observes the channel output, 
the channel output of the first stage of transmission plays the role of channel state in the second transmission stage. This problem is interpreted as the problem of writing on a computer memory while ensuring that the outcome of the writing process is indistinguishable from a computer memory that is being erased, i.e., the innocent symbol $x_0$ is being written on all its memory cells. Therefore, the distribution of a clean memory corresponds to the output distribution when the input distribution is the innocent symbol $x_0$.  
The generic framework studied in this paper can be specialized to various problems, including channels with \textit{rewrite} option. Such problems are natural to be studied in the context of the problem studied in this paper because of our two-part coding scheme.  
In \cite[Section~V]{Weissman10} it is shown that the capacity of the channels with a rewrite option is higher than the capacity of channels without a rewrite option. On the other hand, for channels without a rewrite option, the covert capacity follows the square root law. Here we show that the availability of the rewrite option leads to a positive covert rate. We refer to this problem as ``writing on clean memory," since this problem recovers as a special case, the interesting problem of recording on a computer memory while ensuring that the distribution of the outcome of the recording process is almost identical to the distribution of a clean computer memory. We consider the cases where the rewrite operation is based on noiseless and noisy observations of the output of the first round of the writing operation. We also study the problem of writing on a computer memory with defects \cite{Heegard83}, and a rewrite option. 
\subsection{Noiseless Feedback}
\label{subsec:noiseless_Feedback}
\begin{figure*}
\centering
\includegraphics[width=11.5cm]{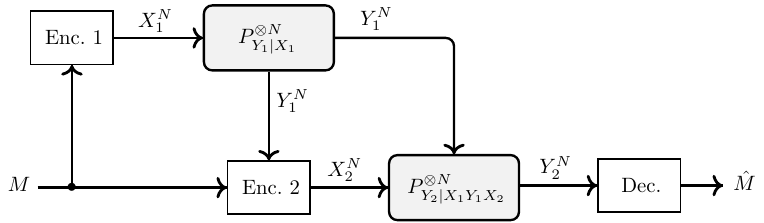}
\caption{Channels with rewrite option and noiseless feedback}
\label{fig:Writing_Memory}
\end{figure*}
Consider a \ac{DMC} characterized by $\big(\calX,\calY,W_{Y|X}\big)$, where $\calX$ and $\calY$ are the channel input and output alphabets, respectively, and $W_{Y|X}$ is the channel law. After using the channel once and observing the output of the channel in a noiseless manner, the transmitter tries to use the channel $W_{Y|X}$ one more time and rewrite on each location that it chooses such that the final output of the channel is indistinguishable from the situation in which the transmitter has transmitted the innocent symbol $x_0\in\calX$ over the channel. Therefore, when an illegitimate user observes the final channel output, it cannot decide whether the channel output is the result of transmitting information over the channel; but the legitimate user can decode the message by accessing a secret key of a negligible rate that it shares with the transmitter.

This problem is a special case of the general setup studied in Section~\ref{sec:Main_Results} if we consider
\begin{itemize}
    \item the action $A$ as the first channel input, which is denoted by $X_1$ in this section and takes values in the alphabet of the input of the channel $W_{Y|X}$, i.e., $\calX$;
    \item the \ac{ADSI} $S$ as the first channel output, which is denoted by $Y_1$ and takes values in the alphabet of the output of the channel $W_{Y|X}$, i.e., $\calY$;
    \item the channel input $X$ as the channel input of the second round, which is denoted by $X_2$ in this section and takes values in $\tilde{\calX}=\{\text{no-rewrite}\}\cup\calX$;
    \item the channel outputs $Y=Z$ as the final channel output, which are denoted by $Y_2$ in this section, and it also takes values in $\calY$;
    \item the innocent symbols $x_{1,0}=x_{2,0}=x_0\in\calX$, for the no-communication mode. Note that, for the no-communication mode, the channel input $X_2$ is always equal to $x_0$ regardless of the channel output $Y_1$, therefore, the channel output $Y_2$ can be considered as a result of one round of transmitting $x_0$.
\end{itemize} 
Therefore, the conditional distribution $Q_{S|A}$ in our general framework is considered to be $P_{Y_1|X_1}=W_{Y|X}$ and the conditional distribution $W_{Y|ASX}$ is considered to be $P_{Y_2|X_1Y_1X_2}$, where
\begin{align}
    P_{Y_2|X_1Y_1X_2}(y_2|x_1,y_1,x_2)&=P_{Y_2|Y_1X_2}(y_2|y_1,x_2)\nonumber\\
    &=\begin{cases}
\indi{1}_{\{y_2=y_1\}},&\text{if}\Squad x_2=\text{no-rewrite}\\
W_{Y|X}(y_2|x_2),&\text{otherwise}
\end{cases}.\label{eq:NL_Rewrite}
\end{align}This problem setup is illustrated in Fig.~\ref{fig:Writing_Memory}. For the non-causal case, applying Corollary~\ref{cor:Ach_NC_Simple}, with the above framework, leads to the following achievable rate.
\begin{corollary}
\label{corr:Rewrite_NC}
Let
\begin{subequations}
\begin{align}
  \calF_{\text{NL-NC}} = \left.\begin{cases}R\geq 0: \exists P_{X_1Y_1UX_2Y_2}\in\calG_{\text{NL-NC}}:\\
  R< \bbI(X_1,U;Y_2)-\bbI(U;Y_1|X_1)\\
\end{cases}\right\},
\end{align}
where
\begin{align}\label{eq:D_Noiseless_rewrite_NC}
  \calG_{\text{NL-NC}}\triangleq \left.\begin{cases}P_{X_1Y_1UX_2Y_2}:\\
P_{X_1Y_1UX_2Y_2}(x_1,y_1,u,x_2,y_2)=P_{X_1}(x_1)W_{Y|X}(y_1|x_1)\times\\
P_{U|X_1Y_1}(u|x_1,y_1)P_{X_2|UY_1X_1}(x_2|u,y_1,x_1)P_{Y_2|Y_1X_2}(y_2|y_1,x_2)\\
\bbI(U;Y_2|X_1)\ge\bbI(U;Y_1|X_1)\\
P_{Y_2}=Q_0\\
\card{\calU}\leq\card{\calX}\card{\calY}\left(\card{\calX}+1\right)+1
\end{cases}\right\},
\end{align}and $P_{Y_2|Y_1X_2}(y_2|y_1,x_2)$ is defined in \eqref{eq:NL_Rewrite}. 
The covert capacity of the channels with a noiseless rewrite option, when the input of the rewrite operation in the second round may depend non-causally on the output of the first round, is lower-bounded as
\begin{align}
\textup{C}_{\mbox{\scriptsize\rm NC}} \ge\mbox{\rm sup}\{R:R\in\calF_{\text{NL-NC}}\}.
\label{eq:Rewrite_NC}
\end{align}
\end{subequations}
\end{corollary}

On the other hand, for the causal case, applying Corollary~\ref{cor:Y=Z_C}, with the above framework leads to the capacity result.
\begin{corollary}
\label{corr:Rewrite_C}
Let

\begin{subequations}
\begin{align}
  \calF_{\text{NL-C}} = \left.\begin{cases}R\geq 0: \exists P_{X_1Y_1UX_2Y_2}\in\calG_{\text{NL-C}}:\\
  R\le \bbI(U;Y_2)\\
\end{cases}\right\},
\end{align}
where
\begin{align}\label{eq:Binary_Noiseless_Rewrite_Example}
  \calG_{\text{NL-C}}\triangleq \left.\begin{cases}P_{X_1Y_1UX_2Y_2}:\\
P_{X_1Y_1UX_2Y_2}(x_1,y_1,u,x_2,y_2)=P_U(u)P_{X_1|U}(x_1|u)\times\\
W_{Y|X}(y_1|x_1)P_{X_2|UY_1}(x_2|u,y_1)P_{Y_2|Y_1X_2}(y_2|y_1,x_2)\\
P_{Y_2}=Q_0\\
\card{\calU}\leq\card{\calY}
\end{cases}\right\},
\end{align}and $P_{Y_2|Y_1X_2}(y_2|y_1,x_2)$ is defined in \eqref{eq:NL_Rewrite}. 
The covert capacity of the channels with a noiseless rewrite option, when the input of the rewrite operation in the second round depends causally on the output of the first round, is
\begin{align}
\textup{C}_{\mbox{\scriptsize\rm C}} =\mbox{\rm max}\{R:R\in\calF_{\text{NL-C}}\}.
\label{eq:Rewrite_C}
\end{align}
\end{subequations}
\end{corollary}
\begin{remark}[Comparison with the Results in {\cite[Section~V]{Weissman10}}]
    By removing the covertness constraint $P_Z=Q_0$ and therefore the condition $\bbI(U;Y_2|X_1)\ge\bbI(U;Y_1|X_1)$ in Corollary~\ref{corr:Rewrite_NC}, it recovers the capacity result in \cite[Equation~(57)]{Weissman10}. Moreover, by removing the covertness constraint $P_Z=Q_0$ in Corollary~\ref{corr:Rewrite_C}, it recovers the capacity result in \cite[Equation~(55)]{Weissman10}.
\end{remark}
\subsubsection{Writing on a \texorpdfstring{\ac{BSC}}{BSC} with Causal Access to the First Round of Writing}
\label{ex:BSC_Noiseless}
When the innocent symbol $x_0=0$ and the channel $W_{Y|X}$ is \ac{BSC} with parameter $\epsilon\in\sbra{0}{0.5}$, by applying Corollary~\ref{corr:Rewrite_C}, we obtain the following result. 
\begin{corollary}
\label{cor:BSC_Writing_Causal}
The covert capacity of the \ac{BSC} with parameter $\epsilon\in\sbra{0}{0.5}$ and a noiseless rewrite option, described above, when the input of the rewrite operation in the second round depends causally on the output of the first round, is given by
\begin{align}
\textup{C}_{\mbox{\scriptsize\rm C}} =\dent_b(\epsilon)-\dent_b(\epsilon^2).
\nonumber
\end{align}
\end{corollary}
\begin{proof}
We have,
\begin{align}
    \bbI(U;Y_2)&=\bbH(Y_2)-\bbH(Y_2|U)\nonumber\\
    &\mathop=\limits^{(a)}\dent_b(\epsilon)-\bbH(Y_2|U)\nonumber\\
    &\mathop\le\limits^{(b)}\dent_b(\epsilon)-\dent_b\left(\epsilon^2\right),\nonumber
\end{align}where $(a)$ follows from the covertness constraint $P_{Y_2}=Q_0$ since $Q_0(y_2=1)=W_{Y|X=x_0}(y_2=1|x_0=0)=\epsilon$ and $\dent_b(\epsilon)\triangleq-\epsilon\log(\epsilon)-(1-\epsilon)\log(1-\epsilon)$;  
and $(b)$ follows since the lowest value for the probability of $P(Y_2=1|U=u)$, when we search among the set of all the possible joint distributions, is $\epsilon^2$ \cite[Section~V-A]{Weissman10}, therefore, $\epsilon^2$ is still the lowest value for $P(Y_2=1|U=u)$ when we add the covertness constraint $P_{Y_2}=Q_0$ on the feasible set of joint distributions and $\bbH(Y_2|U=u)\ge\dent_b\left(\epsilon^2\right)$. 
Now we show that the covert rate $\dent_b(\epsilon)-\dent_b\left(\epsilon^2\right)$ is also achievable. Let $U$ be a Bernoulli \ac{RV} with parameter $\alpha$, $X_1=U$, and $X_2$ be a function of $U$ and $Y_1$ defined as follows
\begin{align}
    X_2=f(U,Y_1)\triangleq\begin{cases}
\text{no-rewrite},&\text{if}\Squad U=Y_1\\
U,&\text{otherwise}
\end{cases}.
\end{align}Under these definitions one can verify that the upper bound presented above is achievable and therefore $\textup{C}_{\mbox{\scriptsize\rm C}}=\dent_b(\epsilon)-\dent_b\left(\epsilon^2\right)$, where the optimal $\alpha$ is $\alpha=\frac{\epsilon(1-\epsilon)}{1-2\epsilon^2}$.
\end{proof}
\subsubsection{Writing on a \texorpdfstring{\ac{BSC}}{BSC} with Non-Causal Access to the First Round of Writing}
Now, by applying Corollary~\ref{corr:Rewrite_NC}, we provide a lower bound on the covert capacity when the rewrite operation depends non-causally on the output of the first round of writing. 
\begin{corollary}
\label{cor:BSC_Writing_NonCausal}
The covert capacity of the \ac{BSC} with parameter $\epsilon\in\sbra{0}{0.5}$ and a noiseless rewrite option, described above, when the input of the rewrite operation in the second round depends non-causally on the output of the first round, is lower bounded by
\begin{align}
\textup{C}_{\mbox{\scriptsize\rm C}} \ge\max\limits_{0\le\beta\le1}\left[(1-\alpha)\left[\dent_b\left(\epsilon\right)-\dent_b(\epsilon^2)\right]+\alpha(1-\beta\epsilon)\left[\dent_b\left(\frac{(1-\beta)\epsilon}{1-\epsilon\beta}\right)-\dent_b\left(\frac{(1-\beta)\epsilon^2}{1-\epsilon\beta}\right)\right]\right],
\nonumber
\end{align}such that
\begin{align}
    \alpha\left[\dent_b\big(\epsilon^2(1-\beta)+\epsilon\beta\big)-\dent_b(\epsilon)\right]+\alpha(1-\beta\epsilon)\left[\dent_b\left(\frac{(1-\beta)\epsilon^2}{1-\epsilon\beta}\right)+\dent_b\left(\frac{(1-\beta)\epsilon}{1-\epsilon\beta}\right)\right]\ge0,\label{eq:Condition_NL_Achie_Writing}
\end{align}and $\alpha=\frac{\epsilon-\epsilon^2}{1-2\epsilon^2-\beta(\epsilon-\epsilon^2)}$.
\end{corollary}
\begin{proof}
Consider the joint distribution in Corollary~\ref{corr:Rewrite_NC} with the channel input of the first round $X_1$ as a Bernoulli \ac{RV} with parameter $\alpha$, the \ac{RV} $U$ as a deterministic function of the channel input and channel output of the first round of writing as
\begin{align}
U&=\begin{cases}
0,&\text{if}\Squad Y_1=X_1\Squad\text{or}\Squad{X_1=0}\\
\text{Bernoulli}(\beta),&\text{otherwise}
\end{cases},
\end{align}
where the condition $Y_1=X_1$ accounts for reliability and the condition $X_1=0$ accounts for covertness. Also, let the channel input of the second round of writing be a deterministic function of $U,X_1$ and $Y_1$ as
\begin{align}
    X_2&=f(U,Y_1,X_1)
    =\begin{cases}
\text{no-rewrite},&\text{if}\Squad Y_1=X_1\Squad\text{or}\Squad{U=1}\\
X_1,&\text{otherwise}
\end{cases},
\end{align}and $Y_2$ be the channel output of the second round of writing. Under these definitions, one can compute
\begin{subequations}
\begin{align}
    P_{Y_2}(y_2=1)&=\epsilon^2-2\alpha\epsilon^2-\alpha\beta\epsilon+\alpha\beta\epsilon^2+\alpha,\label{eq:Y2_Probability}\\
    \bbH(Y_1|X_1)&=\dent_b(\epsilon),\label{eq:Ent_Y1_given_X1}\\
    \bbH(Y_1|U,X_1)&=(1-\alpha)\dent_b(\epsilon)+\alpha(1-\beta\epsilon)\dent_b\left(\frac{(1-\beta)\epsilon}{1-\epsilon\beta}\right),\label{eq:Ent_Y1_given_UX1}\\
    \bbH(Y_2|U,X_1)&=(1-\alpha)\dent_b(\epsilon^2)+\alpha(1-\beta\epsilon)\dent_b\left(\frac{(1-\beta)\epsilon^2}{1-\epsilon\beta}\right),\label{eq:Ent_Y2_given_UX1}\\
    \bbH(Y_2|X_1)&=(1-\alpha)\dent_b(\epsilon^2)+\alpha\dent_b\left(\epsilon^2(1-\beta)+\epsilon\beta\right).\label{eq:Ent_Y2_given_X1}
\end{align}
\end{subequations}Therefore,
\begin{subequations}
\begin{align}
    \bbI(U;Y_1|X_1)&=\bbH(Y_1|X_1)-\bbH(Y_1|U,X_1)\nonumber\\
    &=\dent_b(\epsilon)-(1-\alpha)\dent_b(\epsilon)-\alpha(1-\beta\epsilon)\dent_b\left(\frac{(1-\beta)\epsilon}{1-\epsilon\beta}\right)\nonumber\\
    &=\alpha\dent_b(\epsilon)-\alpha(1-\beta\epsilon)\dent_b\left(\frac{(1-\beta)\epsilon}{1-\epsilon\beta}\right),\label{eq:IUtoY1givenX1}\\
    \bbI(U;Y_2|X_1)&=\bbH(Y_2|X_1)-\bbH(Y_2|X_1,U)\nonumber\\
    &=(1-\alpha)\dent_b(\epsilon^2)+\alpha\dent_b(\epsilon^2(1-\beta)+\epsilon\beta)-(1-\alpha)\dent_b(\epsilon^2)+\alpha(1-\beta\epsilon)\dent_b\left(\frac{(1-\beta)\epsilon^2}{1-\epsilon\beta}\right)\nonumber\\
    &=\alpha\dent_b\big(\epsilon^2(1-\beta)+\epsilon\beta\big)+\alpha(1-\beta\epsilon)\dent_b\left(\frac{(1-\beta)\epsilon^2}{1-\epsilon\beta}\right),\label{eq:IUtoY2givenX1}\\
   \bbI(U,X_1;Y_2)&=\bbH(Y_2)-\bbH(Y_2|U,X_1)\nonumber\\
    &=\dent_b\left(\epsilon^2-2\alpha\epsilon^2-\alpha\beta\epsilon+\alpha\beta\epsilon^2+\alpha\right)-(1-\alpha)\dent_b(\epsilon^2)-\alpha(1-\beta\epsilon)\dent_b\left(\frac{(1-\beta)\epsilon^2}{1-\epsilon\beta}\right)\nonumber\\
    &\mathop=\limits^{(a)}\dent_b\left(\epsilon\right)-(1-\alpha)\dent_b(\epsilon^2)-\alpha(1-\beta\epsilon)\dent_b\left(\frac{(1-\beta)\epsilon^2}{1-\epsilon\beta}\right),\label{eq:IUX1toY2}
\end{align}where $(a)$ follows from the covertness constraint $P_Z=Q_0$ and
\begin{align}
    &Q_0(y_2=1)=W_{Y|X=x_0}(y_2=1|0)=\epsilon,\nonumber\\
    &\Rightarrow \epsilon^2-2\alpha\epsilon^2-\alpha\beta\epsilon+\alpha\beta\epsilon^2+\alpha=\epsilon,\nonumber\\
    &\Rightarrow \alpha=\frac{\epsilon-\epsilon^2}{1-2\epsilon^2-\beta(\epsilon-\epsilon^2)}.
    \label{eq:Cov_Cons_RoM_NC}
\end{align}
\end{subequations}Therefore from \eqref{eq:IUtoY1givenX1} and \eqref{eq:IUX1toY2},
\begin{align}
    \bbI(U,X_1;Y_2)-\bbI(U;Y_1|X_1)
    &=(1-\alpha)\left[\dent_b\left(\epsilon\right)-\dent_b(\epsilon^2)\right]\nonumber\\
    &\qquad+\alpha(1-\beta\epsilon)\left[\dent_b\left(\frac{(1-\beta)\epsilon}{1-\epsilon\beta}\right)-\dent_b\left(\frac{(1-\beta)\epsilon^2}{1-\epsilon\beta}\right)\right].\nonumber
\end{align}
Also the condition $\bbI(U;Y_2|X_1)\ge\bbI(U;Y_1|X_1)$ in \eqref{eq:D_Noiseless_rewrite_NC} is,
\begin{align}
\Rightarrow
&\alpha\dent_b\big(\epsilon^2(1-\beta)+\epsilon\beta\big)+\alpha(1-\beta\epsilon)\dent_b\left(\frac{(1-\beta)\epsilon^2}{1-\epsilon\beta}\right)\ge\alpha\dent_b(\epsilon)-\alpha(1-\beta\epsilon)\dent_b\left(\frac{(1-\beta)\epsilon}{1-\epsilon\beta}\right).\nonumber
\end{align}
\end{proof}
\begin{remark}[Non-Causal Access to the First Round of Writing May Increase the Covert Rate]
A straightforward calculation indicates that when $\epsilon=\frac{1}{2}$ then we have $\alpha=\frac{1}{2-\beta}$ and Corollary~\ref{cor:BSC_Writing_NonCausal} leads to $\textup{C}_{\mbox{\scriptsize\rm NC}}\geq0.200094$, where the maximum is achieved at $\beta=0.36064$, as compared to the covert capacity under the causality constraint given by Corollary~\ref{cor:BSC_Writing_Causal} is $\textup{C}_{\mbox{\scriptsize\rm C}}=0.188722$. Therefore, when the \ac{BSC} parameter is $\frac{1}{2}$, the non-causality constraint increases the covert capacity by at least 5 percent. 
\end{remark}
\subsection{Noisy Feedback}
\label{subsec:noisy_Feedback}
We now consider a more realistic scenario, where a noisy version of the output of the first round of writing is available for the rewrite operation. Consider a \ac{DMC} characterized by $\big(\calX,W_{Y|X},\calY\big)$, where $\calX$ and $\calY$ are the channel input and output alphabets, respectively, and $W_{Y|X}$ is the channel law, we refer to $W_{Y|X}$ as the forward channel. Also, consider a \ac{DMC} characterized by $\big(\calY,Q_{T|Y},\calT\big)$, where $\calT$ is the corrupted version of $\calY$, through the \ac{DMC} $Q_{T|Y}$ as the channel law, we refer to $Q_{T|Y}$ as the backward channel. The channel output of the first round of transmitting over the forward channel $W_{Y|X}$ is observed through the backward channel $Q_{T|Y}$ by the rewrite encoder. After using the forward channel once and observing the output through the backward channel, the transmitter tries to use the channel $W_{Y|X}$ one more time and rewrite on each location that it chooses such that the final output of the channel is indistinguishable from the situation in which the transmitter has transmitted the innocent symbol $x_0\in\calX$ over the channel. Therefore, when an illegitimate user observes the final channel output, it cannot decide whether the channel output is the result of writing information over the channel, but the legitimate user can decode the message by accessing a secret key of a negligible rate that it shares with the transmitter.

This problem is also a special case of the general setup studied in Section~\ref{sec:Main_Results} if we consider,
\begin{itemize}
\item the action $A$ as the first channel input, which is denoted by $X_1$ in this section and takes values in $\calX$, $X_1$ passes through the forward channel $W_{Y|X}$ and the output is denoted by $Y_1$ and takes values in the alphabet of the output of the channel $W_{Y|X}$, i.e., $\calY$; 

\item the \ac{ADSI} $S$ as the corrupted version of the first channel output $Y_1$, which is denoted by $T$ and takes values in the alphabet of the output of the channel $Q_{T|Y}$, i.e., $\calT$, therefore $X_1-Y_1-T$ forms a Markov chain; 

\item the channel input $X$ as the channel input of the second round, which is denoted by $X_2$ in this section and takes values in $\calX_2=\{\text{no-rewrite}\}\cup\calX$;

\item the channel output $Y=Z$ as the final channel output, which is denoted by $Y_2$ in this section, and it also takes values in $\calY$, therefore, the conditional distribution $Q_{S|A}$ in the original problem setup is $P_{T|X}(t|x_1)\triangleq\sum_{y_1}W_{Y|X}(y_1|x_1)Q_{T|Y}(t|y_1)$ and the conditional distribution $W_{Y|ASX}$ is considered to be $P_{Y_2|X_1TX_2}$, which is defined as
\begin{subequations}\label{eq:NL_Rewrite_Noisy}
\begin{align}
    P_{Y_2|X_1TX_2}(y_2|x_1,t,x_2)&\triangleq\begin{cases}
P_{Y|XT}(y_2|x_1,t),&\text{if}\Squad x_2=\text{no-rewrite}\\
W_{Y|X}(y_2|x_2),&\text{otherwise}
\end{cases},
\end{align}where 
\begin{align}
    P_{Y|XT}(y|x,t)\triangleq\frac{W_{Y|X}(y|x)Q_{T|Y}(t|y)}{\sum\limits_{\tilde{y}}W_{Y|X}(\tilde{y}|x)Q_{T|Y}(t|\tilde{y})}.
\end{align}
\end{subequations}
\item The innocent symbols $x_{1,0}=x_{2,0}=x_0\in\calX$, for the no-communication mode. Note that, similar to the noiseless feedback, for the no-communication mode, the channel input $X_2$ is always equal to $x_0$ regardless of the channel output $Y_1$, therefore, the channel output $Y_2$ can be considered as the result of one round of transmitting $x_0$;
\end{itemize}

For the non-causal case, applying Corollary~\ref{cor:Ach_NC_Simple}, with the above framework, leads to the following result.
\begin{corollary}
\label{corr:Rewrite_Noisy_NC}
Let
\begin{subequations}
\begin{align}
  \calF_{\text{N-NC}} = \left.\begin{cases}R\geq 0: \exists P_{X_1Y_1TUX_2Y_2}\in\calG_{\text{N-NC}}:\\
  R< \bbI(X_1,U;Y_2)-\bbI(U;T|X_1)\\
\end{cases}\right\},
\end{align}
where
\begin{align}\label{eq:D_Noisey_rewrite_NC}
  \calG_{\text{N-NC}}\triangleq \left.\begin{cases}P_{X_1Y_1TUX_2Y_2}:\\
P_{X_1Y_1TUX_2Y_2}(x_1,y_1,t,u,x_2,y_2)=P_{X_1}(x_1)W_{Y|X}(y_1|x_1)P_{T|Y}(t|y_1)\times\\
P_{U|X_1T}(u|x_1,t)P_{X_2|UTX_1}(x_2|u,t,x_1)P_{Y_2|X_1TX_2}(y_2|y_1,x_2)\\
\bbI(U;Y_2|X_1)\ge\bbI(U;T|X_1)\\
P_{Y_2}=Q_0\\
\card{\calU}\leq\card{\calX}\left(\card{\calX}+1\right)\card{\calZ}+1
\end{cases}\right\},
\end{align}and $P_{X_2|UTX_1}$ is defined in \eqref{eq:NL_Rewrite_Noisy}. The covert capacity of the channels with a noisy rewrite option, when the input of the rewrite operation in the second round may depend non-causally on the output of the first round, is lower-bounded as
\begin{align}
\textup{C}_{\mbox{\scriptsize\rm NC}} \ge\mbox{\rm sup}\{R:R\in\calF_{\text{N-NC}}\}.
\label{eq:Rewrite_Noisy_NC}
\end{align}
\end{subequations}
\end{corollary}
Also, for the causal case, applying Corollary~\ref{cor:Y=Z_C}, with the above framework, leads to the following result.
\begin{corollary}
\label{corr:Rewrite_Noisy_C}
Let

\begin{subequations}
\begin{align}
  \calF_{\text{N-C}}= \left.\begin{cases}R\geq 0: \exists P_{UX_1Y_1TX_2Y_2}\in\calG_{\text{N-C}}:\\
  R\le \bbI(U;Y_2)\\
\end{cases}\right\},
\end{align}
where
\begin{align}\label{eq:Binary_Noisy_Rewrite_Example}
  \calG_{\text{N-C}}\triangleq \left.\begin{cases}P_{UX_1Y_1TX_2Y_2}:\\
P_{UX_1Y_1TX_2Y_2}(u,x_1,y_1,t,x_2,y_2)=P_U(u)P_{X_1|U}(x_1|u)\times\\
W_{Y|X}(y_1|x_1)Q_{T|Y}(t|y_1)P_{X_2|UT}(x_2|u,t)P_{Y_2|X_1TX_2}(y_2|x_1,t,x_2)\\
P_{Y_2}=Q_0\\
\card{\calU}\leq\card{\calY}
\end{cases}\right\},
\end{align}and $P_{Y_2|X_1TX_2}$ is defined in \eqref{eq:NL_Rewrite_Noisy}. 
The covert capacity of the channels with a noisy rewrite option, when the input of the rewrite operation in the second round depends causally on the output of the first round, is
\begin{align}
\textup{C}_{\mbox{\scriptsize\rm C}} =\mbox{\rm max}\{R:R\in\calF_{\text{N-C}}\}.
\label{eq:Rewrite_Noisy_C}
\end{align}
\end{subequations}
\end{corollary}
\begin{remark}[Comparison with the Results in {\cite[Section~V]{Weissman10}}]
    By removing the covertness constraint $P_Z=Q_0$ and the inequality conditions in Corollary~\ref{corr:Rewrite_Noisy_NC} and Corollary~\ref{corr:Rewrite_Noisy_C}, they recover the corresponding results in \cite[Section~V.B]{Weissman10}.
\end{remark}
\subsubsection{Writing on a \texorpdfstring{\ac{BSC}}{BSC} with Causal Access to the First Round of Writing}
\label{ex:BSC_Noisy}
Let the innocent symbol $x_0=0$, the forward channel $W_{Y|X}$ be a \ac{BSC} with parameter $\epsilon\in\sbra{0}{0.5}$, and the backward channel $Q_{Y|X}$ be a \ac{BSC} with parameter $\delta\in\sbra{0}{0.5}$. 
\begin{figure}
\centering
\includegraphics[width=9cm]{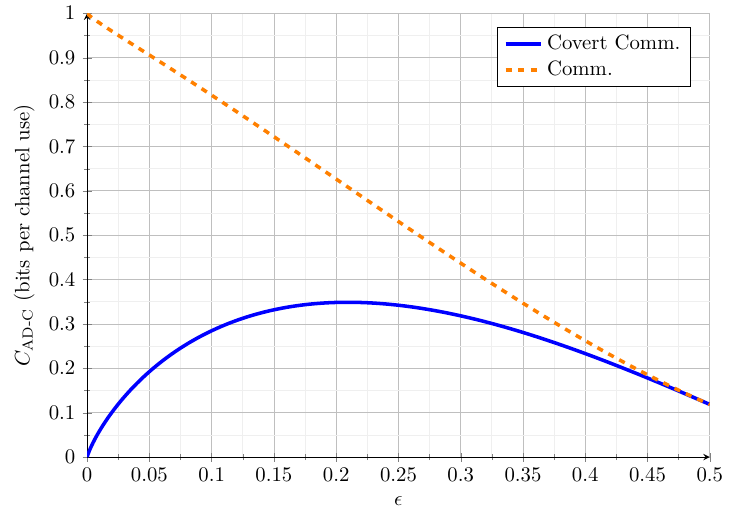}
\caption{Capacity and covert capacity with rewrite option, when a noisy version of the output of the first round of writing is available causally for the second round. Depicted for $\delta=0.1$ as a function of $\epsilon$.}
\label{fig:Causal_epsilon}
\end{figure}
\begin{corollary}
\label{cor:BSC_Writing_Causal_Noisy}
The covert capacity of the \ac{BSC} with a noisy rewrite option, described above, when the input of the rewrite operation in the second round depends causally on the noisy version of the output of the first round, is
\begin{align}
\textup{C}_{\mbox{\scriptsize\rm C}} =\dent_b(\epsilon)-\dent_b\left(\epsilon\delta(2-\epsilon)+\epsilon^2(1-\delta)\right).
\nonumber
\end{align}
\end{corollary}
\begin{proof}
When a noisy version of the first round of writing is available \textit{causally} at the encoder, the arguments of the achievability proof in Corollary~\ref{cor:BSC_Writing_Causal} carry over to this noisy case and the capacity achieving choices of the channel inputs $X_1$ and $X_2$ are the same as those in Corollary~\ref{cor:BSC_Writing_Causal} for the causal case. The converse proof is also similar to that of Corollary~\ref{cor:BSC_Writing_Causal} by considering $U\sim\text{Bernoulli}(\alpha)$ we have
\begin{align}
    \bbI(U;Y_2)&=\bbH(Y_2)-\bbH(Y_2|U)\nonumber\\
    &\mathop=\limits^{(a)}\dent_b(\epsilon)-\bbH(Y_2|U)\nonumber\\
    &\mathop\le\limits^{(b)}\dent_b(\epsilon)-\dent_b\left(\epsilon\delta(2-\epsilon)+\epsilon^2(1-\delta)\right),\label{eq:Capacity_Causal_Noisy_Ex}
\end{align}where $(a)$ follows from the covertness constraint $P_{Y_2}=Q_0$ since $Q_0(Y_2=1)=W_{Y|X=x_0}(y_2=1|x_0=0)=\epsilon$ and $\dent_b(\epsilon)\triangleq-\epsilon\log(\epsilon)-(1-\epsilon)\log(1-\epsilon)$; and $(b)$ follows since the lowest value for the probability of $P(Y_2=1|U=u)$, when we search among the set of all the possible joint distributions, is $\epsilon\delta(2-\epsilon)+\epsilon^2(1-\delta)$ \cite[Section~V-B]{Weissman10}, therefore, $\epsilon\delta(2-\epsilon)+\epsilon^2(1-\delta)$ is still the lowest value for $P(Y_2=1|U=u)$ when we add the covertness constraint $P_{Y_2}=Q_0$ on the feasible set of joint distributions and $\bbH(Y_2|U=u)\ge\dent_b\left(\epsilon\delta(2-\epsilon)+\epsilon^2(1-\delta)\right)$. Also, one can show that the optimal choice of the parameter $\alpha$ is $\alpha=\frac{\epsilon-\left[\epsilon\delta(2-\epsilon)+\epsilon^2(1-\delta)\right]}{1-2\left[\epsilon\delta(2-\epsilon)+\epsilon^2(1-\delta)\right]}$. 
\end{proof}
\begin{figure}
\centering
\includegraphics[width=9cm]{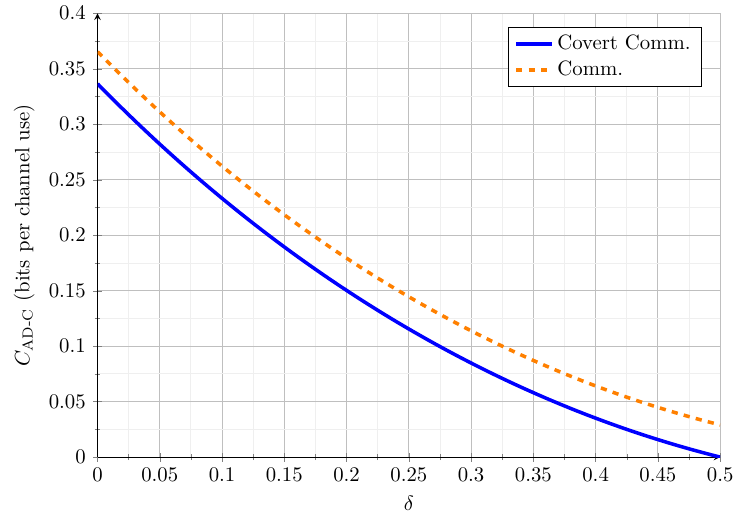}
\caption{Capacity and covert capacity with rewrite option, when a noisy version of the output of the first round of writing is available causally for the second round. Depicted for $\epsilon=0.4$ as a function of $\delta$.}
\label{fig:Causal_delta}
\end{figure}
The covert capacity in Corollary~\ref{cor:BSC_Writing_Causal_Noisy} is plotted in Fig.~\ref{fig:Causal_epsilon} alongside the capacity provided in \cite[Eq.~(59)]{Weissman10}, as a function of $\epsilon$ for $\delta=0.1$. Note, in particular, the two extreme cases. First when $\epsilon=0$, the covert capacity follows the square-root law which goes to zero when the block length grows, but the capacity is maximum and equal to one. This is because when $\epsilon=0$ the writing operation is noiseless and the distribution induced on the output when $x=x_0=0$ is $Q_0\sim\text{Bernoulli}(0)$; therefore, it is not possible to transmit any covert information over the channel. Second, when $\epsilon=\frac{1}{2}$, the channel output is independent of the channel input and therefore the covert capacity and the capacity are equal, but the capacity is at its minimum.  Fig.~\ref{fig:Causal_delta} illustrates the covert capacity in Corollary~\ref{cor:BSC_Writing_Causal_Noisy} alongside the capacity provided in \cite[Eq.~(59)]{Weissman10}, as a function of $\delta$ for $\epsilon=0.4$. For a fixed $\epsilon$, the difference between the capacity and the covert capacity is $1-\dent_b(\epsilon)$, which is constant and is seen in this figure. Also, when $\delta=0$, the covert capacity and the capacity are at their maximum. This is because, in this case, the encoder has access to a noiseless version of the output of the first round of writing. And when $\delta=\frac{1}{2}$, the covert capacity is zero because, in this case, the encoder does not have any information about the output of the first round of writing operation and therefore the problem reduces to a point-to-point channel, and the capacity is also at its minimum. 
\subsubsection{Writing on a \texorpdfstring{\ac{BSC}}{BSC} with Non-Causal Access to the First Round of Writing}
Now we provide a lower bound on the covert capacity, based on Corollary~\ref{corr:Rewrite_Noisy_NC}, when a noisy version of the output of the first round of writing operation is available non-causally at the encoder of the second round of writing operation. 
\begin{corollary}
\label{cor:BSC_Writing_Causal_Noisy_NC}
The covert capacity of the \ac{BSC} with parameter $\epsilon\in\sbra{0}{0.5}$ and a noisy rewrite option, described above, when the input of the rewrite operation in the second round depends non-causally on the noisy version of the output of the first round, is lower bounded as
\begin{align}
\textup{C}_{\mbox{\scriptsize\rm NC}} &\ge\max\limits_{0\le\beta\le1}\Big[\dent_b\big(\epsilon\big)-(1-\alpha)\dent_b\big((1-\epsilon)\left[1-\delta+(\epsilon*\delta)\right]\big)\nonumber\\
    &\qquad-\alpha\left[1-(\epsilon*\delta)\beta\right]\dent_b\left(\frac{(\epsilon*\delta)(1-\beta)\epsilon+\epsilon\delta}{1-(\epsilon*\delta)\beta}\right)-\alpha(\epsilon*\delta)\beta \dent_b\left(\frac{(1-\epsilon)\delta}{\epsilon*\delta}\right)\nonumber\\
    &\qquad\left.-\alpha \dent_b(\epsilon*\delta)+\alpha[1-(\epsilon*\delta)\beta]\dent_b\left(\frac{(\epsilon*\delta)(1-\beta)}{1-(\epsilon*\delta)\beta}\right)\right],
\nonumber
\end{align}such that
\begin{align}
&\alpha \dent_b\big((\epsilon*\delta)(1-\beta)\epsilon+\beta\epsilon(1-\delta)+\epsilon\delta\big)-\alpha\left[1-(\epsilon*\delta)\beta\right]\dent_b\left(\frac{(\epsilon*\delta)(1-\beta)\epsilon+\epsilon\delta}{1-(\epsilon*\delta)\beta}\right)\nonumber\\
&\qquad-\alpha(\epsilon*\delta)\beta \dent_b\left(\frac{(1-\epsilon)\delta}{\epsilon*\delta}\right)\ge\alpha \dent_b(\epsilon*\delta)-\alpha[1-(\epsilon*\delta)\beta]\dent_b\left(\frac{(\epsilon*\delta)(1-\beta)}{1-(\epsilon*\delta)\beta}\right),\nonumber
\end{align}where $\alpha=\frac{\epsilon(1-\epsilon)(1-\delta)}{1-2\epsilon+(2-\beta)\epsilon(1-2\delta)}$ and $\epsilon*\delta=\epsilon(1-\delta)+(1-\epsilon)\delta$.
\end{corollary}
\begin{proof} 
Consider the joint distribution in Corollary~\ref{corr:Rewrite_Noisy_NC} with the channel input of the first round $X_1$ as a Bernoulli \ac{RV} with parameter $\alpha$, the \ac{RV} $U$ as a deterministic function of the channel input $X_1$ and a noisy version of the channel output $Y_1$, i.e., $T$, as
\begin{align}
U&=\begin{cases}
0,&\text{if}\Squad T=X_1\Squad\text{or}\Squad{X_1=0}\\
\text{Bernoulli}(\beta),&\text{otherwise}
\end{cases},
\end{align}
where the condition $T=X_1$ accounts for reliability and the condition $X_1=0$ accounts for covertness. Also, let the channel input of the second round of writing be a deterministic function of $U,X_1$, and $T$ as
\begin{align}
    X_2&=f(U,Y_1,X_1)
    =\begin{cases}
\text{no-rewrite},&\text{if}\Squad T=X_1\Squad\text{or}\Squad{U=1}\\
X_1,&\text{otherwise}
\end{cases}.
\end{align}Under these definitions, one can compute
\begin{subequations}
\begin{align}
      P_{Y|XT}(y|x,t)&\triangleq\frac{W_{Y|X}(y|x)Q_{T|Y}(t|y)}{\sum\limits_{\tilde{y}}W_{Y|X}(\tilde{y}|x)Q_{T|Y}(t|\tilde{y})}&&\nonumber\\
      &=\begin{cases}
\frac{(1-\epsilon)(1-\delta)}{1-\epsilon*\delta},&\text{if}\Squad (x_1,t,y_2)=(0,0,0)\Squad\text{or}\Squad(x_1,t,y_2)=(1,1,1)\\
1-\frac{(1-\epsilon)(1-\delta)}{1-\epsilon*\delta},&\text{if}\Squad (x_1,t,y_2)=(0,0,1)\Squad\text{or}\Squad(x_1,t,y_2)=(1,1,0)\\
\frac{(1-\epsilon)\delta}{\epsilon*\delta},&\text{if}\Squad (x_1,t,y_2)=(0,1,0)\Squad\text{or}\Squad(x_1,t,y_2)=(1,0,1)\\
1-\frac{(1-\epsilon)\delta}{\epsilon*\delta},&\text{if}\Squad (x_1,t,y_2)=(0,1,1)\Squad\text{or}\Squad(x_1,t,y_2)=(1,0,0)
\end{cases},\nonumber\\
    P_{Y_2}(y_2=1)&=(1-\alpha)+\left[\alpha(1-\beta)-(1-\alpha)\right](1-\epsilon)\left(\epsilon*\delta\right)+\alpha\beta(1-\epsilon)\delta+\nonumber\\
    &\qquad\left(2\alpha-1\right)(1-\epsilon)(1-\delta),\label{eq:Y2_Probability_Noisy}\\
    \bbH(T|X_1)&=\dent_b(\epsilon*\delta),\label{eq:Ent_Y1_given_X1_Noisy}\\
    \bbH(T|U,X_1)&=(1-\alpha)\dent_b\left(\epsilon*\delta\right)+\alpha[1-(\epsilon*\delta)\beta]\dent_b\left(\frac{(\epsilon*\delta)(1-\beta)}{1-(\epsilon*\delta)\beta}\right),\label{eq:Ent_Y1_given_UX1_Noisy}\\
    \bbH(Y_2|U,X_1)&=(1-\alpha)\dent_b\big((1-\epsilon)\left[1-\delta+(\epsilon*\delta)\right]\big)\nonumber\\
    &\qquad+\alpha\left[1-(\epsilon*\delta)\beta\right]\dent_b\left(\frac{(\epsilon*\delta)(1-\beta)\epsilon+\epsilon\delta}{1-(\epsilon*\delta)\beta}\right)+\alpha(\epsilon*\delta)\beta \dent_b\left(\frac{(1-\epsilon)\delta}{\epsilon*\delta}\right),\label{eq:Ent_Y2_given_UX1_Noisy}\\
    \bbH(Y_2|X_1)&=(1-\alpha)\dent_b\big((1-\epsilon)\left(1-\delta+\epsilon*\delta\right)\big)+\alpha \dent_b\big((\epsilon*\delta)(1-\beta)\epsilon+\beta\epsilon(1-\delta)+\epsilon\delta\big),\label{eq:Ent_Y2_given_X1_Noisy}
\end{align}where $\epsilon*\delta=\epsilon(1-\delta)+(1-\epsilon)\delta$. 
\end{subequations}Therefore,
\begin{subequations}
\begin{align}
    \bbI(U;T|X_1)&=\bbH(T|X_1)-\bbH(T|U,X_1)\nonumber\\
    &=\alpha\dent_b(\epsilon*\delta)-\alpha[1-(\epsilon*\delta)\beta]\dent_b\left(\frac{(\epsilon*\delta)(1-\beta)}{1-(\epsilon*\delta)\beta}\right),\label{eq:IUtoY1givenX1_Noisy}\\
    \bbI(U;Y_2|X_1)&=\bbH(Y_2|X_1)-\bbH(Y_2|X_1,U)\nonumber\\
    &=\alpha\dent_b\big((\epsilon*\delta)(1-\beta)\epsilon+\beta\epsilon(1-\delta)+\epsilon\delta\big)\nonumber\\
    &\quad-\alpha\left[1-(\epsilon*\delta)\beta\right]\dent_b\left(\frac{(\epsilon*\delta)(1-\beta)\epsilon+\epsilon\delta}{1-(\epsilon*\delta)\beta}\right)-\alpha(\epsilon*\delta)\beta \dent_b\left(\frac{(1-\epsilon)\delta}{\epsilon*\delta}\right),\label{eq:IUtoY2givenX1_Noisy}\\
   \bbI(U,X_1;Y_2)&=H(Y_2)-H(Y_2|U,X_1)\nonumber\\
    &=\dent_b\big((1-\alpha)+\left[\alpha(1-\beta)-(1-\alpha)\right](1-\epsilon)\left(\epsilon*\delta\right)+\alpha\beta(1-\epsilon)\delta\nonumber\\
    &\qquad+\left[2\alpha-1\right](1-\epsilon)(1-\delta)\big)-(1-\alpha)\dent_b\big((1-\epsilon)\left[1-\delta+(\epsilon*\delta)\right]\big)\nonumber\\
    &\qquad-\alpha\left[1-(\epsilon*\delta)\beta\right]\dent_b\left(\frac{(\epsilon*\delta)(1-\beta)\epsilon+\epsilon\delta}{1-(\epsilon*\delta)\beta}\right)-\alpha(\epsilon*\delta)\beta \dent_b\left(\frac{(1-\epsilon)\delta}{\epsilon*\delta}\right)\nonumber\\
    &\mathop=\limits^{(a)}\dent_b\big(\epsilon\big)-(1-\alpha)\dent_b\big((1-\epsilon)\left[1-\delta+(\epsilon*\delta)\right]\big)\nonumber\\
    &\qquad-\alpha\left[1-(\epsilon*\delta)\beta\right]\dent_b\left(\frac{(\epsilon*\delta)(1-\beta)\epsilon+\epsilon\delta}{1-(\epsilon*\delta)\beta}\right)-\alpha(\epsilon*\delta)\beta \dent_b\left(\frac{(1-\epsilon)\delta}{\epsilon*\delta}\right),\label{eq:Final_Covert_Rate_Noisy_NC}
\end{align}where $(a)$ follows from the covertness constraint $P_Z=Q_0$ and
\begin{align}
    &Q_0(y_2=1)=W_{Y|X=x_0}(y_2=1|0)=\epsilon\nonumber\\
    &\Rightarrow (1-\alpha)+\left[\alpha(1-\beta)-(1-\alpha)\right](1-\epsilon)\left(\epsilon*\delta\right)+\alpha\beta(1-\epsilon)\delta+\left[2\alpha-1\right](1-\epsilon)(1-\delta)=\epsilon,\nonumber\\
    &\Rightarrow \alpha=\frac{\epsilon(1-\epsilon)(1-\delta)}{1-2\epsilon+(2-\beta)\epsilon(1-2\delta)},\label{eq:Cov_Cons_RoM_Noiasy_NC}
\end{align}where \eqref{eq:Cov_Cons_RoM_Noiasy_NC} follows by some algebraic manipulations.
\end{subequations}
\begin{figure}
\centering
\includegraphics[width=9cm]{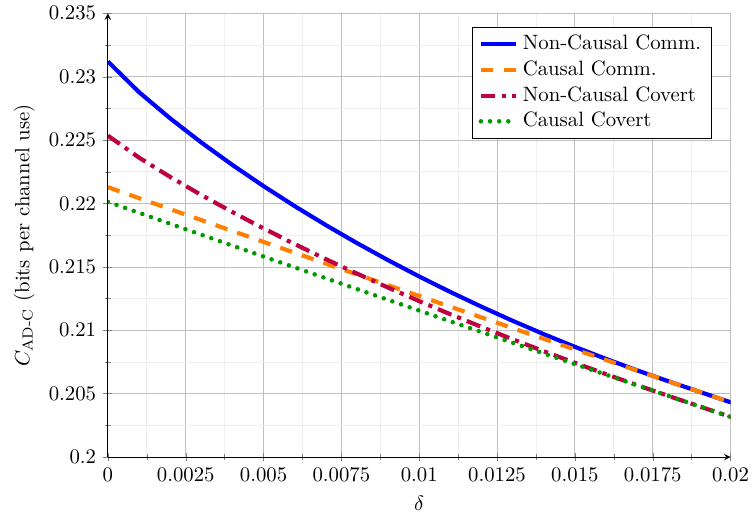}
\caption{Capacity and Covert capacity, in bits per channel use, with rewrite based on corrupted feedback, when the output of the first round of writing operation is available causally for the second round, and the achievable rate and the achievable covert rate when the output of the first round of writing operation is available non-causally for the second round. Depicted for $\epsilon=0.48$ as a function of $\delta$.}
\label{fig:NC_vs_Causal}
\end{figure}
Therefore,
\begin{align}
    \textup{C}_{\mbox{\scriptsize\rm NC}} &\ge\max\left[\bbI(U,X_1;Y_2)-\bbI(U;T|X_1)\right]\nonumber\\
    &=\max\left[\dent_b\big(\epsilon\big)-(1-\alpha)\dent_b\big((1-\epsilon)\left[1-\delta+(\epsilon*\delta)\right]\big)\right.\nonumber\\
    &\qquad-\alpha\left[1-(\epsilon*\delta)\beta\right]\dent_b\left(\frac{(\epsilon*\delta)(1-\beta)\epsilon+\epsilon\delta}{1-(\epsilon*\delta)\beta}\right)-\alpha(\epsilon*\delta)\beta \dent_b\left(\frac{(1-\epsilon)\delta}{\epsilon*\delta}\right)\nonumber\\
    &\qquad\left.-\alpha \dent_b(\epsilon*\delta)+\alpha[1-(\epsilon*\delta)\beta]\dent_b\left(\frac{(\epsilon*\delta)(1-\beta)}{1-(\epsilon*\delta)\beta}\right)\right],\label{eq:RoM_Covert_Rate_Noisy_NC}
\end{align}where the maximum in \eqref{eq:RoM_Covert_Rate_Noisy_NC} is over $\beta$ such that \eqref{eq:Cov_Cons_RoM_Noiasy_NC} holds. Also,
\begin{align}
&\bbI(U;Y_2|X_1)\ge\bbI(U;T|X_1)\nonumber\\
\Rightarrow
&\alpha \dent_b\big((\epsilon*\delta)(1-\beta)\epsilon+\beta\epsilon(1-\delta)+\epsilon\delta\big)-\alpha\left[1-(\epsilon*\delta)\beta\right]\dent_b\left(\frac{(\epsilon*\delta)(1-\beta)\epsilon+\epsilon\delta}{1-(\epsilon*\delta)\beta}\right)\nonumber\\
&\qquad-\alpha(\epsilon*\delta)\beta \dent_b\left(\frac{(1-\epsilon)\delta}{\epsilon*\delta}\right)\ge\alpha \dent_b(\epsilon*\delta)-\alpha[1-(\epsilon*\delta)\beta]\dent_b\left(\frac{(\epsilon*\delta)(1-\beta)}{1-(\epsilon*\delta)\beta}\right).\nonumber
\end{align}
\end{proof}
Fig.~\ref{fig:NC_vs_Causal} illustrates the causal covert capacity in Corollary~\ref{cor:BSC_Writing_Causal_Noisy}, the non-causal achievable covert rate in Corollary~\ref{cor:BSC_Writing_Causal_Noisy_NC}, the causal capacity in \cite[Equation~(59)]{Weissman10}, and the non-causal achievable rate in \cite[Equation~(68)]{Weissman10}, when the forward channel is a \ac{BSC} with parameter $\epsilon=0.48$ and the backward channel is a \ac{BSC} with parameter $\delta$. As seen in this figure, the non-causal access to a noisy version of the first round of writing operation may increase the covert capacity compared to the causal case. More interestingly, as seen in Fig.~\ref{fig:NC_vs_Causal}, the non-causal access to a noisy version of the first round of writing operation can lead to a higher covert capacity as compared to the causal capacity \textit{without} covertness constraint.

\subsection{Computer Memory with Defects and a Rewrite Option}
\label{subsec:Memory_with_Defects}
Consider a computer memory with defects, where we denote the distribution of the state of each memory cell by $P_D$, and the channel law by $W_{Y|XD}$. Consider a two-round coding scheme, where the state of the memory is not known at the encoder and the decoder. After writing on the memory once and observing the output of the writing operation in a noiseless manner, the encoder starts writing on the memory one more time, where it may rewrite on each memory cell that it chooses. Note that the state of each memory cell remains the same regardless of the writing operation.

This problem is a special case of the general setup studied in Section~\ref{sec:Main_Results} if we consider
\begin{itemize}
    \item the action $A$ as the first channel input, which is denoted by $X_1$ in this section and takes values in the alphabet of the input of the channel $W_{Y|XD}$, i.e., $\calX$;
    \item the \ac{ADSI} $S$ as the first channel output, which is denoted by $Y_1$ and takes values in the alphabet of the output of the channel $W_{Y|XD}$,, i.e., $\calY$;
    \item the channel input $X$ as the channel input of the second round, which is denoted by $X_2$ in this section and takes values in $\tilde{\calX}=\{\text{no-rewrite}\}\cup\calX$;
    \item the channel outputs $Y=Z$ as the final channel output, which is denoted by $Y_2$ in this section, and it also takes values in $\calY$.
\end{itemize}Therefore, $Q_{S|A}$ is equal to $P_{Y_1|X_1}$ induced by the original channel, that is
\begin{align}
    P_{Y_1|X_1}=\sum\limits_{d}P_D(d)W_{Y|XD}(y_1|x_1,d),\label{eq:First_Channel}
\end{align}and $W_{Y|ASX}$ is equal to
\begin{align}
    P_{Y_2|X_1Y_1X_2}=\begin{cases}
\indi{1}_{\{y_2=y_1\}},&\text{if}\Squad X_2=\text{no-rewrite}\\
\sum\limits_dW_{Y|XD}(y_2|x_2,d)P_{D|XY}(d|x_1,y_1),&\text{otherwise}
\end{cases}\label{eq:Conditional_Y2_C}
\end{align}where
\begin{align}
    P_{D|XY}(d|x_1,y_1)&=\frac{P_X(x_1)P_D(d)W_{Y|XD}(y_1|x_1,d)}{P_X(x_1)P_{Y_1|X_1}(y_1|x_1)}\nonumber\\
    &=\frac{P_D(d)W_{Y|XD}(y_1|x_1,d)}{\sum\limits_{d'}P_D(d')W_{Y|XD}(y|x_1,d')}.\label{eq:D_given_Y1X1}
\end{align}Note that the knowledge of $Y_1$ provides information about the state $D$ and therefore affects the conditional distribution of the second round of writing. 

For the non-causal case, applying Corollary~\ref{cor:Ach_NC_Simple}, with the above framework, leads to the following achievable rate.
\begin{corollary}
\label{corr:Memory_with_Defects_NC}
Let
\begin{subequations}
\begin{align}
  \calF_{\text{D-NC}}= \left.\begin{cases}R\geq 0: \exists P_{X_1Y_1UX_2Y_2}\in\calG_{\text{D-NC}}:\\
  R< \bbI(X_1,U;Y_2)-\bbI(U;Y_1|X_1)\\
\end{cases}\right\},
\end{align}
where
\begin{align}\label{eq:D_Memory_with_Defects_NC}
  \calG_{\text{D-NC}}\triangleq \left.\begin{cases}P_{X_1Y_1UX_2Y_2}:\\
P_{X_1Y_1UX_2Y_2}(x_1,y_1,u,x_2,y_2)=P_{X_1}(x_1)P_{Y_1|X_1}(y_1|x_1)\times\\
P_{U|X_1Y_1}(u|x_1,y_1)P_{X_2|UX_1Y_1}(x_2|u,x_1,y_1)P_{Y_2|X_1Y_1X_2}(y_2|x_1,y_1,x_2)\\
\bbI(U;Y_2|X_1)\ge\bbI(U;Y_1|X_1)\\
P_{Y_2}=Q_0\\
\card{\calU}\leq\card{\calX}\card{\calY}\left(\card{\calX}+1\right)+1
\end{cases}\right\},
\end{align}and $P_{Y_1|X_1}$ and $P_{Y_2|X_1Y_1X_2}$ are defined in \eqref{eq:First_Channel} and \eqref{eq:Conditional_Y2_C}, respectively. 
The covert capacity of the channels with defects and rewrite option, when the input of the rewrite operation in the second round depends non-causally on the output of the first round, is lower-bounded as
\begin{align}
\textup{C}_{\mbox{\scriptsize\rm NC}} \ge\mbox{\rm sup}\{R:R\in\calF_{\text{D-NC}}\}.
\label{eq:Capacity_Memory_with_Defects_NC}
\end{align}
\end{subequations}
\end{corollary}

For the causal case, applying Corollary~\ref{cor:Y=Z_C}, with the above framework, leads to the following capacity result.
\begin{corollary}
\label{corr:Memory_with_Defects_C}
Let
\begin{subequations}
\begin{align}
  \calF_{\text{D-C}}= \left.\begin{cases}R\geq 0: \exists P_{UX_1Y_1X_2Y_2}\in\calG_{\text{D-C}}:\\
  R\le \bbI(U;Y_2)\\
\end{cases}\right\},
\end{align}
where
\begin{align}\label{eq:Memory_with_Defects_C}
  \calG_{\text{D-C}}\triangleq \left.\begin{cases}P_{UX_1Y_1X_2Y_2}:\\
P_{UX_1Y_1X_2Y_2}(u,x_1,y_1,x_2,y_2)=P_U(u)P_{X_1|U}(x_1|u)\times\\
P_{Y_1|X_1}(y_1|x_1)P_{X_2|UY_1}(x_2|u,y_1)P_{Y_2|X_1Y_1X_2}(y_2|x_1,ty_1,x_2)\\
P_{Y_2}=Q_0\\
\card{\calU}\leq\card{\calY}.
\end{cases}\right\},
\end{align}and $P_{Y_1|X_1}$ and $P_{Y_2|X_1Y_1X_2}$ are defined in \eqref{eq:First_Channel} and \eqref{eq:Conditional_Y2_C}, respectively. 
The covert capacity of the channels with defects and rewrite option, when the input of the rewrite operation in the second round depends causally on the output of the first round, is
\begin{align}
\textup{C}_{\mbox{\scriptsize\rm C}} =\mbox{\rm max}\{R:R\in\calF_{\text{D-C}}\}.
\label{eq:Capacity_Memory_with_Defects_C}
\end{align}
\end{subequations}
\end{corollary}
\begin{remark}
    The achievable covert rate for the non-causal access to the noisy feedback in Corollary~\ref{corr:Memory_with_Defects_NC} and the covert capacity for the causal access to the noisy feedback in Corollary~\ref{corr:Memory_with_Defects_C} recover as special case the corresponding capacity results for the channels with defects and rewrite option in \cite[Section~V.C]{Weissman10}.
\end{remark}
\section{Gaussian Channels}
\label{sec:Gaussian}
In this section, we study the \ac{AWGN} channel when the \ac{ADSI} is available non-causally at the encoder and provide a lower and an upper bound on the covert capacity first, and then we provide an example in which user cooperation leads to a positive covert rate.
\subsection{\texorpdfstring{\ac{AWGN}}{AWGN} Channel}
\label{sec:NC_AWGN}
Consider an \ac{AWGN} channel with \ac{ADSI} in which the channel output of the legitimate receiver and the warden are as follows (see also Fig.~\ref{fig:SYstem_AWGN}),
\begin{align}
    Y&=X+S+N_Y,\quad\text{and}\quad
    Z=X+S+N_Z,\label{eq:system_Model}
\end{align}
where $N_Y\sim\calN(0,\sigma_Y^2)$, and $N_Z\sim\calN(0,\sigma_Z^2)$ are independent. We assume that the test channel $Q_{S|A}$ is also an \ac{AWGN} channel and $S=A+N_S$, where the output of the action encoder $A$ and the additive noise $N_S\sim\calN(0,T)$ are independent, and they are also independent of $N_Y$ and $N_Z$. Therefore,
\begin{align}
    Y&=A+X+N_S+N_Y,\quad\text{and}\quad
    Z=A+X+N_S+N_Z.\label{eq:system_Model_Final}
\end{align}We assume that the action $A^N$ and the channel input $X^N$ are restricted to 
\begin{align}
    \expec\left[\frac{1}{N}\sum\limits_{t=1}^NA_t^2\right]&\leq P_A,\quad\text{and}\quad
    \expec\left[\frac{1}{N}\sum\limits_{t=1}^NX_t^2\right]\leq P_X.\label{eq:Channel_Input_Power}
\end{align}
The innocent symbols are $a_0=x_0=0$ therefore when the transmitter is not communicating with the receiver, the distribution induced on the warden's observation is $Q_0^{\otimes N}$ where $Q_0\sim\calN(0,T+\sigma_Z^2)$. Here, the covertness constraint is \eqref{eq:Covertness}, and the covert capacity is defined as the supremum of all achievable covert rates and is denoted by $\textup{C}_{\mbox{\scriptsize\rm GAD-NC}}$. Now we present the main results of this section, which are a lower and an upper bound on the covert capacity $\textup{C}_{\mbox{\scriptsize\rm GAD-NC}}$.
Let,
\begin{subequations}\label{eq:Gaussian_Rate_Defi}
    \begin{align}
    R_\A&\triangleq\frac{1}{2}\log\left(1+\frac{4T\sigma_Y^2(P_X+P_A)+4T^2P_X-(T+\sigma_Y^2)(P_X+P_A)^2}{4T\sigma_Y^2\big(T-P_A+\sigma_Y^2\big)}\right),\\
    R_\B&\triangleq\frac{1}{2}\log\left(1+\frac{2\sqrt{TP_X}-P_X}{T+\sigma_Y^2+P_X-2\sqrt{TP_X}}\right),\\
    R_\Ci&\triangleq\frac{1}{2}\log\left(1+\frac{T}{\sigma_Y^2}\right).
    \end{align}
\end{subequations}Also, for a fixed $P_X$, $P_A$, and $T$ consider the following three cases 
\begin{subequations}\label{eq:Different_Cases}
\begin{itemize}
    \item Case A:
    \begin{align}
        P_A\le2\sqrt{TP_X}-P_X\,\,\,\text{and}\,\,\, P_X+P_A<2T.\label{eq:Case_A}
    \end{align}
    \item Case B:
    \begin{align}
        P_A>2\sqrt{TP_X}-P_X\,\,\,\text{and}\,\,\, P_X<T.\label{eq:Case_B}
    \end{align}
    \item Case C:
    \begin{align}
        P_X+P_A\ge2\sqrt{TP_X}-P_X\,\,\,\text{and}\,\,\, P_X\ge T.\label{eq:Case_C}
    \end{align}
\end{itemize}
\end{subequations}
\begin{theorem}
\label{thm:Gaussian_Converse}
The covert capacity of the \ac{AWGN} channel with \ac{ADSI} defined in \eqref{eq:system_Model_Final} is upper bounded~as,
\begin{itemize}
    \item For the Case A in \eqref{eq:Case_A}, $\textup{C}_{\mbox{\scriptsize\rm GAD-NC}}\le R_\A$;
    \item For the Case B in \eqref{eq:Case_B}, $\textup{C}_{\mbox{\scriptsize\rm GAD-NC}}\le R_\B$;
    \item For the Case C in \eqref{eq:Case_C}, $\textup{C}_{\mbox{\scriptsize\rm GAD-NC}}\le R_\Ci$.
\end{itemize}
\end{theorem}
Theorem~\ref{thm:Gaussian_Converse} is proved in Appendix~\ref{proof:thm:Converse_AWGN}.
\begin{remark}[Intuitions]
As a consequence of the covertness constraint, the sum of the action input power and the channel input power can not exceed the interference power, i.e., $T$, therefore, the highest covert capacity that one can achieve in the \ac{AWGN} channel defined above is $\frac{1}{2}\log\left(1+\frac{T}{\sigma_Y^2}\right)$. This rate can be achieved when the action encoder and the encoder have enough power to pre-subtract the interference, i.e., $N_S$, and substitute it with signals that carry information about the message, this regime corresponds to the last bullet point of Theorem~\ref{thm:Gaussian_Converse}. Note that the encoder knows the action input and the \ac{ADSI} and therefore it can recover the interference $N_S$. When the interference is strong and the action encoder and the encoder do not have enough power to pre-subtract the entire interference, they can pre-subtract part of it and as a result, the covert capacity that they can achieve will be smaller than $\frac{1}{2}\log\left(1+\frac{T}{\sigma_Y^2}\right)$, this can be seen in the first two bullet points of Theorem~\ref{thm:Gaussian_Converse}.
\end{remark}
\begin{figure*}
\centering
\includegraphics[width=11cm]{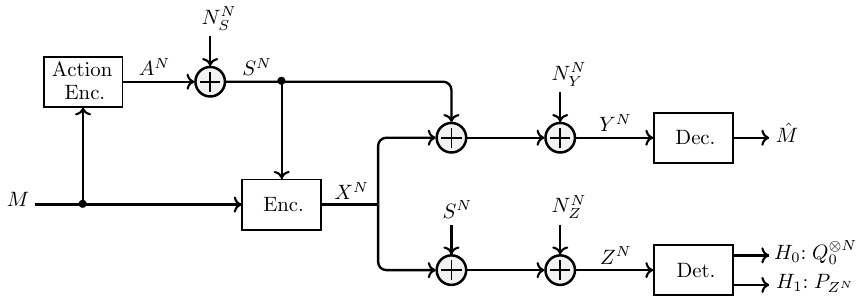}
\caption{The \ac{AWGN} channel with \ac{ADSI}}
\label{fig:SYstem_AWGN}
\end{figure*}
Let $R_\A$, $R_\B$, $R_\Ci$ be defined as \eqref{eq:Gaussian_Rate_Defi}, then the following theorem provides an achievable covert communication rate.
\begin{theorem}
\label{thm:Gaussian_Inner}
The covert capacity of the Gaussian channel with \ac{ADSI} defined in \eqref{eq:system_Model_Final} is lower bounded~as,
\begin{itemize}
    \item For the Case A in \eqref{eq:Case_A}, $\textup{C}_{\mbox{\scriptsize\rm GAD-NC}}\ge R_\A$ when
    \begin{align}
        &\frac{1}{2}\log\left(\frac{P_X^*(T+\sigma_Y^2)}{\sigma_Y^2\left(P_X^*+\sigma_Y^2\right)\big(T-P_A+\sigma_Y^2\big)}\right)\nonumber\\
        &\ge\max\left\{\frac{1}{2}\log\left(\frac{(T+\sigma_Z^2)}{(1-\alpha)^2T\big(P_X^*\sigma_Z^2+\sigma_1^4\big)+\sigma_Z^2\left(P_X^*+\sigma_Y^2\right)^2}\right),\frac{1}{2}\log\left(\frac{(T+\sigma_Z^2)}{\left(T+\sigma_Z^2-P_A\right)\left(P_X^*+\sigma_Y^2\right)^2}\right)\right\},\label{eq:Condion_Gauusian}
    \end{align}where $\alpha\triangleq\frac{P_X+P_A}{2T},T^*\triangleq(1-\alpha)^2T$, and $P_X^*+P_A=T-T^*$;
    \item For the Case B in \eqref{eq:Case_B}, $\textup{C}_{\mbox{\scriptsize\rm GAD-NC}}\ge R_\B$ when $\sigma_Y^2\le\sigma_Z^2$;
    \item For the Case C in \eqref{eq:Case_C}, $\textup{C}_{\mbox{\scriptsize\rm GAD-NC}}\ge R_\Ci$ when $\sigma_Y^2\le\sigma_Z^2$.
\end{itemize}
\end{theorem}
\begin{proof}
We prove the achievable rate in Theorem~\ref{thm:Gaussian_Inner} in three different cases.
\subsubsection{When \texorpdfstring{$P_A\le2\sqrt{TP_X}-P_X$}{} and \texorpdfstring{$P_X+P_A<2T$}{} }
\label{sec:Low_Power_Both}
Let, 
\begin{subequations}\label{eq:Parameter_Definitions}
    \begin{align}
        \alpha&\triangleq\frac{P_X+P_A}{2T},\label{eq:alpha_G}\\
        \beta&\triangleq1,\label{eq:beta_G}\\
        T^*&\triangleq(1-\alpha)^2T,\label{eq:Power_Tstar}\\
        P_X^*+\beta^2P_A&=T-T^*.\label{eq:Power_Constraint_Covertness}
    \end{align}
\end{subequations}
We modify the achievability proof in Corollary~\ref{cor:Ach_NC_Simple}, such that it also applies to the case where the channel is \ac{AWGN} with channel input constraints. Roughly speaking, the idea is to quantize only at the decoder. 
The main idea for achievability is to pre-subtract part of the \ac{ADSI} power to make room for the action and the channel input. Here, we set the action as $A\sim\calN\left(0,\beta^2P_A\right)$, where $\beta=1$, and we set the channel input to be $X=X^*-\alpha N_S$, where $X^*$ is independent of $N_S$ and $A$ and is going to be the channel input of the channel with reduced interference power $T^*\triangleq(1-\alpha)^2T$. Now, the covertness constraint implies that $P_X^*+\beta^2P_A=T-T^*$. 
Note that the power constraint on the action input is satisfied since we choose the action power to be $P_A$ and the channel input power is also satisfied since
\begin{align}
    \bbE\left[X^2\right]&=\bbE\left[X^{*2}\right]+\alpha^2T\nonumber\\
    &=P_X^*+\alpha^2T\nonumber\\
    &\mathop=\limits^{(a)} T-(1-\alpha)^2T-\beta^2P_A+\alpha^2T,\nonumber\\
    &=2\alpha T-\beta^2P_A\nonumber\\
    &\mathop=\limits^{(b)} P_X,
\end{align}where $(a)$ follows from the covertness constraint $P_X^*+\beta^2P_A=T-T^*$; and $(b)$ follows from \eqref{eq:alpha_G} and \eqref{eq:beta_G}. 
Now we choose the auxiliary \ac{RV} $U$, the conditional distribution $P_{U|SA}$, and $P_{X|US}$ in Corollary~\ref{cor:Ach_NC_Simple} as,
\begin{subequations}\label{eq:UXX*_No_Power_Constraint}
\begin{align}
    X^*&\sim\calN(0,P_X^*),\label{eq:PX*_No_Power_Constraint}\\
    U&=X^*+\frac{P_X^*}{P_X^*+\sigma_Y^2}(1-\alpha)N_S,\label{eq:U_No_Power_Constraint}\\
    X&=U-\frac{P_X^*+\alpha\sigma_Y^2}{P_X^*+\sigma_Y^2}N_S.\label{eq:X_No_Power_Constraint}
\end{align}where the last equality follows since $X=X^*-\alpha N_S$.
\end{subequations}Then, we employ the same coding scheme as the coding scheme in Corollary~\ref{cor:Ach_NC_Simple} with the difference that \acp{PMF} are substituted with \acp{PDF}. Now, as an additional step for the encoder, for a small and fixed $\epsilon$, we show that the achievability scheme in Corollary~\ref{cor:Ach_NC_Simple} can be adapted such that the resulting input sequences $X^N$ and $A^N$ satisfy the power constraints,
\begin{align}
    \frac{1}{N}\sum\limits_{t=1}^NX_t^2&\le P_X+\epsilon,\quad\text{and}\quad\frac{1}{N}\sum\limits_{t=1}^NA_t^2\le P_A+\epsilon.\label{eq:Power_Constraints_t}
\end{align}
Therefore we modify the codebook generation in the achievability scheme for Corollary~\ref{cor:Ach_NC_Simple} such that the power constraints are satisfied for all the codewords in our codebooks and show that this modification does not affect our covertness analysis. Let $\calE^*$ denote the error event that \eqref{eq:Power_Constraints_t} is not satisfied. When $\calE^*$ occurs, we replace the input sequence $X^N$ and the action $A^N$ with the all-zero codewords. Similar to \cite{LeeWang18,ZivBC17}, for a fixed $\epsilon>0$, one can show that the probability of $\calE^*$ vanishes when $N$ grows. For a small $\epsilon$, the coding scheme described above satisfies the power constraints in \eqref{eq:Power_Constraints_t}.

Now we show that the coding scheme described above does not affect our covert analysis. Let $\hat{P}_{Z^N}$ denote the distribution induced at the warden's channel output by the coding scheme described above and $\bar{P}_{Z^N}$ denote the distribution induced at the warden's channel output by the coding scheme described above but without the additional step of replacing the codewords that violate the power constraints in \eqref{eq:Power_Constraints_t} with the all-zero codeword. The covertness analysis in the proof of Corollary~\ref{cor:Ach_NC_Simple} is valid when we use \acp{PDF} instead of \acp{PMF} without any power constraint, therefore we have
\begin{align}
    \bbE_{C^{(N)}}\bbV\Big(\bar{P}_{Z^N|C^{(N)}},Q_0^{\otimes N}\Big)\xrightarrow[]{N\to\infty}0.\label{eq:Vanishing_Pbar_TV}
\end{align}
For a fixed codebook $\calC^{(N)}$,
\begin{subequations}\label{eq:Pbar_hat_Gaussian}
\begin{align}
    \bar{P}_{Z^N|\calC^{(N)}}&=P_{\calE^*|\calC^{(N)}}P^*+\big(1-P_{\calE^*|\calC^{(N)}}\big)P'\label{eq:Pbar_Gaussian}\\
    \hat{P}_{Z^N|\calC^{(N)}}&=P_{\calE^*|\calC^{(N)}}Q_0^{\otimes N}+\big(1-P_{\calE^*|\calC^{(N)}}\big)P'\label{eq:Phat_Gaussian}
\end{align}
\end{subequations}where $P^*$ denotes the distribution $\bar{P}_{Z^N|\calC^{(N)}}$ condition on the error event $\calE^*$ and $P'$ denotes the distribution $\bar{P}_{Z^N|\calC^{(N)}}$ condition on the complement of the error event $\calE^*$. Therefore, 
\begin{align}
    \bbV\Big(\hat{P}_{Z^N|\calC^{(N)}},Q_0^{\otimes N}\Big)&=\big(1-P_{\calE^*|\calC^{(N)}}\big)\bbV\Big(P',Q_0^{\otimes N}\Big)\nonumber\\
    &\le\bbV\Big(P',Q_0^{\otimes N}\Big).\label{eq:Phat_Q0}
\end{align}We also have,
\begin{align}
    \bbV\Big(\bar{P}_{Z^N|\calC^{(N)}},Q_0^{\otimes N}\Big)&\mathop\ge\limits^{(a)}\big(1-P_{\calE^*|\calC^{(N)}}\big)\bbV\Big(P',Q_0^{\otimes N}\Big)-P_{\calE^*|\calC^{(N)}}\bbV\Big(P^*,Q_0^{\otimes N}\Big)\nonumber\\
    &\mathop\ge\limits^{(b)}\bbV\Big(P',Q_0^{\otimes N}\Big)-2P_{\calE^*|\calC^{(N)}}\label{eq:Bounding_TV_G}
\end{align}where $(a)$ follows from \eqref{eq:Pbar_Gaussian} and since for $a,b\in\bbR$ we have $|a+b|\ge|a|-|b|$; and $(b)$ follows since the total variation distance is less than 2. Now combining \eqref{eq:Phat_Q0} and \eqref{eq:Bounding_TV_G} leads to
\begin{align}
    \bbE_{C^{(N)}}\bbV\Big(\hat{P}_{Z^N|C^{(N)}},Q_0^{\otimes N}\Big)\le\bbE_{C^{(N)}}\bbV\Big(\bar{P}_{Z^N|C^{(N)}},Q_0^{\otimes N}\Big)+2P_{\calE^*|C^{(N)}}.\label{eq:Average_Phat_TV}
\end{align}From \eqref{eq:Vanishing_Pbar_TV} and the fact that $P_{\calE^*|C^{(N)}}$ goes to zero when $N$ grows, the \ac{RHS} of \eqref{eq:Average_Phat_TV} vanishes when $N$ grows. Since $\hat{P}_{Z^N|C^{(N)}}$ is absolutely continuous \ac{wrt} $Q_0^{\otimes N}$ \eqref{eq:Average_Phat_TV} implies that $\bbE_{C^{(N)}}\bbD\Big(\hat{P}_{Z^N|C^{(N)}}\big|\big|Q_0^{\otimes N}\Big)\xrightarrow[]{n\to\infty}0$.

Similar to \cite[Section~8.6]{Cover_Book} and \cite[Section~VI.A]{LeeWang18} by quantizing the involved \acp{RV} one can show that the achievable rate in Corollary~\ref{cor:Ach_NC_Simple} is achievable. Finally, computing the mutual information and entropy terms in Corollary~\ref{cor:Ach_NC_Simple} by considering \eqref{eq:Parameter_Definitions} and  \eqref{eq:UXX*_No_Power_Constraint} completes the proof of Theorem~\ref{thm:Gaussian_Inner}.

\subsubsection{When \texorpdfstring{$P_X\ge T$}{} and \texorpdfstring{$P_X+P_A\ge2T$}{} }
\label{sec:High_Power}
The achievability proof for this regime is similar to that in Section~\ref{sec:Low_Power_Both}, the only difference is that since we have enough channel input power and action power we choose the auxiliary \ac{RV} $U$, the conditional distribution $P_{U|SA}$, and $P_{X|US}$ in Corollary~\ref{cor:Ach_NC_Simple} as,
    \begin{align*}
    X^*&\sim\calN(0,P_X^*),\quad U=X^*,\quad
    X=U-N_S.
    \end{align*}

\subsubsection{When \texorpdfstring{$P_X+P_A>2\sqrt{TP_X}-P_X$}{} and \texorpdfstring{$P_X<T$}{} }
\label{sec:High_Power_Only_For_Action}
The achievability proof for this regime is also similar to that in Section~\ref{sec:Low_Power_Both}, the only difference is that we choose the auxiliary \ac{RV} $U$, the conditional distribution $P_{U|SA}$, and $P_{X|US}$ in Corollary~\ref{cor:Ach_NC_Simple} as,
    \begin{align*}
    X^*&\sim\calN(0,P_X^*),\quad U=X^*+\frac{P_X^*}{P_X^*+\sigma_Y^2}(1-\alpha)N_S,\\
    X&=X^*-\alpha N_S\Rightarrow X=U-\frac{P_X^*+\alpha\sigma_Y^2}{P_X^*+\sigma_Y^2}N_S,
    \end{align*}where
    \begin{align*}
    \alpha&\triangleq\sqrt{\frac{P_X}
    {T}},\quad\beta\triangleq\sqrt{\frac{2\sqrt{TP_X}-P_X}{P_A}},\quad T^*\triangleq(1-\alpha)^2T,\quad P_X^*+\beta^2P_A=T-T^*,
    \end{align*}and $X^*$, $N_S$, and $A$ are mutually independent.
\end{proof}
\begin{remark}[Comparing the Lower and Upper Bound]
    When the transmitter and the receiver share a secret key of sufficient rate such the conditions \eqref{eq:Condion_Gauusian} and $\sigma_Y^2\le\sigma_Z^2$ are satisfied the lower and upper bounds in Theorem~\ref{thm:Gaussian_Converse} and Theorem~\ref{thm:Gaussian_Inner} meet.
\end{remark}

\subsection{Cooperative Gaussian Channel}
\label{sec:Cooperation}
\begin{figure*}
\centering
\includegraphics[width=11cm]{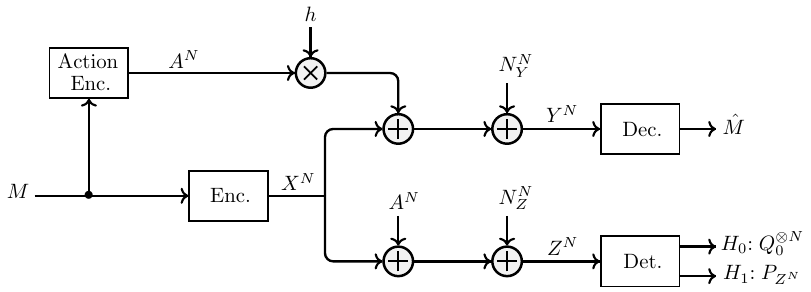}
\caption{Cooperation for covert communication over a Gaussian channel}
\label{fig:SYstem_Gaussian}
\end{figure*}
In the state-dependent channels when the randomness generated by nature is not present, i.e., $S=\emptyset$, the problem reduces to the point-to-point channel \cite{LeeWang18,Keyless22}, in which the covert capacity is shown to follow the square root law and as a result the covert capacity vanishes when the block length $N$ grows. Also, in the \ac{AWGN} channel studied in Section~\ref{sec:NC_AWGN}, since the channel outputs $Y$ and $Z$ are symmetric \ac{wrt} the action $A$ and the channel input $X$, the only resource that the transmitter can exploit to achieve a positive covert rate is the randomness introduced by nature, i.e.,  $N_S$, and when this randomness is absent, i.e., $T=0$, the covert capacity is zero, as seen in Theorem~\ref{thm:Gaussian_Converse}. In this section, we show that when the randomness generated by nature is not present but the channel outputs $Y$ and $Z$ are non-symmetric \ac{wrt} the channel inputs $A$ and $X$, it is still possible to achieve a positive covert rate through user cooperation. 

In the \ac{AWGN} problem discussed in Section~\ref{sec:NC_AWGN} let $N_S=\emptyset$ and the channel law be as (see Fig.~\ref{fig:SYstem_Gaussian}),
\begin{align}
    Y&=X+hA+N_Y,\quad\text{and}\quad
    Z=X+A+N_Z,\label{eq:system_Model_Gauusian}
\end{align}
where $h\in\bbR$ is a constant and known by all the terminals, $N_Y\sim\calN(0,\sigma_Y^2)$, and $N_Z\sim\calN(0,\sigma_Z^2)$ are independent Gaussian \acp{RV}, and $A$ is the output of the action encoder and $X$ is the channel input, which are restricted as \eqref{eq:Channel_Input_Power}. 
Similar to the previous section, the innocent symbols are $a_0=x_0=0$ therefore when the transmitter is not communicating with the receiver, the distribution induced on the warden's observation is $Q_0^{\otimes N}$, where $Q_0\sim\calN(0,\sigma_Z^2)$. This problem can be seen as covert communication over a two-user \ac{MAC}, where both transmitters try to transmit the same message $M$ to a legitimate receiver. Also, this problem can be seen as covert communication over a transmitter with two antennas and a receiver and a warden each with one antenna. Here we assume that the transmitters use a deterministic encoder, as a result since each transmitter knows the message it also knows the other transmitter's channel input. Here, the covertness constraint is \eqref{eq:Covertness}, and the covert capacity is defined as the supremum of all achievable covert rates and is denoted by $\textup{C}_{\mbox{\scriptsize\rm GAD}}$. 
\begin{theorem}
\label{thm:Cooperative_Capacity}
The covert capacity of the Gaussian cooperative \ac{MAC} with one message, defined in \eqref{eq:system_Model_Gauusian}, is
\begin{align*}
    \textup{C}_{\mbox{\scriptsize\rm GAD}}=\frac{1}{2}\log\left(1+\frac{(h-1)^2\beta^2P_A}{\sigma_Y^2}\right),
\end{align*}where $\beta\triangleq\min\left\{1,\sqrt{\frac{P_X}{P_A}}\right\}$.
\end{theorem}
\begin{proof}
We first prove the achievability by using the achievable rate in Corollary~\ref{cor:Ach_NC_Simple}, by setting $S=A$. Similar to the achievability proof of the Theorem~\ref{thm:Gaussian_Inner} in Section~\ref{sec:NC_AWGN}, one can show that the covertness analysis in the achievability proof of Corollary~\ref{cor:Ach_NC_Simple} can be extended to the problem setup studied in this section. We now choose the auxiliary \ac{RV} $U$ in Corollary~\ref{cor:Ach_NC_Simple} as $U=A$ and to satisfy the covertness constraint, we choose the channel input as $X=-A$. This ensures that the channel input $X$ will cancel the effect of the action input $A$ at the warden's channel output, and as a result the warden only observes $N_Z$. Note that when $P_A>P_X$ the action encoder needs to choose the action $A$ such that the encoder has enough power to choose $X=-A$. Therefore, we choose the action input as $A\sim\calN(0,\beta^2P_A)$, where $\beta\triangleq\min\left\{1,\sqrt{\frac{P_X}{P_A}}\right\}$. Note that, $\beta=1$ is equivalent to $P_A\le P_X$ and the action input will be $A\sim\calN(0,P_A)$. In this case, since $P_A\le P_X$ the encoder will have enough power to choose $X=-A$. Also, $\beta=\sqrt{\frac{P_X}{P_A}}$ is equivalent to $P_A>P_X$ and as a result, the action input will be $A\sim\calN(0,P_X)$, which is a valid choice for both the action encoder and the encoder. One can show that by choosing the \acp{RV} as described above, the achievable rate in Corollary~\ref{cor:Ach_NC_Simple} reduces to the covert rate provided in Theorem~\ref{thm:Cooperative_Capacity}.

We now present the converse proof. Let $\tilde{P}_A\triangleq\bbE[A^2]\le P_A$ and $\tilde{P}_X\triangleq\bbE[X^2]\le P_X$. Now we have
\begin{align}
    NR&=H(M)\nonumber\\
    &\mathop\le\limits^{(a)}\bbI(M;Y^N)+N\epsilon\nonumber\\
    &\mathop\le\limits^{(b)}\bbI\big(M,A^N(M),X^N(M);Y^N\big)+N\epsilon\nonumber\\
    &\mathop=\limits^{(c)}\bbI\big(A^N(M),X^N(M);Y^N\big)+N\epsilon\nonumber\\
    &\mathop\le\limits^{(d)}\sum_{t=1}^N\bbI(A_t,X_t;Y_t)+N\epsilon\nonumber\\
    &=\sum_{t=1}^N\bbI(A_T,X_T;Y_T|T=t)+N\epsilon\nonumber\\
    &=N\bbI(A_T,X_T;Y_T|T)+N\epsilon\nonumber\\
    &\le N\bbI(A_T,T,X_T;Y_T)+N\epsilon\nonumber\\
    &\mathop=\limits^{(e)}N\bbI(A,X;Y)+N\epsilon\nonumber\\
    &=N[\dent(Y)-\dent(Y|A,X)]+N\epsilon\nonumber\\
    &\mathop\le\limits^{(f)}N\left[\frac{1}{2}\log2\pi e\left(\tilde{P}_X+2h\bbE[AX]+h^2\tilde{P}_A+\sigma_Y^2+\sigma_Y^2\right)-\frac{1}{2}\log2\pi e\sigma_Y^2\right]+N\epsilon\label{eq:Rate_Cons_Cooperative}
\end{align}where
\begin{itemize}
    \item[$(a)$] follows from Fano's inequality;
    \item[$(b)$] follows since $A^N$ and $X^N$ are functions of the message $M$\footnote{Note that the assumption that $A^N$ and $X^N$ are functions of the message $M$ is consistent with the achievability scheme in Corollary~\ref{cor:Ach_NC_Simple} since a stochastic encoder is not needed in the achievability proof of Corollary~\ref{cor:Ach_NC_Simple}.};
    \item[$(c)$] follows since $M-\big(A^N,X^N\big)-Y^N$ forms a Markov chain;
    \item[$(d)$] follows since the conditioning does not increase the entropy and the fact that the channel is memoryless;
    \item[$(e)$] follows by defining $A\triangleq(A,T)$, $X\triangleq X_T$, and $Y\triangleq Y_T$;
    \item[$(f)$] follows from the maximum differential entropy lemma \cite{ElGamalKim}.
\end{itemize}Now the covertness constraint $P_Z=Q_0$, where $Q_0\sim\calN\big(0,\sigma_Z^2\big)$, is $\tilde{P}_X+2\bbE[AX]+\tilde{P}_A+\sigma_Z^2=\sigma_Z^2$ which leads to $\bbE[AX]=-\frac{\tilde{P}_X+\tilde{P}_A}{2}$. Therefore, the correlation coefficient is $\rho_{AX}=\frac{\bbE[AX]}{\sqrt{\tilde{P}_A\tilde{P}_X}}=-\frac{\tilde{P}_X+\tilde{P}_A}{2\sqrt{\tilde{P}_A\tilde{P}_X}}$. Note that the correlation coefficient is $-1\le\rho_{AX}\le1$ and since $\tilde{P}_X\ge0$ and $\tilde{P}_A\ge0$ we have $-1\le\rho_{AX}\le0$, and $\rho_{AX}\le0$ is always satisfied. Therefore, we must have $-1\le-\frac{\tilde{P}_X+\tilde{P}_A}{2\sqrt{\tilde{P}_A\tilde{P}_X}}$, which is equivalent to $\tilde{P}_X+\tilde{P}_A\le2\sqrt{\tilde{P}_A\tilde{P}_X}\Leftrightarrow\big(\sqrt{\tilde{P}_X}-\sqrt{\tilde{P}_A}\big)^2\le0$, this inequality is satisfied only when $\tilde{P}_X=\tilde{P}_A$ and therefore $\rho_{AX}=-1$. This means that $X=-A$ and therefore the bound in \eqref{eq:Rate_Cons_Cooperative} reduces to
\begin{align}
    R\le\frac{1}{2}\log\left(1+\frac{(h-1)^2\tilde{P}_A}{\sigma_Y^2}\right).\label{eq:Rate_Cons_Cooperative_1}
\end{align}
Since $X=-A$, the power constraints $0\le\tilde{P}_A\le P_A$ and $0\le\tilde{P}_X\le P_X$ reduces to $0\le\tilde{P}_A\le\min\{P_A,P_X\}$ and therefore
\begin{align}
    \argmax_{\substack{0\le\tilde{P}_A\le P_A\\0\le\tilde{P}_X\le P_X\\X=-A}}\left(\frac{(h-1)^2\tilde{P}_A}{\sigma_Y^2}\right)&=\argmax_{\substack{0\le\tilde{P}_A\le\min\{P_A,P_X\}}}\left(\frac{(h-1)^2\tilde{P}_A}{\sigma_Y^2}\right)\nonumber\\
    &=\min\{P_A,P_X\}.\label{eq:argmax_Cooperative}
\end{align}Substituting \eqref{eq:argmax_Cooperative} in \eqref{eq:Rate_Cons_Cooperative_1} completes the converse proof.
\end{proof}
\section{Conclusions}
\label{sec:Conclusion}
This paper studies asymptotically keyless covert communication over channels with \ac{ADSI} when the \ac{ADSI} is known either non-causally or causally at the transmitter but unknown at the warden. Our results show the feasibility of asymptotically keyless covert communication with a positive rate in \acp{DMC} with action-dependent states, \ac{AWGN} channels with action-dependent states, cooperative Gaussian channels, and channels with a rewrite option. The results in this paper can also be used to study covert communication over \ac{MAC} with a common message and \ac{CSI} at one transmitter.

\begin{appendices}

\section{Proof of Theorem~\ref{thm:Acievability_KG}}
\label{proof:thm:Acievability_KG}
Our achievability proof is based on a block Markov coding scheme where we transmit $B$ independent messages over $B$ blocks each of length $n$ and $N=nB$. Therefore, the warden's observation is $Z^N=(Z_1^n,Z_2^n,\dots,Z_B^n)$. The distribution induced by our coding scheme at the warden's channel output observation is denoted by $P_{Z^N}\triangleq P_{Z_1^n,Z_2^n,\dots,Z_B^n}$, and the target distribution at the warden's channel output observation is $Q_0^{\otimes N}\triangleq\prod_{b = 1}^BQ_0^{\otimes n}$. Hence,
\begin{align}
    \kd{P_{Z^N}||}{Q_0^{\otimes N}}&=\kd{P_{Z_1^n,Z_2^n,\dots,Z_B^n}}{\big|\big|Q_0^{\otimes nB}}\nonumber\\
    &=\sum\limits_{b=1}^B\kd{P_{Z_b^n|Z_{b+1}^{B,n}}||}{Q_0^{\otimes n}\big|P_{Z_{b+1}^{B,n}}}\nonumber\\
    &=\sum_{b=1}^B\left[\kd{P_{Z_b^n}||}{Q_0^{\otimes n}}+\kd{P_{Z_b^n|Z_{b+1}^{B,n}}\big|\big|}{P_{Z_b^n}\big|P_{Z_{b+1}^{B,n}}}\right]\nonumber\\
    &=\sum_{b=1}^B\left[\kd{P_{Z_b^n}||}{Q_0^{\otimes n}}+\bbI\left(Z_b^n;Z_{b+1}^{B,n}\right)\right]\label{eq:Cov_Comm_Constraint_NC}
\end{align}where the second inequality follows by defining $Z_{b+1}^{B,n}\triangleq\left(Z_{b+1}^n,Z_{b+2}^n,\dots,Z_B^n\right)$ and the chain rule \cite[Theorem~2.15]{YuryWu_Book}. Thus, $\kd{P_{Z^N}||}{Q_0^{\otimes N}}\xrightarrow[]{n\to\infty}0$ is equivalent to,
\begin{subequations}
\begin{align}
    \kd{P_{Z^n}||}{Q_0^{\otimes n}}&\xrightarrow[]{n\to\infty}0,\label{eq:Eblevel_iid_NC}\\
    \bbI\left(Z_b^n;Z_{b+1}^{B,n}\right)&\xrightarrow[]{n\to\infty}0,\label{eq:Block_Indep_NC}
\end{align}for each $b\in[B]$.
\end{subequations}This requires a coding scheme which induces $Q_0^{\otimes n}$ at the warden's channel output for each block $b\in[B]$ while eliminating the dependencies across the blocks formed as a result of the block-Markov encoding scheme. 
Now fix $\epsilon>0$, $P_A$, $P_{V|S}$, $P_{U|AS}$, and $P_{X|US}$. 

\subsection{Codebook Generation}
\subsubsection{Codebooks for Key Generation}
For each block $b\in[B]$, let $C_{V_b}^{(n)}\triangleq\big(V^n(j_b)\big)_{j_b\in\calJ}$, where $\calJ\triangleq\brk{2^{nR_J}}$, be a set of random codewords generated \ac{iid} according to $P_V$, where $P_V=\sum_{a\in\calA}\sum_{s\in\calS}P_A(a)Q_{S|A}(s|a)P_{V|S}(v|s)$. A realization of $C_{V_b}^{(n)}$ is denoted by $\calC_{V_b}^{(n)}\triangleq\big(v^n(j_b)\big)_{j_b\in\calJ}$. For each block $b\in[B]$, partition the indices $j_b\in\calJ$ into bins $\calB(\ell_b)$, where $\ell_b\in[2^{n\tilde{R}_L}]$, by applying function $\Psi_L:v^n(j_b)\mapsto[2^{n\tilde{R}_L}]$ via random binning by selecting $\Psi_L\big(v^n(j_b)\big)$ independently and uniformly at random for every $v^n(j_b)\in\calV^n$. For each block $b\in[B]$, create a function $\Psi_K:v^n(j_b)\mapsto[2^{n\tilde{R}_K}]$ via random binning by selecting $\Psi_K\big(v^n(j_b)\big)$ independently and uniformly at random for every $v^n(j_b)\in\calV^n$. The transmitter split the secret key $k_b=\Psi_K\big(v^n(j_b)\big)$ generated in the block $b\in[B]$ from the description of the \ac{ADSI} $v^n(j_b)$ into two independent parts $k_{1,b}$, with the rate $\tilde{R}_{K_1}$, and $k_{2,b}$, with the rate $\tilde{R}_{K_2}$, where $\tilde{R}_K=\tilde{R}_{K_1}+\tilde{R}_{K_2}$ \cite[Section~4.3*]{ElGamalKim}, and these secret keys will be used to help the transmitter and the receiver in the block~$b+2$. 

For each block $b\ge3$, we denote the secret key generated in the previous blocks between the transmitter and the receiver by $R_K=R_{K_1}+R_{K_2}$ while the secret that will be generated in the current block $b$ is denoted by $\tilde{R}_K=\tilde{R}_{K_1}+\tilde{R}_{K_2}$, therefore, we should have $\tilde{R}_K\ge R_K$. 
\subsubsection{Action Codebooks}For each block $b\in[B]$, let $C_{A_b}^{(n)}\triangleq\big(A^n(m_b,k_{1,b-2})\big)_{(m_b,k_{1,b-2})\in\calM\times\calK_1}$, where $\calM\triangleq\big[2^{nR}\big]$ and $\calK_1\triangleq\big[2^{nR_{K_1}}\big]$, be a set of random codewords generated \ac{iid} according to $P_A$. We denote a realization of $C_{A_b}^{(n)}$ by $\calC_{A_b}^{(n)}\triangleq\big(a^n(m_b,k_{1,b-2})\big)_{(m_b,k_{1,b-2})\in\calM\times\calK_1}$. 
\subsubsection{Codebooks for Message Transmission}
For each block $b\in[B]$ and for each $(m_b,,k_{1,b-2})\in\calM\times\calK_1$, let $C_{U_b}^{(n)}\triangleq\big(U^n(m_b,k_{1,b-2},\ell_{b-1},k_{2,b-2},i_b)\big)_{(m_b,k_{1,b-2},\ell_{b-1},k_{2,b-2},i_b)\in\calM\times\calK_1\times\calL\times\calK_2\times\calI}$, where $\calL\triangleq\big[2^{nR_L}\big]$,  $\calK_2\triangleq\big[2^{nR_{K_2}}\big]$, and $\calI\triangleq\big[2^{nR_I}\big]$, be a random codebook generated \ac{iid} according to $\prod\nolimits_{i=1}^n P_{U|A}\big(\cdot|A_i(m_b,k_{1,b-2})\big)$, where $P_{U|A}(\cdot|a)=\sum_{s\in\calS}Q_{S|A}(s|a)P_{U|AS}(u|a,s)$. The indices $(\ell_{b-1},k_{b-2},i_b)$ can also be interpreted as three layer random binning. A realization of $C_{U_b}^{(n)}$ is denoted by $\calC_{U_b}^{(n)}\triangleq\big(u^n(m_b,k_{1,b-2},\ell_{b-1},k_{2,b-2},i_b)\big)_{(m_b,k_{1,b-1},\ell_{b-1},k_{2,b-2},i_b)\in\calM\times\calK_1\times\calL\times\calK_2\times\calI}$. Let $C_b^{(n)}\triangleq\left(C_{A_b}^{(n)},C_{V_b}^{(n)},C_{U_b}^{(n)}\right)$, $\calC_b^{(n)}\triangleq\left(\calC_{A_b}^{(n)},\calC_{V_b}^{(n)},\calC_{U_b}^{(n)}\right)$, $C_N\triangleq\left(C_b^{(n)}\right)_{b\in[B]}$, and~$\calC_N\triangleq\left(\calC_b^{(n)}\right)_{b\in[B]}$. Hereafter, the distributions induced by a fixed codebook $\calC_b^{(n)}$ are denoted by $P_{\cdot|\calC_b^{(n)}}$, and the distributions induced by a random codebook $C_b^{(n)}$ are denoted by $P_{\cdot|C_b^{(n)}}$. For $b\in[B]$, to facilitate the analysis, we consider the following joint \ac{PMF} for a fixed codebook $\calC_b^{(n)}$ in which all the indices $m_b$, $k_{1,b-2}$, $j_b$, $\ell_{b-1}$, $k_{2,b-2}$, and $i_b$ are chosen uniformly at random,
\begin{subequations}\label{eq:Encoding_Ideal_PMF_NC}
\begin{align}
    &\Gamma_{M_bK_{1,b-2}A^nJ_bV^nL_{b-1}K_{2,b-2}I_bU^nS_b^nY_b^nZ_b^nK_{b-1}L_bK_b|\calC_b^{(n)}}(m_b,k_{1,b-2},\tilde{a}^n,j_b,\tilde{v}^n,\ell_{b-1},k_{2,b-2},i_b,\tilde{u}^n,s_b^n,\nonumber\\
    &\qquad y_b^n,z_b^n,k_{b-1},\ell_b,k_b)=\frac{1}{\abs{\calM}\abs{\calK_1}\abs{\calJ}\abs{\calL}\abs{\calK_2}\abs{\calI}}\indi{1}_{\{a^n(m_b,k_{1,b-2})=\tilde{a}^n\}\cap\{v^n(j_b)=\tilde{v}^n\}}\times\nonumber\\
    &\qquad\indi{1}_{\{u^n(m_b,k_{1,b-2},\ell_{b-1},k_{2,b-2},i_b)=\tilde{u}^n\}}Q_{S|AVU}^{\otimes n}\big(s^n|\tilde{a}^n,\tilde{v}^n,\tilde{u}^n\big) W_{YZ|SU}^{\otimes n}(y_b^n,z_b^n|s^n,\tilde{u}^n)\times\nonumber\\
    &\qquad\frac{1}{\abs{\calK_1}\abs{\calK_2}}\indi{1}_{\big\{\ell_b=\Psi_L\big(v^n(j_b)\big)\big\}\bigcap\big\{k_b=\Psi_K\big(v^n(j_b)\big)\big\}},\label{eq:Encoding_Ideal_Joint_NC}
\end{align}where $W_{YZ|SU}(y,z|s,u)=\sum_{x\in\calX}P_{X|US}(x|u,s)W_{YZ|SX}(y,z|s,x)$, and
\begin{align}
    Q_{S|AVU}(s|a,v,u)&=\frac{P_A(a)Q_{S|A}(s|a)P_{V|S}(v|s)P_{U|AS}(u|a,s)}{\sum_{s\in\calS} P_A(a)Q_{S|A}(s|a)P_{V|S}(v|s)P_{U|AS}(u|a,s)}.\label{eq:Ideal_PMF_S_Dist_NC}
\end{align}
\end{subequations}

\subsection{Encoding}To initiate the key generation process, the transmitter and the receiver are assumed to share $k_{-1}\triangleq(k_{1,-1},k_{2,-1})$ and $k_0\triangleq(k_{1,0},k_{2,0})$ to be used in blocks $b=1$ and $b=2$, respectively. After the block $b=2$, the transmitter and the receiver use the secret key that they generate from the \ac{ADSI}.
\subsubsection{Encoding Scheme for the First Block}
Given the key $k_{-1}\triangleq(k_{1,-1},k_{2,-1})$, to transmit the message $m_1\in\calM$, the encoder first chooses the action sequence $a^n(m_1,k_{1,-1})$, then, the nature chooses the channel state $s_1^n$ according to the action sequence $a^n(m_1,k_{1,-1})$. Next, the encoder chooses the reconciliation index $\ell_0$ uniformly at random. Note that, the reconciliation index $\ell_0$ does not convey any information about the channel state. Now given the key $k_{-1}$, the encoder chooses the indices $i_1$ and $j_1$ according to the following distribution with $b=1$,
\begin{align}
&g_{\text{\tiny{LE}}}\big(i_b,j_b|m_b,\ell_{b-1},k_{b-2},a^n(m_b,k_{1,b-2}),s_b^n\big)\triangleq\nonumber\\
&\quad\frac{Q_{S|AVU}^{\otimes n}\big(s_b^n|a^n(m_b,k_{1,b-2}),v^n(j_b),u^n(m_b,k_{1,b-2},\ell_{b-1},k_{2,b-2},i_b)\big)}{\sum_{i'_b\in\calI}\sum_{j'_b\in\calJ}Q_{S|AVU}^{\otimes n}\big(s_b^n|a^n(m_b,k_{1,b-2}),v^n(j'_b),u^n(m_b,k_{1,b-2},\ell_{b-1},k_{2,b-2},i'_b)\big)},\label{eq:LE_NC}
\end{align}where $k_{b-2}\triangleq(k_{1,b-2},k_{2,b-2})$. 
Then, based on $(m_1,k_{1,-1},\ell_0,k_{2,-1},i_1)$ and, $j_1$ the encoder computes $u^n(m_1,k_{1,-1},\ell_0,k_{2,-1},i_1)$ and $v^n(j_1)$ and transmits $x_1^n$, where $x_{1,i}$ is generated by passing $u_i(m_1,k_{1,-1},\ell_0,k_{2,-1},i_1)$ and $s_{1,i}$ through the test channel $P_{X|US}\big(x_{1,i}|u_i(m_1,k_{1,-1},\ell_0,k_{2,-1},i_1),s_i\big)$. Simultaneously, the encoder generates a reconciliation index $\ell_1$ and a key $k_1$, which will be split into two independent parts $k_{1,1}$ and $k_{2,1}$, from the description of the \ac{ADSI} $v^n(j_1)$ to be transmitted in the second and the third blocks, respectively. 
\subsubsection{Encoding Scheme for the Second Block}
Similarly, given the key $k_0\triangleq(k_{1,0},k_{2,0})$, to transmit the message $m_2\in\calM$, the encoder first chooses the action sequence $a^n(m_2,k_{1,0})$, then, the nature chooses the channel state $s_2^n$ according to the action sequence $a^n(m_2,k_{1,0})$. Now given the key $k_0$ and the reconciliation information $\ell_1$, generated in the previous block, the encoder chooses the indices $i_2$ and $j_2$ according to \eqref{eq:LE_NC} with $b=2$. Then, based on $(m_2,k_{1,0},\ell_1,k_{2,0},i_2)$ and, $j_2$ the encoder computes $u^n(m_2,k_{1,0},\ell_1,k_{2,0},i_2)$ and $v^n(j_2)$ and transmits $x_2^n$, where $x_{2,i}$ is generated by passing $u_i(m_2,k_{1,0},\ell_1,k_{2,0},i_2)$ and $s_{2,i}$ through the test channel $P_{X|US}\big(x_{2,i}|u_i(m_2,k_{1,0},\ell_1,k_{2,0},i_2),s_{2,i}\big)$. Simultaneously, the encoder generates a reconciliation index $\ell_2$ and a key $k_2$, which will be split in two independent parts $k_{1,2}$ and $k_{2,2}$, from the description of the \ac{ADSI} $v^n(j_2)$ to be transmitted in the third and the fourth blocks, respectively.
\subsubsection{Encoding Scheme for the Block \texorpdfstring{$b\in\sbra{3}{B}$}{Lg}}
Similarly, given the key $k_{b-2}\triangleq(k_{1,b-2},k_{2,b-2})$, to transmit the message $m_b\in\calM$, the encoder first chooses the action sequence $a^n(m_b,k_{1,b-2})$, then, the nature chooses the channel state $s_b^n$ according to the action sequence $a^n(m_b,k_{1,b-2})$. Now given the key $k_{b-2}$, generated in the block $b-2$, and the reconciliation information $\ell_{b-1}$, generated in the previous block, the encoder chooses the indices $i_b$ and $j_b$ according to \eqref{eq:LE_NC}. Then, based on $(m_b,k_{1,b-2},\ell_{b-1},k_{2,b-2},i_b)$ and, $j_b$ the encoder computes $u^n(m_b,k_{1,b-2},\ell_{b-1},k_{2,b-2},i_b)$ and $v^n(j_b)$ and transmits $x_b^n$, where $x_{b,i}$ is generated by passing $u_i(m_b,k_{1,b-2},\ell_{b-1},k_{2,b-2},i_b)$ and $s_{b,i}$ through the test channel $P_{X|US}\big(x_{b,i}|u_i(m_b,k_{1,b-2},\ell_{b-1},k_{2,b-2},i_b),s_{b,i}\big)$. Simultaneously, the encoder generates a reconciliation index $\ell_b$ and a key $k_b\triangleq(k_{1,b},k_{2,b})$, which will be split in two independent parts $k_{1,b}$ and $k_{2,b}$, from the description of the \ac{ADSI} $v^n(j_b)$ to be transmitted in block $b+1$ and block $b+2$, respectively. 
Note that, in this scheme, the description of the \ac{ADSI} is used only for the key generation, not the message transmission.

Therefore, noting that our coding scheme ensures that the reconciliation index $L_{b-1}$ and the keys $K_{b-1}$ and $K_b$ are (arbitrarily) nearly uniformly distributed, the joint \ac{PMF} between the involved \acp{RV} is given by,
\begin{align}
    &P_{M_bK_{1,b-2}L_{b-1}K_{2,b-2}A^nS_b^nI_bJ_bU^nV^nX^nY_b^nZ_b^nK_{b-1}L_bK_b|\calC_b^{(n)}}(m_b,k_{1,b-2},\ell_{b-1},k_{2,b-2},\tilde{a}^n,s_b^n,i_b,j_b,\tilde{u}^n,\tilde{v}^n\nonumber\\
    &\quad,\tilde{x}^n,y_b^n,z_b^n,k_{b-1},\ell_b,k_b)=\frac{1}{\abs{\calM}\abs{\calK_1}\abs{\calL}\abs{\calK_2}}\indi{1}_{\{a^n(m_b,k_{1,b-2})=\tilde{a}^n\}}Q_{S|A}^{\otimes n}\big(s_b^n|\tilde{a}^n\big)\nonumber\\
    &\quad\times g_{\text{\tiny{LE}}}\big(i_b,j_b|m_b,\ell_{b-1},k_{b-2},\tilde{a}^n,s_b^n\big)\indi{1}_{\{u^n(m_b,k_{1,b-2},\ell_{b-1},k_{2,b-2},i_b)=\tilde{u}^n\}\cap\{v^n(j_b)=\tilde{v}^n\}}\nonumber\\
    &\quad\times P_{X|US}^{\otimes n}\left(\tilde{x}^n|u^n(m_b,k_{1,b-2},\ell_{b-1},k_{2,b-2},i_b),s_b^n\right)W_{YZ|SX}^{\otimes n}\big(y_b^n,z_b^n|s_b^n,\tilde{x}^n\big)\nonumber\\
    &\quad\times\frac{1}{\abs{\calK_1}\abs{\calK_2}}\indi{1}_{\big\{\ell_b=\Psi_L\big(v^n(j_b)\big)\big\}\bigcap\big\{k_b=\Psi_K\big(v^n(j_b)\big)\big\}}.\label{eq:Encoding_Joint_NC}
\end{align}
\begin{figure*}
\centering
\includegraphics[width=6.0in]{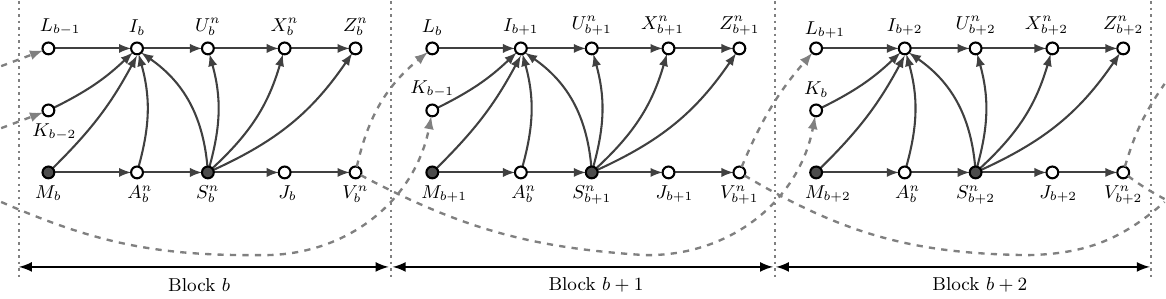}
\caption{Functional dependency graph for all the involved \acp{RV}.}
\label{fig:Dependency_NC}
\end{figure*}
\subsection{Covert Analysis}To prove that our coding scheme is covert, we show that $\bbE_{C_N}\bbD\left(P_{Z^N|C_N}||Q_Z^{\otimes N}\right)\xrightarrow[]{N\to\infty}0$, where
\begin{align}
    Q_Z(\cdot)&\triangleq\sum_{a\in\calA}\sum_{v\in\calV}\sum_{u\in\calU}\sum_{s\in\calS}\sum_{x\in\calX}P_A(a)P_V(v)P_{U|A}(u|a)Q_{S|AVU}(s|a,v,u)P_{X|US}(x|u,s)W_{Z|XS}(\cdot|x,s),\nonumber
\end{align}such that,
\begin{align}
    &\sum_{v\in\calV}\sum_{u\in\calU}P_V(v)P_{U|A}(u|a)Q_{S|AVU}(s|a,v,u)=Q_{S|A}(\cdot|a).\nonumber
\end{align}
Then we choose the \acp{PMF} $P_A$, $P_{V|S}$, $P_{U|SA}$, and $P_{X|US}$ such that $Q_0=Q_Z$. Now, for $b\in[B]$
\begin{subequations}\label{eq:Dep_Graph_equ2_NC}
\begin{align}
   \bbI\left(Z_b^n;Z_{b+1}^{B,n}\right)&\le\bbI\left(Z_b^n;K_{b-1},L_b,K_b,Z_{b+1}^{B,n}\right)\nonumber\\
   &=\bbI\left(Z_b^n;K_{b-1},L_b,K_b\right),\label{eq:Dep_Graph_equ_NC}
\end{align}where \eqref{eq:Dep_Graph_equ_NC} follows since $Z_b^n-\left(K_{b-1},L_b,K_b\right)-Z_{b+1}^{B,n}$ forms a Markov chain, as illustrated in the functional dependency graph in Fig.~\ref{fig:Dependency_NC}. Now,
\begin{align}
    \bbI\left(Z_b^n;K_{b-1},L_b,K_b\right)&=\bbD\left(P_{Z_b^nK_{b-1}L_bK_b|\calC_b^{(n)}}\big|\big|P_{Z_b^n|\calC_b^{(n)}}P_{K_{b-1}L_bK_b|\calC_b^{(n)}}\right)\nonumber\\
    &\le\bbD\left(P_{Z_b^nK_{b-1}L_bK_b|\calC_b^{(n)}}\big|\big|Q_Z^{\otimes n}P_{K_{b-1}}^UP_{L_b}^UP_{K_b}^U\right),\label{eq:Dep_Graph_equ21_NC}
\end{align}where $P_{K_{b-1}}^UP_{L_b}^UP_{K_b}^U$ is the uniform distribution over $\Big[2^{nR_K}\Big]\times\left[2^{n\tilde{R}_L}\right]\times\left[2^{n\tilde{R}_K}\right]$ and \eqref{eq:Dep_Graph_equ21_NC} holds since
\begin{align}
    &\bbD\left(P_{Z_b^nK_{b-1}L_bK_b|\calC_b^{(n)}}\big|\big|P_{Z_b^n|\calC_b^{(n)}}P_{K_{b-1}L_bK_b|\calC_b^{(n)}}\right)=\bbD\left(P_{Z_b^nK_{b-1}L_bK_b|\calC_b^{(n)}}\big|\big|Q_Z^{\otimes n}P_{K_{b-1}}^UP_{L_b}^UP_{K_b}^U\right)-\nonumber\\
    &\quad\bbD\left(P_{Z_b^n|\calC_b^{(n)}}\big|\big|Q_Z^{\otimes n}\right)-\bbD\left(P_{K_{b-1}L_bK_b|\calC_b^{(n)}}\big|\big|P_{K_{b-1}}^UP_{L_b}^UP_{K_b}^U\right).\label{eq:Dep_Graph_equ22_NC}
\end{align}
\end{subequations}Therefore, combining \eqref{eq:Dep_Graph_equ2_NC} with the expansion in  \eqref{eq:Cov_Comm_Constraint_NC}, by substituting $Q_0$ with $Q_Z$ and considering the expectation over the random codebook construction, results to
\begin{align}
    \bbE_{C_N}\kd{P_{Z^N|C_N}||}{Q_0^{\otimes N}}&\le 2\sum\limits_{b=1}^B\bbE_{C_b^{(n)}}\bbD\left(P_{Z_b^nK_{b-1}L_bK_b|C_b^{(n)}}\big|\big|Q_Z^{\otimes n}P_{K_{b-1}}^UP_{L_b}^UP_{K_b}^U\right).\label{eq:Cov_Cons_General_NC}
\end{align}To bound the \ac{RHS} of \eqref{eq:Cov_Cons_General_NC}, using \eqref{eq:Reverse_Pinsker} in Lemma~\ref{lemma:KLD_TV}, it is sufficient to bound,
\begin{align}
    &\ToV\left(P_{Z_b^nK_{b-1}L_bK_b|\calC_b^{(n)}},Q_Z^{\otimes n}P_{K_{b-1}}^UP_{L_b}^UP_{K_b}^U\right)\nonumber\\
    &\le\ToV\left(P_{Z_b^nK_{b-1}L_bK_b|\calC_b^{(n)}},\Gamma_{Z_b^nK_{b-1}L_bK_b|\calC_b^{(n)}}\right)+\ToV\left(\Gamma_{Z_b^nK_{b-1}L_bK_b|\calC_b^{(n)}},Q_Z^{\otimes n}P_{K_{b-1}}^UP_{L_b}^UP_{K_b}^U\right),\label{eq:General_KLD_NC}
\end{align}where the inequality follows from the triangle inequality. 
Note that, from \eqref{eq:Encoding_Ideal_PMF_NC} and \eqref{eq:Encoding_Joint_NC} we have the following marginal distributions,
\begin{align}
    &\Gamma_{Z_b^nK_{b-1}L_bK_b|\calC_b^{(n)}}(z^n,k_{b-1},\ell_b,k_b)=\sum\limits_{m_b}\sum\limits_{k_{1,b-2}}\sum\limits_{j_b}\sum\limits_{\ell_{b-1}}\sum\limits_{k_{2,b-2}}\sum\limits_{i_b}\sum\limits_{s_b^n}\frac{1}{2^{n(R+R_J+R_L+R_I+2R_K)}}\nonumber\\
    &\qquad\times Q_{S|AVU}^{\otimes n}\Big(s^n|a^n(m_b,k_{1,b-2}),v^n(j_b),u^n\big(m_b,k_{1,b-2},\ell_{b-1},k_{2,b-2},i_b\big)\Big)\nonumber\\
    &\qquad\times W_{Z|SU}^{\otimes n}\Big(z_b^n|s_b^n,u^n(m_b,k_{1,b-2},\ell_{b-1},k_{2,b-2},i_b)\Big)\indi{1}_{\big\{\ell_b=\Psi_L\big(v^n(j_b)\big)\big\}\bigcap\big\{k_b=\Psi_K\big(v^n(j_b)\big)\big\}}\nonumber\\
    &=\sum\limits_{m_b}\sum\limits_{k_{1,b-2}}\sum\limits_{j_b}\sum\limits_{\ell_{b-1}}\sum\limits_{k_{2,b-2}}\sum\limits_{i_b}\frac{1}{2^{n(R+R_J+R_L+R_I+2R_K)}}\nonumber\\ 
    &\qquad\times W_{Z|AVU}^{\otimes n}\Big(z_b^n|a^n(m_b,k_{1,b-2}),v^n(j_b),u^n(m_b,k_{1,b-2},\ell_{b-1},k_{2,b-2},i_b)\Big)\nonumber\\
    &\qquad\times\indi{1}_{\big\{\ell_b=\Psi_L\big(v^n(j_b)\big)\big\}\bigcap\big\{k_b=\Psi_K\big(v^n(j_b)\big)\big\}},\label{eq:Encoding_GZ_NC}
\end{align}where
\begin{subequations}
\begin{align}
    W_{Z|SU}(z|s,u)&\triangleq\sum_{y\in\calY}\sum_{x\in\calX}P_{X|US}(x|u,s)W_{YZ|SX}(y,z|s,x),\nonumber\\
    W_{Z|AVU}(z|a,v,u)&\triangleq\sum_{y\in\calY}\sum_{x\in\calX}\sum_{s\in\calS}Q_{S|AVU}(s|a,v,u)P_{X|US}(x|u,s)W_{YZ|SX}(y,z|s,x).\nonumber
\end{align}
\end{subequations}To bound the first term on the \ac{RHS} of \eqref{eq:General_KLD_NC} we have
\begin{align}
    &\bbE_{C_b^{(n)}}\ToV\left(P_{Z_b^nK_{b-1}L_bK_b|C_b^{(n)}},\Gamma_{Z_b^nK_{b-1}L_bK_b|C_b^{(n)}}\right)\nonumber\\
    &\le\bbE_{C_b^{(n)}}\ToV\left(P_{M_bK_{1,b-2}L_{b-1}K_{2,b-2}A^nS_b^nI_bJ_bU^nV^nX^nY_b^nZ_b^nK_{b-1}L_bK_b|C_b^{(n)}},\right.\nonumber\\
    &\qquad\left.\Gamma_{M_bK_{1,b-2}L_{b-1}K_{2,b-2}A^nS_b^nI_bJ_bU^nV^nX^nY_b^nZ_b^nK_{b-1}L_bK_b|C_b^{(n)}}\right)\nonumber\\
    &\mathop=\limits^{(a)}\bbE_{C_b^{(n)}}\ToV\left(P_{M_bL_{b-1}K_{1,b-2}K_{2,b-2}A^nS_b^n|C_b^{(n)}},\Gamma_{M_bL_{b-1}K_{1,b-2}K_{2,b-2}A^nS_b^n|C_b^{(n)}}\right)\nonumber\\
    &\mathop=\limits^{(b)}\bbE_{C_b^{(n)}}\ToV\left(P_{A^nS_b^n|M_b=1,L_{b-1}=1,K_{1,b-2}=1,K_{2,b-2}=1,C_b^{(n)}},\Gamma_{A^nS_b^n|M_b=1,L_{b-1}=1,K_{1,b-2}=1,K_{2,b-2}=1,C_b^{(n)}}\right)\nonumber\\
    &\mathop=\limits^{(c)}\bbE_{C_{A_b}^{(n)}}\left[\bbE_{C_{U_b}^{(n)},C_{V_b}^{(n)}|C_{A_b}^{(n)}}\left[\ToV\left(P_{A^nS_b^n|M_b=1,L_{b-1}=1,K_{1,b-2}=1,K_{2,b-2}=1,C_b^{(n)}},\right.\right.\right.\nonumber\\
    &\qquad\left.\left.\left.\Gamma_{A^nS_b^n|M_b=1,L_{b-1}=1,K_{1,b-2}=1,K_{2,b-2}=1,C_b^{(n)}}\right)\left|C_{A_b}^{(n)}\right.\right]\right],\label{eq:General_TV_NC}
\end{align}where
\begin{itemize}
    \item[$(a)$] follows since
    \begin{subequations}
\begin{align}
    &P_{I_bJ_b|M_bL_{b-1}K_{1,b-2}K_{2,b-2}A^nS_b^n\calC_b^{(n)}}=g_{\text{\tiny{LE}}}\big(I_b,J_b|M_b,K_{1,b-2},L_{b-1},K_{2,b-2},a^n(M_b,K_{1,b-2}),S_b^n\big)\nonumber\\
    &\qquad=\Gamma_{I_bJ_b|M_bL_{b-1}K_{1,b-2}K_{2,b-2}A^nS_b^n\calC_b^{(n)}}\label{eq:PGNC_3},\\
    &P_{U^n|M_bL_{b-1}K_{1,b-2}K_{2,b-2}A^nS_b^nI_bJ_b\calC_b^{(n)}}=\indi{1}_{\{U^n=u^n(M_b,K_{1,b-2},L_{b-1},K_{2,b-2},I_b)\}}\nonumber\\
    &\qquad=\Gamma_{U^n|M_bL_{b-1}K_{1,b-2}K_{2,b-2}A^nS_b^nI_bJ_b\calC_b^{(n)}}\label{eq:PGNC_4},\\
    &P_{V^n|M_bL_{b-1}K_{1,b-2}K_{2,b-2}A^nS_b^nI_bJ_bU^n\calC_b^{(n)}}=\indi{1}_{\{V^n=v^n(J_b)\}}=\Gamma_{V^n|M_bL_{b-1}K_{1,b-2}K_{2,b-2}A^nS_b^nI_bJ_bU^n\calC_b^{(n)}}\label{eq:PGNC_5},\\
    &P_{X^n|M_bL_{b-1}K_{1,b-2}K_{2,b-2}A^nS_b^nI_bJ_bU^nV^n\calC_b^{(n)}}=P_{X|US}^{\otimes n}=\Gamma_{X^n|M_bL_{b-1}K_{1,b-2}K_{2,b-2}A^nS_b^nI_bJ_bU^nV^n\calC_b^{(n)}}\label{eq:PGNC_6},\\
    &P_{Y^nZ^n|M_bL_{b-1}K_{1,b-2}K_{2,b-2}A^nS_b^nI_bJ_bU^nV^nX^n\calC_b^{(n)}}=W_{YZ|SX}^{\otimes n}\nonumber\\
    &\qquad=\Gamma_{Y^nZ^n|M_bL_{b-1}K_{1,b-2}K_{2,b-2}A^nS_b^nI_bJ_bU^nV^nX^n\calC_b^{(n)}}\label{eq:PGNC_7},\\
    &P_{L_bK_b|M_bL_{b-1}K_{1,b-2}K_{2,b-2}A^nS_b^nI_bJ_bU^nV^nX^nY^nZ^n\calC_b^{(n)}}=\indi{1}_{\big\{\ell_b=\Psi_L\big(v^n(j_b)\big)\big\}\bigcap\big\{k_b=\Psi_K\big(v^n(j_b)\big)\big\}}\nonumber\\
    &\qquad=\Gamma_{L_bK_b|M_bL_{b-1}K_{1,b-2}K_{2,b-2}A^nS_b^nI_bJ_bU^nV^nX^nY^nZ^n\calC_b^{(n)}}\label{eq:PGNC_8},
\end{align}
\end{subequations}while \eqref{eq:PGNC_3} follows since
\begin{align}
    &\Gamma_{I_bJ_b|M_bL_{b-1}K_{1,b-2}K_{2,b-2}A^nS_b^n\calC_b^{(n)}}\nonumber\\
    &=\frac{\Gamma_{I_bJ_bM_bL_{b-1}K_{1,b-2}K_{2,b-2}A^nS_b^n|\calC_b^{(n)}}}{\Gamma_{M_bL_{b-1}K_{1,b-2}K_{2,b-2}A^nS_b^n|\calC_b^{(n)}}}\nonumber\\
    &=\frac{Q_{S|AVU}^{\otimes n}\big(s^n|a^n(m_b,k_{1,b-2}),v^n(j_b),u^n(m_b,k_{1,b-2},\ell_{b-1},k_{2,b-2},i_b)\big)}{\sum_{i_b'}\sum_{j_b'}Q_{S|AVU}^{\otimes n}\big(s^n|a^n(m_b,k_{1,b-2}),v^n(j'_b),u^n(m_b,k_{1,b-2},\ell_{b-1},k_{2,b-2},i'_b)\big)}\nonumber\\
    &=g_{\text{\tiny{LE}}}\big(i_b,j_b|m_b,\ell_{b-1},k_{b-2},a^n(m_b,k_{1,b-2}),s_b^n\big);
\end{align}
    \item[$(b)$] follows from since
    \begin{align}
        P_{M_bL_{b-1}K_{1,b-2}K_{2,b-2}|\calC_b^{(n)}}&=\frac{1}{\abs{\calM}\abs{\calL}\abs{\calK_1}\abs{\calK_2}}=\Gamma_{M_bL_{b-1}K_{1,b-2}K_{2,b-2}|\calC_b^{(n)}},\nonumber
    \end{align}
    the codebook $C_N$ is independent of $(M_b,L_{b-1},K_{1,b-2},K_{2,b-2})$, and the symmetry of the codebook construction \ac{wrt} $M_b$, $L_{b-1}$, $K_{1,b-2}$, and $K_{2,b-2}$;
    \item[$(c)$] follows from the law of total expectation.
    \end{itemize}
Now fix the codebook $C_{A_b}^{(n)}=\calC_{A_b}^{(n)}$ and consider the following quantity,
\begin{align}
    &\bbE_{C_{U_b}^{(n)},C_{V_b}^{(n)}|C_{A_b}^{(n)}=\calC_{A_b}^{(n)}}\left[\ToV\left(P_{A^nS_b^n|M_b=1,L_{b-1}=1,K_{b-2}=1,C_b^{(n)}},\Gamma_{A^nS_b^n|M_b=1,L_{b-1}=1,K_{b-2}=1,C_b^{(n)}}\right)\left|C_{A_b}^{(n)}=\calC_{A_b}^{(n)}\right.\right]\nonumber\\
    &\quad\mathop=\limits^{(a)}\bbE_{C_{U_b}^{(n)},C_{V_b}^{(n)}|C_{A_b}^{(n)}=\calC_{A_b}^{(n)}}\left[\ToV\left(P_{S_b^n|M_b=1,L_{b-1}=1,K_{b-2}=1,a^n(1,1),C_b^{(n)}}\right.\right.\nonumber\\
    &\qquad\left.\left.,\Gamma_{S_b^n|M_b=1,L_{b-1}=1,K_{b-2}=1,a^n(1,1),C_b^{(n)}}\right)\left|C_{A_b}^{(n)}=\calC_{A_b}^{(n)}\right.\right]\nonumber\\
    &\quad=\bbE_{C_{U_b}^{(n)},C_{V_b}^{(n)}|C_{A_b}^{(n)}=\calC_{A_b}^{(n)}}\left[\ToV\left(Q_{S|A}^{\otimes n}\big(\cdot|a^n(1,1)\big),\Gamma_{S^n|a^n(1,1),C_b^{(n)}}\right)\left|C_{A_b}^{(n)}=\calC_{A_b}^{(n)}\right.\right],\label{eq:First_Level_TV_NC}
\end{align}where $(a)$ follows since,
\begin{align}
    P_{A^n|M_bL_{b-1}K_{1,b-2}K_{2,b-2}\calC_b^{(n)}}&=\indi{1}_{\{A^n=a^n(M_b,K_{1,b-2})\}}=\Gamma_{A^n|M_bL_{b-1}K_{1,b-2}K_{2,b-2}\calC_b^{(n)}}.\nonumber
\end{align}
By \cite[Theorem~1]{Frey18}, the \ac{RHS} of \eqref{eq:First_Level_TV_NC} vanishes when $n$ grows if
\begin{subequations}\label{eq:Conditional_SCL_NC}
    \begin{align}
    R_I&>\bbI(U;S|A),\label{eq:Conditional_SCL1_NC}\\
    R_J&>\bbI(V;S|A),\label{eq:Conditional_SCL2_NC}\\
    R_I+R_J&>\bbI(U,V;S|A).\label{eq:Conditional_SCL3_NC}
\end{align}
\end{subequations}
Since for every $C_{A_b}^{(n)}=\calC_{A_b}^{(n)}$ the \ac{RHS} of \eqref{eq:General_TV_NC} vanishes when $n$ grows and the total variation distance is non-negative, the expectation over $C_{A_b}^{(n)}$, in \eqref{eq:General_TV_NC}, also vanishes when $n\to\infty$. 
To bound the second term on the \ac{RHS} of \eqref{eq:General_KLD_NC}, by using Pinker's inequality in  Lemma~\ref{lemma:KLD_TV} and considering \eqref{eq:Encoding_GZ_NC}, we have
\begin{align}
    &\bbE_{C_b^{(n)}}\bbD\left(\Gamma_{Z_b^nK_{b-1}L_bK_b|\calC_b^{(n)}}||Q_Z^{\otimes n}P_{K_{b-1}}^UP_{L_b}^UP_{K_b}^U\right)\nonumber\\
    &=\bbE_{C_b^{(n)}}\left[\sum\limits_{(z^n,k_{b-1},\ell_b,k_b)}\Gamma_{Z_b^nK_{b-1}L_bK_b|\calC_b^{(n)}}(z^n,k_{b-1},\ell_b,k_b)\log\frac{\Gamma_{Z_b^nK_{b-1}L_bK_b|\calC_b^{(n)}}(z^n,k_{b-1},\ell_b,k_b)}{Q_Z^{\otimes n}(z^n)P_{K_{b-1}}^U(k_{b-1})P_{L_b}^U(\ell_b)P_{K_b}^U(k_b)}\right]\nonumber\\
    &=\bbE_{C_b^{(n)}}\left[\sum\limits_{(z^n,k_{b-1},\ell_b,k_b)}\sum\limits_{m_b}\sum\limits_{k_{1,b-2}}\sum\limits_{\ell_{b-1}}\sum\limits_{k_{2,b-2}}\sum\limits_{i_b}\sum\limits_{j_b}\frac{1}{2^{n(R+R_J+R_L+R_I+2R_K)}}\times\right.\nonumber\\
    &\quad W_{Z|AVU}^{\otimes n}\Big(z_b^n|a^n(m_b,k_{1,b-2}),v^n(j_b),u^n(m_b,k_{1,b-2},\ell_{b-1},k_{2,b-2},i_b)\Big)\times\nonumber\\
    &\quad\indi{1}_{\big\{\ell_b=\Psi_L\big(v^n(j_b)\big)\big\}\bigcap\big\{k_b=\Psi_K\big(v^n(j_b)\big)\big\}}\times\nonumber\\
    &\quad\left.\log\left[\frac{1}{2^{n(R+R_J+R_L+R_I+R_K-\tilde{R}_L-\tilde{R}_K)}Q_Z^{\otimes n}(z)}\sum\limits_{\tilde{m}_b}\sum\limits_{\tilde{k}_{1,b-2}}\sum\limits_{\tilde{\ell}_{b-1}}\sum\limits_{\tilde{k}_{2,b-2}}\sum\limits_{\tilde{i}_b}\sum\limits_{\tilde{j}_b}\right.\right.\nonumber\\
    &\quad W_{Z|AVU}^{\otimes n}\Big(z_b^n|A^n(\tilde{m}_b,\tilde{k}_{1,b-2}),V^n(\tilde{j}_b),U^n(\tilde{m}_b,\tilde{k}_{1,b-2},\tilde{\ell}_{b-1},\tilde{k}_{2,b-2},\tilde{i}_b)\Big)\times\nonumber\\
    &\quad \left.\left.\indi{1}_{\big\{\ell_b=\Psi_L\big(V^n(\tilde{j}_b)\big)\big\}\bigcap\big\{k_b=\Psi_K\big(V^n(\tilde{j}_b)\big)\big\}}\right]\right]\nonumber\\
    &\mathop\le\limits^{(a)}\frac{1}{2^{n(R+R_J+R_L+R_I+2R_K)}}\sum\limits_{(z^n,k_{b-1},\ell_b,k_b)}\sum\limits_{m_b}\sum\limits_{k_{1,b-2}}\sum\limits_{\ell_{b-1}}\sum\limits_{k_{2,b-2}}\sum\limits_{i_b}\sum\limits_{j_b}\sum\limits_{(a^n,v^n,u^n)}\nonumber\\
    &\quad\Gamma_{ZAVU}^{\otimes n}\Big(z_b^n,a^n(m_b,k_{1,b-2}),v^n(j_b),u^n(m_b,k_{1,b-2},\ell_{b-1},k_{2,b-2},i_b)\Big)\times\nonumber\\
    &\quad\bbE_{\Psi_L\big(v^n(j_b)\big)}\left[\indi{1}_{\big\{\ell_b=\Psi_L\big(v^n(j_b)\big)\big\}}\right]\bbE_{\Psi_K\big(v^n(j_b)\big)}\left[\indi{1}_{\big\{k_b=\Psi_K\big(v^n(j_b)\big)\big\}}\right]\times\nonumber\\
    &\quad\log\bbE_{\mathop {\backslash (m_b,k_{1,b-2},\ell_{b-1},k_{2,b-2},i_b,j_b),}\limits_{\backslash(\Psi_L (v^n(j_b)),\Psi_K (v^n(j_b)))} }\left[\frac{1}{2^{n(R+R_J+R_L+R_I+R_K-\tilde{R}_L-\tilde{R}_K)}Q_Z^{\otimes n}(z)}\times\right.\nonumber\\
    &\quad\sum\limits_{\tilde{m}_b}\sum\limits_{\tilde{k}_{1,b-2}}\sum\limits_{\tilde{\ell}_{b-1}}\sum\limits_{\tilde{k}_{2,b-2}}\sum\limits_{\tilde{i}_b}\sum\limits_{\tilde{j}_b}W_{Z|AVU}^{\otimes n}\Big(z_b^n|A^n(\tilde{m}_b,\tilde{k}_{1,b-2}),V^n(\tilde{j}_b),U^n(\tilde{m}_b,\tilde{k}_{1,b-2},\tilde{\ell}_{b-1},\tilde{k}_{2,b-2},\tilde{i}_b)\Big)\times\nonumber\\
    &\left.\quad\indi{1}_{\big\{\ell_b=\Psi_L\big(V^n(\tilde{j}_b)\big)\big\}\bigcap\big\{k_b=\Psi_K\big(V^n(\tilde{j}_b)\big)\big\}}\right]\nonumber\\
    &\mathop\le\limits^{(b)}\frac{1}{2^{n(R+R_J+R_L+R_I+2R_K)}}\sum\limits_{(z^n,k_{b-1},\ell_b,k_b)}\sum\limits_{m_b}\sum\limits_{k_{1,b-2}}\sum\limits_{\ell_{b-1}}\sum\limits_{k_{2,b-2}}\sum\limits_{i_b}\sum\limits_{j_b}\sum\limits_{(a^n,v^n,u^n)}\nonumber\\
    &\quad\Gamma_{ZAVU}^{\otimes n}\Big(z_b^n,a^n(m_b,k_{1,b-2}),v^n(j_b),u^n(m_b,k_{1,b-2},\ell_{b-1},k_{2,b-2},i_b)\Big)\frac{1}{2^{n(\tilde{R}_L+\tilde{R}_K)}}\times\nonumber\\
    &\quad\log\frac{1}{2^{n(R+R_J+R_L+R_I+R_K-\tilde{R}_L-\tilde{R}_K)}Q_Z^{\otimes n}(z)}\bbE_{\mathop {\backslash (m_b,k_{1,b-2},\ell_{b-1},k_{2,b-2},i_b,j_b),}\limits_{\backslash(\Psi_L (v^n(j_b)),\Psi_K (v^n(j_b)))} }\Bigg[\nonumber\\
    &\quad\left. W_{Z|AVU}^{\otimes n}\Big(z_b^n|a^n(m_b,k_{1,b-2}),v^n(j_b),u^n(m_b,k_{1,b-2},\ell_{b-1},k_{2,b-2},i_b)\Big)\right.+\nonumber\\
    &\quad\sum\limits_{\tilde{j}_b\ne j_b}W_{Z|AVU}^{\otimes n}\Big(z_b^n|a^n(m_b,k_{1,b-2}),V^n(\tilde{j}_b),u^n(m_b,k_{1,b-2},\ell_{b-1},k_{2,b-2},i_b)\Big)\times\nonumber\\
    &\quad\indi{1}_{\big\{\ell_b=\Psi_L\big(V^n(\tilde{j}_b)\big)\big\}\bigcap\big\{k_b=\Psi_K\big(V^n(\tilde{j}_b)\big)\big\}}+\nonumber\\
    &\quad\sum\limits_{(\tilde{\ell}_{b-1},\tilde{k}_{2,b-2},\tilde{i}_b)\ne(\ell_{b-1},k_{2,b-2},i_b)}W_{Z|AVU}^{\otimes n}\Big(z_b^n|a^n(m_b,k_{1,b-2}),v^n(j_b),U^n(m_b,k_{1,b-2},\tilde{\ell}_{b-1},\tilde{k}_{2,b-2},\tilde{i}_b)\Big)+\nonumber\\
    &\quad\sum\limits_{(\tilde{\ell}_{b-1},\tilde{k}_{2,b-2},\tilde{i}_b,\tilde{j}_b)\ne(\ell_{b-1},k_{2,b-2},i_b,j_b)}W_{Z|AVU}^{\otimes n}\Big(z_b^n|a^n(m_b,k_{1,b-2}),V^n(\tilde{j}_b),U^n(m_b,k_{1,b-2},\tilde{\ell}_{b-1},\tilde{k}_{2,b-2},\tilde{i}_b)\Big)\times\nonumber\\
    &\indi{1}_{\big\{\ell_b=\Psi_L\big(V^n(\tilde{j}_b)\big)\big\}\bigcap\big\{k_b=\Psi_K\big(V^n(\tilde{j}_b)\big)\big\}}+\nonumber\\
    &\quad\sum\limits_{(\tilde{m}_b,\tilde{k}_{1,b-2})\ne (m_b,k_{1,b-2})}\Squad\sum\limits_{(\tilde{\ell}_{b-1},\tilde{k}_{2,b-2},\tilde{i}_b)}W_{Z|AVU}^{\otimes n}\Big(z_b^n|A^n(\tilde{m}_b,\tilde{k}_{1,b-2}),v^n(j_b),U^n(\tilde{m}_b,\tilde{k}_{1,b-2},\tilde{\ell}_{b-1},\tilde{k}_{2,b-2},\tilde{i}_b)\Big)+\nonumber\\
    &\sum\limits_{(\tilde{m}_b,\tilde{k}_{1,b-2},\tilde{j}_b)\ne (m_b,k_{1,b-2},j_b)}\Squad\sum\limits_{(\tilde{\ell}_{b-1},\tilde{k}_{2,b-2},\tilde{i}_b)}\hspace{-5mm}W_{Z|AVU}^{\otimes n}\Big(z_b^n|A^n(\tilde{m}_b,\tilde{k}_{1,b-2}),V^n(\tilde{j}_b),U^n(\tilde{m}_b,\tilde{k}_{1,b-2},\tilde{\ell}_{b-1},\tilde{k}_{2,b-2},\tilde{i}_b)\Big)\times\nonumber\\
    &\quad\left.\indi{1}_{\big\{\ell_b=\Psi_L\big(V^n(\tilde{j}_b)\big)\big\}\bigcap\big\{k_b=\Psi_K\big(V^n(\tilde{j}_b)\big)\big\}}\right]\nonumber\\
    &\le\frac{1}{2^{n(R+R_J+R_L+R_I+2R_K+\tilde{R}_L+\tilde{R}_K)}}\sum\limits_{(z^n,k_{b-1},\ell_b,k_b)}\sum\limits_{m_b}\sum\limits_{k_{1,b-2}}\sum\limits_{\ell_{b-1}}\sum\limits_{k_{2,b-2}}\sum\limits_{i_b}\sum\limits_{j_b}\sum\limits_{(a^n,v^n,u^n)}\nonumber\\
    &\quad\Gamma_{ZAVU}^{\otimes n}\Big(z_b^n,a^n(m_b,k_{1,b-2}),v^n(j_b),u^n(m_b,k_{1,b-2},\ell_{b-1},k_{2,b-2},i_b)\Big)\times\nonumber\\
    &\quad\log\frac{1}{2^{n(R+R_J+R_L+R_I+R_K-\tilde{R}_L-\tilde{R}_K)}Q_Z^{\otimes n}(z)}\Bigg[\nonumber\\
    &\quad W_{Z|AVU}^{\otimes n}\Big(z_b^n|a^n(m_b,k_{1,b-2}),v^n(j_b),u^n(m_b,k_{1,b-2},\ell_{b-1},k_{2,b-2},i_b)\Big)+\nonumber\\
    &\quad\sum\limits_{\tilde{j}_b\ne j_b}W_{Z|AU}^{\otimes n}\Big(z_b^n|a^n(m_b,k_{1,b-2}),u^n(m_b,k_{1,b-2},\ell_{b-1},k_{2,b-2},i_b)\Big)2^{-n(\tilde{R}_L+\tilde{R}_K)}+\nonumber\\
    &\quad\sum\limits_{(\tilde{\ell}_{b-1},\tilde{k}_{2,b-2},\tilde{i}_b)\ne(\ell_{b-1},k_{2,b-2},i_b)}W_{Z|AV}^{\otimes n}\Big(z_b^n|a^n(m_b,k_{1,b-2}),v^n(j_b)\Big)+\nonumber\\
    &\quad\sum\limits_{(\tilde{\ell}_{b-1},\tilde{k}_{2,b-2},\tilde{i}_b,\tilde{j}_b)\ne(\ell_{b-1},k_{2,b-2},i_b,j_b)}W_{Z|A}^{\otimes n}\Big(z_b^n|a^n(m_b,k_{1,b-2})\Big)2^{-n(\tilde{R}_L+\tilde{R}_K)}+\nonumber\\
    &\quad\left.\sum\limits_{(\tilde{m}_b,\tilde{k}_{1,b-2})\ne (m_b,k_{1,b-2})}\Squad\sum\limits_{(\tilde{\ell}_{b-1},\tilde{k}_{2,b-2},\tilde{i}_b)}W_{Z|V}^{\otimes n}\Big(z_b^n|v^n(j_b)\Big)+1\right]\nonumber\\
    &\le\frac{1}{2^{n(R+R_J+R_L+R_I+2R_K+\tilde{R}_L+\tilde{R}_K)}}\sum\limits_{(z^n,k_{b-1},\ell_b,k_b)}\sum\limits_{m_b}\sum\limits_{k_{1,b-2}}\sum\limits_{\ell_{b-1}}\sum\limits_{k_{2,b-2}}\sum\limits_{i_b}\sum\limits_{j_b}\sum\limits_{(a^n,v^n,u^n)}\nonumber\\
    &\quad\Gamma_{ZAVU}^{\otimes n}\Big(z_b^n,a^n(m_b,k_{1,b-2}),v^n(j_b),u^n(m_b,k_{1,b-2},\ell_{b-1},k_{2,b-2},i_b)\Big)\times\nonumber\\
    &\quad\log\frac{1}{2^{n(R+R_J+R_L+R_I+R_K-\tilde{R}_L-\tilde{R}_K)}Q_Z^{\otimes n}(z)}\Bigg[\nonumber\\
    &\quad W_{Z|AVU}^{\otimes n}\Big(z_b^n|a^n(m_b,k_{1,b-2}),v^n(j_b),u^n(m_b,k_{1,b-2},\ell_{b-1},k_{2,b-2},i_b)\Big)+\nonumber\\
    &\quad 2^{n(R_J-\tilde{R}_L-\tilde{R}_K)}W_{Z|AU}^{\otimes n}\Big(z_b^n|a^n(m_b,k_{1,b-2}),u^n(m_b,k_{1,b-2},\ell_{b-1},k_{2,b-2},i_b)\Big)+\nonumber\\
    &\quad2^{n(R_L+R_{K_2}+R_I)}W_{Z|AV}^{\otimes n}\Big(z_b^n|a^n(m_b,k_{1,b-2}),v^n(j_b)\Big)+\nonumber\\
    &\quad2^{n(R_L+R_{K_2}+R_I+R_J-\tilde{R}_L-\tilde{R}_K)}W_{Z|A}^{\otimes n}\Big(z_b^n|a^n(m_b,k_{1,b-2})\Big)+\nonumber\\
    &\quad2^{n(R+R_{K_1}+R_L+R_{K_2}+R_I)}W_{Z|V}^{\otimes n}\Big(z_b^n|v^n(j_b)\Big)+1\Bigg]\nonumber\\
    &\le\frac{1}{2^{n(R+R_J+R_L+R_I+2R_K+\tilde{R}_L+\tilde{R}_K)}}\sum\limits_{(z^n,k_{b-1},\ell_b,k_b)}\sum\limits_{m_b}\sum\limits_{k_{1,b-2}}\sum\limits_{\ell_{b-1}}\sum\limits_{k_{2,b-2}}\sum\limits_{i_b}\sum\limits_{j_b}\sum\limits_{(a^n,v^n,u^n)}\nonumber\\
    &\quad\Gamma_{ZAVU}^{\otimes n}\Big(z_b^n,a^n(m_b,k_{1,b-2}),v^n(j_b),u^n(m_b,k_{1,b-2},\ell_{b-1},k_{2,b-2},i_b)\Big)\times\nonumber\\
    &\quad\log\left[\frac{W_{Z|AVU}^{\otimes n}\Big(z_b^n|a^n(m_b,k_{1,b-2}),v^n(j_b),u^n(m_b,k_{1,b-2},\ell_{b-1},k_{2,b-2},i_b)\Big)}{2^{n(R+R_J+R_L+R_I+R_K-\tilde{R}_L-\tilde{R}_K)}Q_Z^{\otimes n}(z)}\right.+\nonumber\\
    &\quad\frac{W_{Z|AU}^{\otimes n}\Big(z_b^n|a^n(m_b,k_{1,b-2}),u^n(m_b,k_{1,b-2},\ell_{b-1},k_{2,b-2},i_b)\Big)}{2^{n(R+R_L+R_I+R_K)}Q_Z^{\otimes n}(z)}+\nonumber\\
    &\quad\left.\frac{W_{Z|AV}^{\otimes n}\Big(z_b^n|a^n(m_b,k_{1,b-2}),v^n(j_b)\Big)}{2^{n(R+R_J+R_{K_1}-\tilde{R}_L-\tilde{R}_K)}Q_Z^{\otimes n}(z)}+\frac{W_{Z|A}^{\otimes n}\Big(z_b^n|a^n(m_b,k_{1,b-2})\Big)}{2^{n(R+R_{K_1})}Q_Z^{\otimes n}(z)}+\frac{W_{Z|V}^{\otimes n}\Big(z_b^n|v^n(j_b)\Big)}{2^{n(R_J-\tilde{R}_L-\tilde{R}_K)}Q_Z^{\otimes n}(z)}+1\right]\nonumber\\
    &\mathop=\limits^{(c)}\Delta_1+\Delta_2\label{eq:Bounding_General_KLD_NC}
\end{align}where
\begin{itemize}
    \item[$(a)$] follows from Jensen's inequality;
    \item[$(b)$] follows by expanding the summation in the argument of the $\log$ function and considering $\indi{1}_{\{\cdot\}}\le1$;
    \item[$(c)$] follows by defining $\Delta_1$ and $\Delta_2$ as
\end{itemize}
\begin{align}
    \Delta_1&\triangleq\frac{1}{2^{n(R+R_J+R_L+R_I+2R_K+\tilde{R}_L+\tilde{R}_K)}}\sum\limits_{(k_{b-1},\ell_b,k_b)}\sum\limits_{m_b}\sum\limits_{k_{1,b-2}}\sum\limits_{\ell_{b-1}}\sum\limits_{k_{2,b-2}}\sum\limits_{i_b}\sum\limits_{j_b}\nonumber\\
    &\sum\limits_{(z^n,a^n(m_b,k_{1,b-2}),v^n(j_b),u^n(m_b,k_{1,b-2},\ell_{b-1},k_{2,b-2},i_b))\in\calT_\epsilon^{(n)}}\hspace{-25mm}\Gamma_{ZAVU}^{\otimes n}\Big(z_b^n,a^n(m_b,k_{1,b-2}),v^n(j_b),u^n(m_b,k_{1,b-2},\ell_{b-1},k_{2,b-2},i_b)\Big)\times\nonumber\\
    &\quad\log\left[\frac{W_{Z|AVU}^{\otimes n}\Big(z_b^n|a^n(m_b,k_{1,b-2}),v^n(j_b),u^n(m_b,k_{1,b-2},\ell_{b-1},k_{2,b-2},i_b)\Big)}{2^{n(R+R_J+R_L+R_I+R_K-\tilde{R}_L-\tilde{R}_K)}Q_Z^{\otimes n}(z)}\right.+\nonumber\\
    &\quad\frac{W_{Z|AU}^{\otimes n}\Big(z_b^n|a^n(m_b,k_{1,b-2}),u^n(m_b,k_{1,b-2},\ell_{b-1},k_{2,b-2},i_b)\Big)}{2^{n(R+R_L+R_I+R_K)}Q_Z^{\otimes n}(z)}+\nonumber\\
    &\quad\left.\frac{W_{Z|AV}^{\otimes n}\Big(z_b^n|a^n(m_b,k_{1,b-2}),v^n(j_b)\Big)}{2^{n(R+R_J+R_{K_1}-\tilde{R}_L-\tilde{R}_K)}Q_Z^{\otimes n}(z)}+\frac{W_{Z|A}^{\otimes n}\Big(z_b^n|a^n(m_b,k_{1,b-2})\Big)}{2^{n(R+R_{K_1})}Q_Z^{\otimes n}(z)}+\frac{W_{Z|V}^{\otimes n}\Big(z_b^n|v^n(j_b)\Big)}{2^{n(R_J-\tilde{R}_L-\tilde{R}_K)}Q_Z^{\otimes n}(z)}+1\right]\nonumber\\
    &\le\frac{1}{2^{n(R+R_J+R_L+R_I+2R_K+\tilde{R}_L+\tilde{R}_K)}}\sum\limits_{(k_{b-1},\ell_b,k_b)}\sum\limits_{m_b}\sum\limits_{k_{1,b-2}}\sum\limits_{\ell_{b-1}}\sum\limits_{k_{2,b-2}}\sum\limits_{i_b}\sum\limits_{j_b}\nonumber\\
    &\sum\limits_{(z^n,a^n(m_b,k_{1,b-2}),v^n(j_b),u^n(m_b,k_{1,b-2},\ell_{b-1},k_{2,b-2},i_b))\in\calT_\epsilon^{(n)}}\hspace{-25mm}\Gamma_{ZAVU}^{\otimes n}\Big(z_b^n,a^n(m_b,k_{1,b-2}),v^n(j_b),u^n(m_b,k_{1,b-2},\ell_{b-1},k_{2,b-2},i_b)\Big)\times\nonumber\\
    &\quad\log\left[\frac{2^{-n(1-\epsilon)\bbH(Z|A,V,U)}}{2^{n(R+R_J+R_L+R_I+R_K-\tilde{R}_L-\tilde{R}_K)}2^{-n(1+\epsilon)\bbH(Z)}}\right.+\frac{2^{-n(1-\epsilon)\bbH(Z|A,U)}}{2^{n(R+R_L+R_I+R_K)}2^{-n(1+\epsilon)\bbH(Z)}}+\nonumber\\
    &\quad\left.\frac{2^{-n(1-\epsilon)\bbH(Z|A,V)}}{2^{n(R+R_J+R_{K_1}-\tilde{R}_L-\tilde{R}_K)}2^{-n(1+\epsilon)\bbH(Z)}}+\frac{2^{-n(1-\epsilon)\bbH(Z|A)}}{2^{n(R+R_{K_1})}2^{-n(1+\epsilon)\bbH(Z)}}+\frac{2^{-n(1-\epsilon)\bbH(Z|V)}}{2^{n(R_J-\tilde{R}_L-\tilde{R}_K)}2^{-n(1+\epsilon)\bbH(Z)}}+1\right],\label{eq:Si1_NC}\\
    \Delta_2&\triangleq\frac{1}{2^{n(R+R_J+R_L+R_I+2R_K+\tilde{R}_L+\tilde{R}_K)}}\sum\limits_{(k_{b-1},\ell_b,k_b)}\sum\limits_{m_b}\sum\limits_{k_{1,b-2}}\sum\limits_{\ell_{b-1}}\sum\limits_{k_{2,b-2}}\sum\limits_{i_b}\sum\limits_{j_b}\nonumber\\
    &\sum\limits_{(z^n,a^n(m_b,k_{1,b-2}),v^n(j_b),u^n(m_b,k_{1,b-2},\ell_{b-1},k_{2,b-2},i_b))\notin\calT_\epsilon^{(n)}}\hspace{-25mm}\Gamma_{ZAVU}^{\otimes n}\Big(z_b^n,a^n(m_b,k_{1,b-2}),v^n(j_b),u^n(m_b,k_{1,b-2},\ell_{b-1},k_{2,b-2},i_b)\Big)\times\nonumber\\
    &\quad\log\left[\frac{W_{Z|AVU}^{\otimes n}\Big(z_b^n|a^n(m_b,k_{1,b-2}),v^n(j_b),u^n(m_b,k_{1,b-2},\ell_{b-1},k_{2,b-2},i_b)\Big)}{2^{n(R+R_J+R_L+R_I+R_K-\tilde{R}_L-\tilde{R}_K)}Q_Z^{\otimes n}(z)}\right.+\nonumber\\
    &\quad\frac{W_{Z|AU}^{\otimes n}\Big(z_b^n|a^n(m_b,k_{1,b-2}),u^n(m_b,k_{1,b-2},\ell_{b-1},k_{2,b-2},i_b)\Big)}{2^{n(R+R_L+R_I+R_K)}Q_Z^{\otimes n}(z)}+\nonumber\\
    &\quad\left.\frac{W_{Z|AV}^{\otimes n}\Big(z_b^n|a^n(m_b,k_{1,b-2}),v^n(j_b)\Big)}{2^{n(R+R_J+R_{K_1}-\tilde{R}_L-\tilde{R}_K)}Q_Z^{\otimes n}(z)}+\frac{W_{Z|A}^{\otimes n}\Big(z_b^n|a^n(m_b,k_{1,b-2})\Big)}{2^{n(R+R_{K_1})}Q_Z^{\otimes n}(z)}+\frac{W_{Z|V}^{\otimes n}\Big(z_b^n|v^n(j_b)\Big)}{2^{n(R_J-\tilde{R}_L-\tilde{R}_K)}Q_Z^{\otimes n}(z)}+1\right]\nonumber\\
    &\le2\abs{A}\abs{U}\abs{V}\abs{Z}e^{-n\epsilon^2\mu_{A,V,U,Z}}n\log\left(\frac{5}{\mu_Z}+1\right),
\end{align}where $\mu_{A,U,V,Z}=\min\limits_{(a,u,v,z)\in(\calA,\calU,\calV,\calZ)}\Gamma_{AUVZ}(a,u,v,z)$ and $\mu_Z=\min\limits_{z\in\calZ}\Gamma_{Z}(z)$.
When $n\to\infty$ then $\Delta_2\to 0$, and $\Delta_1\to 0$ when
\begin{subequations}\label{eq:resol_2}
\begin{align}
    R+R_J+R_L+R_I+R_K-\tilde{R}_L-\tilde{R}_K&>\bbI(A,V,U;Z),\\
    R+R_L+R_I+R_K&>\bbI(A,U;Z),\\
    R+R_J+R_{K_1}-\tilde{R}_L-\tilde{R}_K&>\bbI(A,V;Z),\\
    R+R_{K_1}&>\bbI(A;Z),\\
    R_J-\tilde{R}_L-\tilde{R}_K&>\bbI(V;Z).
\end{align}
\end{subequations}

\subsection{Decoding}
The following lemma is essential to analyze the probability of error.
\begin{lemma}[Typicality]
\label{lemma:Typicaity}
If $(R_K,R_L,R_I,R_J)\in\bbR_+^4$ satisfy the constraint in \eqref{eq:Conditional_SCL_NC} then for any $m_b\in\calM$ and $\epsilon>0$ we have
\begin{align}
    \bbE_{C_b^{(n)}}\bbP_P\big[\big(A^n(m_b,k_{1,b-2}),S_b^n,U^n(m_b,k_{1,b-2},L_{b-1},K_{2,b-2},I_b),V^n(J_b)\big)\notin\calT_\epsilon^{(n)}\big]\xrightarrow[]{n\to\infty}0.\label{eq:lemma_typicality_NC}
\end{align}
\end{lemma}Lemma~\ref{lemma:Typicaity} is proved in Appendix~\ref{app:Typicality_Proof}. After receiving $Y_b^n$, given the shared key $k_{b-2}\triangleq(k_{1,b-2},k_{2,b-2})$, the decoder looks for the smallest value of $(\hat{m}_b,\hat{\ell}_{b-1},\hat{j}_b)$ for which there exist $\hat{i}_b$ such that $\left(A^n(\hat{m}_b,k_{1,b-2}),U^n\big(\hat{m}_b,k_{1,b-2},\hat{\ell}_{b-1},k_{2,b-2},\hat{i}_b\big),V^n(\hat{j}_b),Y_b^n\right)\in\calT_{\epsilon}^{(n)}(P_{AUVY})$ and choose $(\hat{M}_b,\hat{L}_{b-1},\hat{J}_b)=(1,1,1)$ if such a $(\hat{m}_b,\hat{\ell}_{b-1},\hat{j}_b)$ does not exist. Let,
\begin{subequations}
    \begin{align}
    \calE&\triangleq\left\{M\ne\hat{M}\right\},\label{eq:General_NC_Pe}\\
    \calE_b&\triangleq\left\{M_b\ne\hat{M}_b\right\},\label{eq:General_NC_Pe_b}\\
    \calE_{1,b}&\triangleq\left\{\left(A^n(\hat{m}_b,k_{1,b-2}),U^n\big(\hat{m}_b,k_{1,b-2},\hat{\ell}_{b-1},k_{2,b-2},\hat{i}_b\big),S_b^n\right)\notin\calT_{\epsilon}^{(n)}(P_{AUS}),\Squad\text{for all}\Squad\hat{i}_b\right\},\label{eq:General_NC_Enc}\\
    \calE_{2,b}&\triangleq\left\{\left(A^n(\hat{m}_b,k_{1,b-2}),U^n\big(\hat{m}_b,k_{1,b-2},\hat{\ell}_{b-1},k_{2,b-2},\hat{i}_b\big),Y_b^n\right)\notin\calT_{\epsilon}^{(n)}(P_{AUY})\right\},\label{eq:General_NC_Dec1}\\
    \calE_{3,b}&\triangleq\left\{\left(A^n(\hat{m}_b,k_{1,b-2}),U^n\big(\hat{m}_b,k_{1,b-2},\hat{\ell}_{b-1},k_{2,b-2},\hat{i}_b\big),Y_b^n\right)\in\calT_{\epsilon}^{(n)}(P_{AUY}),\Squad\text{for some}\right.\nonumber\\
    &\qquad\qquad\left.\big(\hat{m}_b,\hat{\ell}_b\big)\ne (M_b,L_b)\Squad\text{and}\Squad \hat{i}_b\in\calI\right\}.\label{eq:General_NC_Dec2}
\end{align}
\end{subequations}Now, we bound the probability of error as
\begin{align}
    P_e^{(N)}=\bbP(\calE)=\bbP\left(\bigcup_{b=1}^B\calE_b\right)\le\sum\limits_{b=1}^B\bbP(\calE_b),\label{eq:Total_PE_NC}
\end{align}where the inequality follows from the union bound. Next for each $b\in[B]$ we bound $\bbP(\calE_b)$ as follows,
\begin{align}
    \bbP(\calE_b)\le\bbP(\calE_{1,b})+\bbP(\calE_{1,b}^c\cap\calE_{2,b})+\bbP(\calE_{1,b}^c\cap\calE_{2,b}^c\cap\calE_{3,b}),\label{eq:Union_Bound_NC}
\end{align}from Lemma~\ref{lemma:Typicaity} the first term on the \ac{RHS} of \eqref{eq:Union_Bound_NC} vanishes when $n$ grows to infinity, the second on the \ac{RHS} of \eqref{eq:Union_Bound_NC} vanishes when $n$ grows by the law of large numbers, by the union bound and \cite[Theorem~1.3]{Kramer_Book}, the last term on the \ac{RHS} of \eqref{eq:Union_Bound_NC} vanishes when $n$ grows if,
\begin{align}
    R+R_L+R_I<\bbI(U;Y).\label{eq:Dec_Constraint_NC}
\end{align}
Next, we bound the error probability for the key generation. Let $L_{b-1}$ and $J_{b-1}$ denote the indices that are chosen by the transmitter, and $\hat{L}_{b-1}$ and $\hat{J}_{b-1}$ denote the estimate of these indices at the receiver. The receiver by decoding $U^n$ at the end of block $b$ knows the index $\hat{L}_{b-1}$. 
Let,
\begin{subequations}
\begin{align}
    \tilde{\calE}&\triangleq\left\{\left(V^n(\hat{J}_{b-1}),S_{b-1}^n,A_{b-1}^n,U_{b-1}^n,Y_{b-1}^n\right)\notin\calT_\epsilon^{(n)}\right\},\label{eq:Key_Gen_Err1_NC}\\
    \tilde{\calE}_1&\triangleq\left\{\left(V^n(j_{b-1}),A_{b-1}^n,S_{b-1}^n\right)\notin\calT_{\tilde{\epsilon}}^{(n)},\quad\text{for all}\Squad j_{b-1}\in\left[2^{nR_J}\right]\right\},\label{eq:Key_Gen_Err2_NC}\\
    \tilde{\calE}_2&\triangleq\left\{\left(V^n(J_{b-1}),S_{b-1}^n,A_{b-1}^n,U_{b-1}^n,Y_{b-1}^n\right)\notin\calT_{\tilde{\epsilon}}^{(n)}\right\},\label{eq:Key_Gen_Err3_NC}\\
    \tilde{\calE}_3&\triangleq\left\{\left(V^n(\hat{j}_{b-1}),A_{b-1}^n,U_{b-1}^n,Y_{b-1}^n\right)\in\calT_{\tilde{\epsilon}}^{(n)},\Squad\text{for some}\Squad\hat{j}_b\in\calB(\hat{L}_{b-1}),\hat{j}_b\ne J_{b-1}\right\}.\label{eq:Key_Gen_Err4_NC}
\end{align}where $\epsilon>\tilde{\epsilon}>0$. 
\end{subequations}Now by the union bound, 
\begin{align}
    \bbP\big(\tilde{\calE}\big)\le\bbP\big(\tilde{\calE}_1\big)+\bbP\big(\tilde{\calE}_1^c\cap\tilde{\calE}_2\big)+\bbP\big(\tilde{\calE}_3\big).\label{eq:Key_Gen_UB_NC}
\end{align}By Lemma~\ref{lemma:Typicaity} the first term on the \ac{RHS} of \eqref{eq:Key_Gen_UB_NC} vanishes when $n$ grows to infinity if \eqref{eq:Conditional_SCL_NC} holds, similar to \cite[Sec.~11.3.1]{ElGamalKim} the second and the third terms on the \ac{RHS} of \eqref{eq:Key_Gen_UB_NC} go to zero when $n\to\infty$ if,
\begin{subequations}\label{eq:WZ_Dec_NC}
\begin{align}
    R_J&>\bbI(V;S|A),\label{eq:WZ_1_NC}\\
    R_J-R_L&<\bbI(V;A,U,Y),\label{eq:WZ_2_NC}
\end{align}
\end{subequations}
Applying Fourier-Motzkin elimination procedure \cite{ElGamalKim,FMEIT} to eliminate $(R_I,R_J,R_L,R_{K_1},R_{K_2},R_K,\tilde{R}_L,\tilde{R}_K)$, in \eqref{eq:Conditional_SCL_NC}, \eqref{eq:resol_2}, \eqref{eq:Dec_Constraint_NC}, and \eqref{eq:WZ_Dec_NC} and considering $\tilde{R}_L+\tilde{R}_K\ge R_L+R_K$ leads to
\begin{subequations}\label{eq:FME_Output}
\begin{align}
R&>\bbI(A;Z)+\bbI(V;Z)-\bbI(V;A,U,Y),\label{eq:FME_Output_1}\\
R&>\bbI(A,V;Z)-\bbI(V;A,U,Y),\label{eq:FME_Output_2}\\
0&>\bbI(V;Z)-\bbI(V;A,U,Y),\label{eq:FME_Output_3}\\
0&>\bbI(A,U;Z)+\bbI(V;Z)-\bbI(A,U;Y)-\bbI(V;A,U,Y),\label{eq:FME_Output_4}\\
0&>\bbI(A,V,U;Z)-\bbI(A,U;Y)-\bbI(V;A,U,Y),\label{eq:FME_Output_5}\\
R&<-\bbI(U,V;S|A)+\bbI(A,U;Y)+\bbI(V;A,U,Y),\label{eq:FME_Output_6}\\
R&<-\bbI(U;S|A)+\bbI(A,U;Y),\label{eq:FME_Output_7}\\
R&<-\bbI(U;S|A)-\bbI(V;S|A)+\bbI(A,U;Y)+\bbI(V;A,U,Y),\label{eq:FME_Output_8}
\end{align}
where, since $\bbI(U,V;S|A)\ge\bbI(U;S|A)+\bbI(V;S|A)$, \eqref{eq:FME_Output_8} is redundant because of \eqref{eq:FME_Output_6}, since $\bbI(A,U,V;Z)\ge\bbI(A,U;Z)+\bbI(V;Z)$, \eqref{eq:FME_Output_4} is redundant because of \eqref{eq:FME_Output_5}, and since $\bbI(A,V;Z)\ge\bbI(A;Z)+\bbI(V;Z)$, \eqref{eq:FME_Output_1} is redundant because of \eqref{eq:FME_Output_2}. Therefore, one can rewrite the rate constraints in \eqref{eq:FME_Output} as the rate constraints in Theorem~\ref{thm:Acievability_KG}. 
\end{subequations}

\section{Proof of Lemma~\ref{lemma:Typicaity}}
\label{app:Typicality_Proof}
Fix $\epsilon>0$ and consider the distribution $\Gamma$ that we considered in \eqref{eq:Encoding_Ideal_Joint_NC}. Then we have
\begin{align}
    &\bbE_{C_b^{(n)}}\bbP_\Gamma\big[\big(A^n(m_b,k_{1,b-2}),S_b^n,U^n(m_b,k_{1,b-2},L_{b-1},K_{2,b-2},I_b),V^n(J_b)\big)\notin\calT_\epsilon^{(n)}\big]
    \xrightarrow[]{n\to\infty}0,\label{eq:lemma_typicality_Gamma_NC}
\end{align}this follows from the law of large numbers, since $S_b^n$ is the output of a \ac{DMC} with inputs $A^n(m_b,k_{1,b-2})$, $V^n(J_b)$, and $U^n(m_b,k_{1,b-2},L_{b-1},K_{2,b-2},I_b)$. Furthermore, from \eqref{eq:General_TV_NC} we have
\begin{align}
    &\bbE_{C_b^{(n)}}\ToV\left(P_{Z^n|C_b^{(n)}},\Gamma_{Z^n|C_b^{(n)}}\right)\le\nonumber\\
    &\bbE_{C_b^{(n)}}\ToV\left(P_{M_bK_{1,b-2}L_{b-1}K_{2,b-2}A^nS_b^nI_bJ_bU^nV^nX^nY_b^nZ_b^nK_{b-1}L_bK_b|C_b^{(n)}},\right.\nonumber\\
    &\left.\Gamma_{M_bK_{1,b-2}L_{b-1}K_{2,b-2}A^nS_b^nI_bJ_bU^nV^nX^nY_b^nZ_b^nK_{b-1}L_bK_b|C_b^{(n)}}\right)\xrightarrow[]{n\to\infty}0.\label{eq:General_Typicality_Distance_NC}
\end{align}Now we define $F_n:\calA^n\times\calS^n\times\calU^n\times\calV^n\to\bbR$ as $F_n(A^n,S^n,U^n,V^n)\triangleq\indi{1}_{\big\{\big(A^n(m_b,k_{1,b-2}),S_b^n,U^n(m_b,k_{1,b-1},L_{b-1},K_{2,b-1},I_b),V^n(J_b)\big)\notin}$ $_{\calT_\epsilon^{(n)}(P_{ASUV})\big\}}$ and consider
\begin{align}
    &\bbE_{C_b^{(n)}}\bbP_P\big[\big(A^n(m_b,k_{1,b-1}),S_b^n,U^n(m_b,k_{1,b-2},L_{b-1},K_{2,b-1},I_b),V^n(J_b)\big)\notin\calT_\epsilon^{(n)}\big]\nonumber\\
    &=\bbE_{C_b^{(n)}}\bbP_P\left[F_n\big(A^n,S^n,U^n,V^n\big)\right]\nonumber\\
    &\le\bbE_{C_b^{(n)}}\bbP_\Gamma\left[F_n\big(A^n,S^n,U^n,V^n\big)\right]+\bbE_{C_b^{(n)}}\Big|\bbP_P\left[F_n\big(A^n,S^n,U^n,V^n\big)\right]-\bbP_\Gamma\left[F_n\big(A^n,S^n,U^n,V^n\big)\right]\Big|\nonumber\\
    &\mathop\le\limits^{(a)}\bbE_{C_b^{(n)}}\bbP_\Gamma\left[F_n\big(A^n,S^n,U^n,V^n\big)\right]+\bbE_{C_b^{(n)}}\ToV\left(P_{M_bK_{1,b-2}L_{b-1}K_{2,b-2}A^nS_b^nI_bJ_bU^nV^nX^nY_b^nZ_b^nK_{b-1}L_bK_b|C_b^{(n)}}\right.\nonumber\\
    &\left.,\Gamma_{M_bK_{1,b-2}L_{b-1}K_{2,b-2}A^nS_b^nI_bJ_bU^nV^nX^nY_b^nZ_b^nK_{b-1}L_bK_b|C_b^{(n)}}\right),\label{eq:Final_Typicality_Term_NC}
\end{align}where $(a)$ follows from \cite[Property~1]{Likelihood_Encoder}, by choosing the constant coefficient to be equal to 1, for any $n\in\bbN_*$. From \eqref{eq:lemma_typicality_Gamma_NC} and \eqref{eq:General_Typicality_Distance_NC}, the first and the second term on the \ac{RHS} of \eqref{eq:Final_Typicality_Term_NC}, respectively, vanishes as $n$ grows.

\section{Proof of Theorem~\ref{thm:Converse_NC}}
\label{proof:thm:Converse_NC}
Consider any sequence of codes with length $N$ for channels with \ac{ADSI} when the state is available non-causally at the transmitter such that $P_e^{(N)}\le\epsilon_N$, $\bbD\left(P_{Z^N}||Q_0^{\otimes N}\right)\le\tilde{\epsilon}$, and $\bar{R}_K^{(N)}/N\triangleq\eta_N$, where $\epsilon_N\xrightarrow[]{N\to\infty}0$ and $\eta_N\xrightarrow[]{N\to\infty}0$. The following lemma from \cite[Lemma~3]{Keyless22} is essential in the converse proof.
\begin{lemma}{(\hspace{-0.15mm}\cite[Lemma~3]{Keyless22})}
\label{lemma:iid_Time_indep}
If $\bbD(P_{Z^N}||Q_0^{\otimes N})\le\epsilon$, then $\sum_{t=1}^N\bbI(Z_t;Z^{t-1})\le\epsilon$ and $\sum_{t=1}^N\bbI(Z_t;Z_{t+1}^N)\le\epsilon$. Also, when $T\in[N]$ is a uniformly distributed \ac{RV} and independent of the other involved \acp{RV}, then $\bbI(T;Z_T)\le\gamma$, where $\gamma\triangleq\frac{\epsilon}{N}$.
\end{lemma}
\subsection{Epsilon Rate}We first define $\calF^{(\epsilon)}_{\text{U-NC}}$ for $\epsilon>0$ which expands the rate defined in \eqref{eq:Converse_AD_NC} as
\begin{subequations}\label{eq:Converse_AD_epsilon_NC}
\begin{align}
  \calF^{(\epsilon)}_{\text{U-NC}} = \left.\begin{cases}R\geq 0: \exists P_{ASUVXYZ}\in\calG^{(\epsilon)}_{\text{U-NC}}:\\
  R\le\min\big\{\bbI(U;Y) - \bbI(U;S|A),\bbI(U,V;Y)-\bbI(U;S|V,A),\bbI(V;A,U,Y)\big\}+\epsilon\\
\end{cases}\right\},\label{eq:Converse_A_epsilon_NC}
\end{align}
where
\begin{align}
  \calG^{(\epsilon)}_{\text{U-NC}}\triangleq \left.\begin{cases}P_{ASUVXYZ}:\\
P_{ASUVXYZ}=P_AQ_{S|A}P_{U|AS}P_{V|US}P_{X|US}W_{YZ|XS}\\
\min\big\{\bbI(U;Y) - \bbI(U;S|A),\bbI(U,V;Y)-\bbI(U;S|V,A),\bbI(V;A,U,Y)\big\}\ge\\
\quad\bbI(V;Z)-\bbI(V;S)-\epsilon\\
\bbD\big(P_Z||Q_0\big)\le\epsilon
\end{cases}\right\}.\label{eq:Converse_D_epsilon_NC}
\end{align}
\end{subequations}Now we show that any achievable rate $R$ belongs to $\calF^{(\epsilon)}_{\text{U-NC}}$, i.e., $R\in\calF^{(\epsilon)}_{\text{U-NC}}$. 
For any $\gamma>0$, and $\epsilon_N>0$ we have,
\begin{align}
    NR&=\bbH(M)\nonumber\\
    &=\bbH(M|K)\nonumber\\
    &\mathop\le\limits^{(a)}\bbI\big(M;Y^N|K\big)+N\epsilon_N\nonumber\\
    &=\sum\limits_{t=1}^N\big[\bbI\big(M;Y_t|K,Y^{t-1}\big)\big]+N\epsilon_N\nonumber\\
    &\le\sum\limits_{t=1}^N\big[\bbI\big(M,K,A_t^N;Y_t|Y^{t-1}\big)\big]+N\epsilon_N\nonumber\\
    &=\sum\limits_{t=1}^N\Big[\bbI\big(M,K,A_t^N,S_{t+1}^N;Y_t|Y^{t-1}\big)-\bbI\big(S_{t+1}^N;Y_t|M,K,A_t^N,Y^{t-1}\big)\Big]+N\epsilon_N\nonumber\\
    &\mathop=\limits^{(b)}\sum\limits_{t=1}^N\Big[\bbI\big(M,K,A_t^N,S_{t+1}^N;Y_t|Y^{t-1}\big)-\bbI\big(Y^{t-1};S_t|M,K,A_t^N,S_{t+1}^N\big)\Big]+N\epsilon_N\nonumber\\
    &\le\sum\limits_{t=1}^N\Big[\bbI\big(M,K,A_t^N,S_{t+1}^N,Y^{t-1};Y_t\big)-\bbH\big(S_t|M,K,A_t^N,S_{t+1}^N\big)+\nonumber\\
    &\qquad\bbH\big(S_t|M,K,A_t^N,S_{t+1}^N,Y^{t-1}\big)\Big]+N\epsilon_N\nonumber\\
    &\mathop=\limits^{(c)}\sum\limits_{t=1}^N\Big[\bbI\big(M,K,A_t^N,S_{t+1}^N,Y^{t-1};Y_t\big)-\bbH\big(S_t|A_t\big)+\bbH\big(S_t|M,K,A_t^N,S_{t+1}^N,Y^{t-1}\big)\Big]+N\epsilon_N\nonumber\\
    &\mathop=\limits^{(d)}\sum\limits_{t=1}^N\big[\bbI(U_t;Y_t)-\bbI(U_t;S_t|A_t)\big]+N\epsilon_N\nonumber\\
    &=N\sum\limits_{t=1}^N\bbP(T=t)\big[\bbI(U_T;Y_T|T=t)-\bbI(U_T;S_T|A_T,T=t)\big]+N\epsilon_N\nonumber\\
    &=N\big[\bbI(U_T;Y_T|T)-\bbI(U_T;S_T|A_T,T)\big]+N\epsilon_N\nonumber\\
    &\le N\big[\bbI(U_T,T;Y_T)-\bbI(U_T;S_T|A_T,T)\big]+N\epsilon_N\nonumber\\
    &\mathop=\limits^{(e)} N\big[\bbI(U;Y)-\bbI(U;S|A)\big]+N\epsilon_N\nonumber\\
    &\mathop\le\limits^{(f)} N\big[\bbI(U;Y)-\bbI(U;S|A)\big]+N\delta,\label{eq:Traditianl_Converse}
\end{align}where
\begin{itemize}
    \item[$(a)$] follows from Fano's inequality;
    \item[$(b)$] follows from Csisz\'{a}r-K\"{o}rner-Marton sum identity \cite[Lemma~7]{BCC:IT78};
    \item[$(c)$] follows since $S_t-A_t-\big(M,K,S_{t+1}^N,A_{t+1}^N\big)$ forms a Markov chain;
    \item[$(d)$] follows by defining $U_t\triangleq\big(M,K,Y^{t-1},A_t^N,S_{t+1}^N\big)$;
    \item[$(e)$] follows by defining $U\triangleq(U_T,T)$, $Y\triangleq Y_T$, $S\triangleq S_T$, and $A\triangleq(A_T,T)$;
    \item[$(f)$] follows by defining $\delta\triangleq\max\{\epsilon_N,\eta_N,\gamma\}$.
\end{itemize}
We also have,
\begin{align}
    NR&=\bbH(M)\nonumber\\
    &=\bbH(M|K)\nonumber\\
    &\mathop\le\limits^{(a)}\bbI\big(M;Y^N|K\big)+N\epsilon_N\nonumber\\
    &=\sum\limits_{t=1}^N\big[\bbI\big(M;Y_t|K,Y^{t-1}\big)\big]+N\epsilon_N\nonumber\\
    &\le\sum\limits_{t=1}^N\big[\bbI\big(M,A_t^N,Z^{t-1};Y_t|K,Y^{t-1}\big)\big]+N\epsilon_N\nonumber\\
    &=\sum\limits_{t=1}^N\Big[\bbI\big(M,A_t^N,Z^{t-1},S_{t+1}^N;Y_t|K,Y^{t-1}\big)-\bbI\big(S_{t+1}^N;Y_t|M,K,A_t^N,Z^{t-1},Y^{t-1}\big)\Big]+N\epsilon_N\nonumber\\
    &\mathop=\limits^{(b)}\sum\limits_{t=1}^N\Big[\bbI\big(M,A_t^N,Z^{t-1},S_{t+1}^N;Y_t|K,Y^{t-1}\big)-\bbI\big(Y^{t-1};S_t|M,K,A_t^N,Z^{t-1},S_{t+1}^N\big)\Big]+N\epsilon_N\nonumber\\
    &\mathop\le\limits^{(c)}\sum\limits_{t=1}^N\Big[\bbI\big(M,K,A_t^N,Z^{t-1},S_{t+1}^N,Y^{t-1};Y_t\big)-\bbI\big(Y^{t-1};S_t|M,K,A_t^N,Z^{t-1},S_{t+1}^N\big)\Big]+N\epsilon_N\nonumber\\
    &\mathop=\limits^{(d)}\sum\limits_{t=1}^N\big[\bbI(U_t,V_t;Y_t)-\bbI(U_t;S_t|V_t,A_t)\big]+N\epsilon_N\nonumber\\
    &=N\sum\limits_{t=1}^N\bbP(T=t)\big[\bbI(U_T,V_T;Y_T|T=t)-\bbI(U_T;S_T|V_T,A_T,T=t)\big]+N\epsilon_N\nonumber\\
    &=N\big[\bbI(U_T,V_T;Y_T|T)-\bbI(U_T;S_T|V_T,A_T,T)\big]+N\epsilon_N\nonumber\\
    &\le N\big[\bbI(U_T,V_T,T;Y_T)-\bbI(U_T;S_T|V_T,A_T,T)\big]+N\epsilon_N\nonumber\\
    &\mathop=\limits^{(e)} N\big[\bbI(U,V;Y)-\bbI(U;S|V,A)\big]+N\epsilon_N\nonumber\\
    &\mathop\le\limits^{(f)} N\big[\bbI(U,V;Y)-\bbI(U;S|V,A)\big]+N\delta\label{eq:Non-Traditianl_Converse}
\end{align}where
\begin{itemize}
    \item[$(a)$] follows from Fano's inequality;
    \item[$(b)$] follows from Csisz\'{a}r-K\"{o}rner-Marton sum identity \cite[Lemma~7]{BCC:IT78};
    \item[$(c)$] follows from the chain rule and the non-negativity of the mutual information;
    \item[$(d)$] follows by defining $U_t\triangleq\big(M,K,Y^{t-1},A_t^N,S_{t+1}^N\big)$ and $V_t\triangleq\big(M,K,Z^{t-1},A_{t+1}^N,S_{t+1}^N\big)$;
    \item[$(e)$] follows by defining $U\triangleq(U_T,T)$, $V\triangleq(V_T,T)$, $Y\triangleq Y_T$, $S\triangleq S_T$, and $A\triangleq(A_T,T)$;
    \item[$(f)$] follows by defining $\delta\triangleq\max\{\epsilon_N,\eta_N,\gamma\}$.
\end{itemize}
Also,
\begin{align}
    NR&=\bbH(M)\nonumber\\
    &=\bbH(M|K)\nonumber\\
    &\mathop\le\limits^{(a)}\bbI\big(M;Y^N|K\big)+N\epsilon_N\nonumber\\
    &=\sum\limits_{t=1}^N\big[\bbI\big(M;Y_t|K,Y^{t-1}\big)\big]+N\epsilon_N\nonumber\\
    &\le\sum\limits_{t=1}^N\big[\bbI\big(M,K,Z^{t-1},A_{t+1}^N,S_{t+1}^N;M,K,Y^{t-1},A_t^N,S_{t+1}^N,Y_t\big)\big]+N\epsilon_N\nonumber\\
    &\mathop=\limits^{(b)}\sum\limits_{t=1}^N\big[\bbI(V_t;A_t,U_t,Y_t)\big]+N\epsilon_N\nonumber\\
    &=N\sum\limits_{t=1}^N\bbP(T=t)\big[\bbI(V_t;A_t,U_t,Y_t|T=t)\big]+N\epsilon_N\nonumber\\
    &=N\bbI(V_T;A_T,U_T,Y_T|T)+N\epsilon_N\nonumber\\
    &\le N\bbI(V_T,T;A_T,U_T,T,Y_T)+N\epsilon_N\nonumber\\
    &\mathop=\limits^{(c)} N\bbI(V;A,U,Y)+N\epsilon_N\nonumber\\
    &\mathop\le\limits^{(d)} N\bbI(V;A,U,Y)+N\delta\label{eq:Non-Traditianl3_Converse}
\end{align}where
\begin{itemize}
    \item[$(a)$] follows from Fano's inequality;
    \item[$(b)$] follows by defining $U_t\triangleq\big(M,K,Y^{t-1},A_t^N,S_{t+1}^N\big)$ and $V_t\triangleq\big(M,K,Z^{t-1},A_{t+1}^N,S_{t+1}^N\big)$;
    \item[$(c)$] follows by defining $U\triangleq(U_T,T)$, $V\triangleq(V_T,T)$, $Y\triangleq Y_T$, $S\triangleq S_T$, and $A\triangleq(A_T,T)$;
    \item[$(d)$] follows by defining $\delta\triangleq\max\{\epsilon_N,\eta_N,\gamma\}$.
\end{itemize}
We can also lower bound the rate $NR+\bar{R}_K^{(N)}$ as follows,
\begin{align}
    NR+\bar{R}_K^{(N)}&=\bbH(M,K)\nonumber\\
    &\ge\bbI\big(M,K;Z^N\big)\nonumber\\
    &=\sum\limits_{t=1}^N\big[\bbI\big(M,K;Z_t|Z^{t-1}\big)\big]\nonumber\\
    &=\sum\limits_{t=1}^N\big[\bbI\big(M,K,S_{t+1}^N,A_{t+1}^N;Z_t|Z^{t-1}\big)-\bbI\big(S_{t+1}^N,A_{t+1}^N;Z_t|M,K,Z^{t-1}\big)\big]\nonumber\\
    &\mathop=\limits^{(a)}\sum\limits_{t=1}^N\big[\bbI\big(M,K,S_{t+1}^N,A_{t+1}^N;Z_t|Z^{t-1}\big)-\bbI\big(S_{t+1}^N;Z_t|M,K,Z^{t-1}\big)\big]\nonumber\\
    &\mathop=\limits^{(b)}\sum\limits_{t=1}^N\big[\bbI\big(M,K,S_{t+1}^N,A_{t+1}^N;Z_t|Z^{t-1}\big)-\bbI\big(Z^{t-1};S_t|M,K,S_{t+1}^N\big)\big]\nonumber\\
    &\ge\sum\limits_{t=1}^N\big[\bbI\big(M,K,S_{t+1}^N,A_{t+1}^N;Z_t|Z^{t-1}\big)-\bbI\big(M,K,S_{t+1}^N,A_{t+1}^N,Z^{t-1};S_t\big)\big]\nonumber\\
    &\mathop\ge\limits^{(c)}\sum\limits_{t=1}^N\big[\bbI\big(M,K,S_{t+1}^N,A_{t+1}^N,Z^{t-1};Z_t\big)-\bbI\big(M,K,S_{t+1}^N,A_{t+1}^N,Z^{t-1};S_t\big)\big]-\tilde{\epsilon}\nonumber\\
    &\mathop=\limits^{(d)}\sum\limits_{t=1}^N\big[\bbI(V_t;Z_t)-\bbI(V_t;S_t)\big]-\tilde{\epsilon}\nonumber\\
    &=N\sum\limits_{t=1}^N\bbP(T=t)\big[\bbI(V_T;Z_T|T=t)-\bbI(V_T;S_T|T=t)\big]-\tilde{\epsilon}\nonumber\\
    &=N\big[\bbI(V_T;Z_T|T)-\bbI(V_T;S_T|T)\big]-\tilde{\epsilon}\nonumber\\
    &\mathop\ge\limits^{(e)} N\big[\bbI(V_T,T;Z_T)-\bbI(V_T;S_T|T)\big]-2\tilde{\epsilon}\nonumber\\
    &= N\big[\bbI(V_T,T;Z_T)-\bbI(V_T,T;S_T)\big]-2\tilde{\epsilon}\nonumber\\
    &\mathop=\limits^{(f)} N\big[\bbI(V;Z)-\bbI(V;S)\big]-2\tilde{\epsilon}\label{eq:Non-Resolvability1_Converse}
\end{align}where
\begin{itemize}
    \item[$(a)$] follows since $Z_t-\big(M,K,S_{t+1}^N,Z^{t-1}\big)-A_{t+1}^N$ forms a Markov chain;
    \item[$(b)$] follows from Csisz\'{a}r-K\"{o}rner-Marton sum identity \cite[Lemma~7]{BCC:IT78};
    \item[$(c)$] and $(e)$ follow from Lemma~\ref{lemma:iid_Time_indep};
    \item[$(d)$] follows by defining  $V_t\triangleq\big(M,K,S_{t+1}^N,A_{t+1}^N,Z^{t-1}\big)$;
    \item[$(f)$] follows by defining $U\triangleq(U_T,T)$, $V\triangleq(V_T,T)$, $Y\triangleq Y_T$, $S\triangleq S_T$, and $A\triangleq(A_T,T)$.
\end{itemize}
For any $\gamma>0$, selecting $N$ large enough ensures that,
\begin{align}
    R+\frac{\bar{R}_K^{(N)}}{N}\ge\bbI(V;Z)-\bbI(V;S)-2\gamma.
\end{align}Hence,
\begin{align}
    R&\ge\bbI(V;Z)-\bbI(V;S)-2\gamma-\frac{\bar{R}_K^{(N)}}{N}\nonumber\\
    &=\bbI(V;Z)-\bbI(V;S)-2\gamma-\eta_N\nonumber\\
    &\ge\bbI(V;Z)-\bbI(V;S)-3\delta,\label{eq:Non-Resolvability_Converse}
\end{align}where the last inequality follows by defining $\delta\triangleq\max\{\epsilon_N,\eta_N,\gamma\}$. 
To prove that $\bbD(P_Z||Q_0)\le\epsilon$, for $N$ large enough we have
\begin{align}
    \bbD(P_Z||Q_0)&=\bbD(P_{Z_T}||Q_0)=\bbD\left(\frac{1}{N}\sum\limits_{t=1}^NP_{Z_t}\Big|\Big|Q_0\right)\le\frac{1}{N}\sum\limits_{t=1}^N\bbD\left(P_{Z_t}\Big|\Big|Q_0\right)\nonumber\\
    &\le\frac{1}{N}\bbD\left(P_{Z^N}\Big|\Big|Q_0^{\otimes N}\right)\le\frac{\tilde{\epsilon}}{N}\le\gamma\le\delta.\nonumber
\end{align}Combining \eqref{eq:Traditianl_Converse}, \eqref{eq:Non-Traditianl_Converse}, \eqref{eq:Non-Traditianl3_Converse}, and \eqref{eq:Non-Resolvability_Converse} proves that $\forall\epsilon_N,\tilde{\epsilon},\eta_N$, $R\le\max\left\{R:R\in\calF^{(\epsilon)}_{\text{U-NC}}\right\}$. Hence,
\begin{align}
    R\le\max\left\{R:R\in\bigcap\limits_{\epsilon>0}\calF^{(\epsilon)}_{\text{U-NC}}\right\}.\nonumber
\end{align}
\subsection{Proof for Continuity at Zero}The proof follows the similar lines as the proof for continuity at zero in \cite[Appendix~F]{Keyless22}.

\section{Proof of Theorem~\ref{thm:Acievability_KG_C}}
\label{proof_Acive_Causal}
Similar to the achievability proof of Theorem~\ref{thm:Acievability_KG}, our achievability proof is based on a block Markov coding scheme where we transmit $B$ independent messages over $B$ blocks each of length $n$ and $N=nB$. Therefore, the warden's observation is $Z^N=(Z_1^n,Z_2^n,\dots,Z_B^n)$. The distribution induced by our coding scheme at the warden's channel output observation is denoted by $P_{Z^N}\triangleq P_{Z_1^n,Z_2^n,\dots,Z_B^n}$, and the target distribution at the warden's channel output observation is $Q_0^{\otimes N}\triangleq\prod_{b = 1}^BQ_0^{\otimes n}$ and the expansion in \eqref{eq:Cov_Comm_Constraint_NC} continues to be valid. Now fix $\epsilon>0$, $P_A$, $P_{V|S}$, $P_{U|A}$, and $P_{X|US}$.

\subsection{Codebook Generation}
\subsubsection{Codebooks for Key Generation}
For each block $b\in[B]$, let $C_{V_b}^{(n)}\triangleq\big(V^n(j_b)\big)_{j_b\in\calJ}$, where $\calJ\triangleq\brk{2^{nR_J}}$, be a set of random codewords generated \ac{iid} according to $P_V$, where $P_V=\sum_{a\in\calA}\sum_{s\in\calS}P_A(a)Q_{S|A}(s|a)P_{V|S}(v|s)$. A realization of $C_{V_b}^{(n)}$ is denoted by $\calC_{V_b}^{(n)}\triangleq\big(v^n(j_b)\big)_{j_b\in\calJ}$. For each block $b\in[B]$, partition the indices $j_b\in\calJ$ into bins $\calB(\ell_b)$, where $\ell_b\in[2^{n\tilde{R}_L}]$, by applying function $\Psi_L:v^n(j_b)\mapsto[2^{n\tilde{R}_L}]$ via random binning by selecting $\Psi_L\big(v^n(j_b)\big)$ independently and uniformly at random for every $v^n(j_b)\in\calV^n$. For each block $b\in[B]$, create a function $\Psi_K:v^n(j_b)\mapsto[2^{n\tilde{R}_K}]$ via random binning by selecting $\Psi_K\big(v^n(j_b)\big)$ independently and uniformly at random for every $v^n(j_b)\in\calV^n$. The key $k_b=\Psi_K\big(v^n(j_b)\big)$ generated in the block $b\in[B]$ from the description of the \ac{ADSI} $v^n(j_b)$ is split into two independent parts $k_{1,b}$, with rate $R_{K_1}$, and $k_{2,b}$, with rate $R_{K_2}$, where $R_K=R_{K_1}+R_{K_2}$, and will be used to help the transmitter and the receiver in the block~$b+2$.
\subsubsection{Action Codebooks}For each block $b\in[B]$, let $C_{A_b}^{(n)}\triangleq\big(A^n(m_b,k_{1,b-2})\big)_{(m_b,k_{1,b-2})\in\calM\times\calK_1}$, where $\calM\triangleq\big[2^{nR}\big]$ and $\calK_1\triangleq\big[2^{nR_{K_1}}\big]$, be a set of random codewords generated \ac{iid} according to $P_A$. We denote a realization of $C_{A_b}^{(n)}$ by $\calC_{A_b}^{(n)}\triangleq\big(a^n(m_b,k_{1,b-2})\big)_{(m_b,k_{1,b-2})\in\calM\times\calK_1}$. 
\subsubsection{Codebooks for Message Transmission}
For each block $b\in[B]$ and for each $(m_b,,k_{1,b-2})\in\calM\times\calK_1$, let $C_{U_b}^{(n)}\triangleq\big(U^n(m_b,k_{1,b-2},\ell_{b-1},k_{2,b-2})\big)_{(m_b,k_{1,b-2},\ell_{b-1},k_{2,b-2})\in\calM\times\calK_1\times\calL\times\calK_2}$, where $\calL\triangleq\big[2^{nR_L}\big]$, and  $\calK_2\triangleq\big[2^{nR_{K_2}}\big]$, be a random codebook generated \ac{iid} according to $\prod\nolimits_{i=1}^n P_{U|A}\big(\cdot|A_i(m_b,k_{1,b-2})\big)$. The indices $(\ell_{b-1},k_{b-2})$ can also be interpreted as two layer random binning. A realization of $C_{U_b}^{(n)}$ is denoted by $\calC_{U_b}^{(n)}\triangleq\big(u^n(m_b,k_{1,b-2},\ell_{b-1},k_{2,b-2})\big)_{(m_b,k_{1,b-1},\ell_{b-1},k_{2,b-2})\in\calM\times\calK_1\times\calL\times\calK_2}$. Let $C_b^{(n)}\triangleq\left(C_{A_b}^{(n)},C_{V_b}^{(n)},C_{U_b}^{(n)}\right)$, $\calC_b^{(n)}\triangleq\left(\calC_{A_b}^{(n)},\calC_{V_b}^{(n)},\calC_{U_b}^{(n)}\right)$, $C_N\triangleq\left(C_b^{(n)}\right)_{b\in[B]}$, and~$\calC_N\triangleq\left(\calC_b^{(n)}\right)_{b\in[B]}$. For $b\in[B]$, to facilitate the analysis, we consider the following joint \ac{PMF} for a fixed codebook $\calC_b^{(n)}$ in which all the indices $m_b$, $k_{1,b-2}$, $j_b$, $\ell_{b-1}$, and $k_{2,b-2}$, are chosen uniformly at random,
\begin{align}
    &\Gamma_{M_bK_{1,b-2}A^nJ_bV^nL_{b-1}K_{2,b-2}U^nS_b^nY_b^nZ_b^nK_{b-1}L_bK_b|\calC_b^{(n)}}(m_b,k_{1,b-2},\tilde{a}^n,j_b,\tilde{v}^n,\ell_{b-1},k_{2,b-2},\tilde{u}^n,s_b^n,y_b^n,z_b^n,k_{b-1},\ell_b,k_b)\nonumber\\
    &=\frac{1}{\abs{\calM}\abs{\calK_1}\abs{\calJ}\abs{\calL}\abs{\calK_2}}\indi{1}_{\{a^n(m_b,k_{1,b-2})=\tilde{a}^n\}\cap\{v^n(j_b)=\tilde{v}^n\}\cap\{u^n(m_b,k_{1,b-2},\ell_{b-1},k_{2,b-2})=\tilde{u}^n\}}\nonumber\\
    &\Squad\times Q_{S|AV}^{\otimes n}\big(s^n|\tilde{a}^n,\tilde{v}^n\big) W_{YZ|SU}^{\otimes n}(y_b^n,z_b^n|s^n,\tilde{u}^n)\frac{1}{\abs{\calK_1}\abs{\calK_2}}\indi{1}_{\big\{\ell_b=\Psi_L\big(v^n(j_b)\big)\big\}\bigcap\big\{k_b=\Psi_K\big(v^n(j_b)\big)\big\}},\label{eq:Encoding_Ideal_PMF_C}
\end{align}where $W_{YZ|SU}(y,z|s,u)=\sum_{x\in\calX}P_{X|US}(x|u,s)W_{YZ|SX}(y,z|s,x)$, and
\begin{align}
    Q_{S|AV}(s|a,v)&=\frac{P_A(a)Q_{S|A}(s|a)P_{V|S}(v|s)}{\sum_{s\in\calS} P_A(a)Q_{S|A}(s|a)P_{V|S}(v|s)}.\nonumber
\end{align}

\subsection{Encoding}To initiate the key generation process, the transmitter, and the receiver are assumed to share $k_{-1}\triangleq(k_{1,-1},k_{2,-1})$ and $k_0\triangleq(k_{1,0},k_{2,0})$ to be used in blocks $b=1$ and $b=2$, respectively. After the block $b=2$, the transmitter and the receiver use the secret key that they generate from the \ac{ADSI}.
\subsubsection{Encoding Scheme for the First Block}
Given the key $k_{-1}\triangleq(k_{1,-1},k_{2,-1})$, to transmit the message $m_1\in\calM$, the encoder first chooses the action sequence $a^n(m_1,k_{1,-1})$, then, the nature chooses the channel state $s_1^n$ according to the action sequence $a^n(m_1,k_{1,-1})$. Next, the encoder chooses the reconciliation index $\ell_0$ uniformly at random. Note that, the reconciliation index $\ell_0$ does not convey any information about the channel state. Now given the key $k_{-1}$, the encoder chooses the index $j_1$ according to the following distribution with $b=1$,
\begin{align}
g_{\text{\tiny{LE}}}\big(j_b|m_b,\ell_{b-1},k_{b-2},a^n(m_b,k_{1,b-2}),s_b^n\big)\triangleq\frac{Q_{S|AV}^{\otimes n}\big(s_b^n|a^n(m_b,k_{1,b-2}),v^n(j_b)\big)}{\sum_{j'_b\in\calJ}Q_{S|AV}^{\otimes n}\big(s_b^n|a^n(m_b,k_{1,b-2}),v^n(j'_b)\big)},\label{eq:LE_C}
\end{align}where $k_{b-2}\triangleq(k_{1,b-2},k_{2,b-2})$. 
Then, based on $(m_1,k_{1,-1},\ell_0,k_{2,-1})$ and, $j_1$ the encoder computes $u^n(m_1,k_{1,-1},\ell_0,k_{2,-1})$ and $v^n(j_1)$ and transmits $x_1^n$, where $x_{1,i}$ is generated by passing $u_i(m_1,k_{1,-1},\ell_0,k_{2,-1})$ and $s_{1,i}$ through the test channel $P_{X|US}\big(x_{1,i}|u_i(m_1,k_{1,-1},\ell_0,k_{2,-1}),s_i\big)$. Simultaneously, the encoder generates a reconciliation index $\ell_1$ and a key $k_1$, which will be split in two independent parts $k_{1,1}$ and $k_{2,1}$, from the description of the \ac{ADSI} $v^n(j_1)$ to be transmitted in the second and the third blocks, respectively. 
\subsubsection{Encoding Scheme for the Second Block}
Similarly, given the key $k_0\triangleq(k_{1,0},k_{2,0})$, to transmit the message $m_2\in\calM$, the encoder first chooses the action sequence $a^n(m_2,k_{1,0})$, then, the nature chooses the channel state $s_2^n$ according to the action sequence $a^n(m_2,k_{1,0})$. Now given the key $k_0$ and the reconciliation information $\ell_1$, generated in the previous block, the encoder chooses the index $j_2$ according to \eqref{eq:LE_C} with $b=2$. Then, based on $(m_2,k_{1,0},\ell_1,k_{2,0})$ and, $j_2$ the encoder computes $u^n(m_2,k_{1,0},\ell_1,k_{2,0})$ and $v^n(j_2)$ and transmits $x_2^n$, where $x_{2,i}$ is generated by passing $u_i(m_2,k_{1,0},\ell_1,k_{2,0})$ and $s_{2,i}$ through the test channel $P_{X|US}\big(x_{2,i}|u_i(m_2,k_{1,0},\ell_1,k_{2,0}),s_{2,i}\big)$. Simultaneously, the encoder generates a reconciliation index $\ell_2$ and a key $k_2$, which will be split in two independent parts $k_{1,2}$ and $k_{2,2}$, from the description of the \ac{ADSI} $v^n(j_2)$ to be transmitted in the third and the fourth blocks, respectively.
\subsubsection{Encoding Scheme for the Block \texorpdfstring{$b\in\sbra{3}{B}$}{Lg}}
Similarly, given the key $k_{b-2}\triangleq(k_{1,b-2},k_{2,b-2})$, to transmit the message $m_b\in\calM$, the encoder first chooses the action sequence $a^n(m_b,k_{1,b-2})$, then, the nature chooses the channel state $s_b^n$ according to the action sequence $a^n(m_b,k_{1,b-2})$. Now given the key $k_{b-2}$, generated in the block $b-2$, and the reconciliation information $\ell_{b-1}$, generated in the previous block, the encoder chooses the index $j_b$ according to \eqref{eq:LE_NC}. Then, based on $(m_b,k_{1,b-2},\ell_{b-1},k_{2,b-2})$ and, $j_b$ the encoder computes $u^n(m_b,k_{1,b-2},\ell_{b-1},k_{2,b-2})$ and $v^n(j_b)$ and transmits $x_b^n$, where $x_{b,i}$ is generated by passing $u_i(m_b,k_{1,b-2},\ell_{b-1},k_{2,b-2})$ and $s_{b,i}$ through the test channel $P_{X|US}\big(x_{b,i}|u_i(m_b,k_{1,b-2},\ell_{b-1},k_{2,b-2}),s_{b,i}\big)$. Simultaneously, the encoder generates a reconciliation index $\ell_b$ and a key $k_b\triangleq(k_{1,b},k_{2,b})$, which will be split in two independent parts $k_{1,b}$ and $k_{2,b}$, from the description of the \ac{ADSI} $v^n(j_b)$ to be transmitted in block $b+1$ and block $b+2$, respectively. 
Note that, in this scheme, the description of the \ac{ADSI} is used only for the key generation, not the message transmission.

Therefore, noting that our coding scheme ensures that the reconciliation index $L_{b-1}$ and the keys $K_{b-1}$ and $K_b$ are (arbitrarily) nearly uniformly distributed, the joint \ac{PMF} between the involved \acp{RV} is given by,
\begin{align}
    &P_{M_bK_{1,b-2}L_{b-1}K_{2,b-2}A^nS_b^nJ_bU^nV^nX^nY_b^nZ_b^nK_{b-1}L_bK_b|\calC_b^{(n)}}(m_b,k_{1,b-2},\ell_{b-1},k_{2,b-2},\tilde{a}^n,s_b^n,j_b,\tilde{u}^n,\tilde{v}^n,\tilde{x}^n,y_b^n,z_b^n,k_{b-1},\ell_b,k_b)\nonumber\\
    &\Squad=\frac{1}{\abs{\calM}\abs{\calK_1}\abs{\calL}\abs{\calK_2}}\indi{1}_{\{a^n(m_b,k_{1,b-2})=\tilde{a}^n\}}Q_{S|A}^{\otimes n}\big(s_b^n|\tilde{a}^n\big)\nonumber\\
    &\Squad\times g_{\text{\tiny{LE}}}\big(j_b|m_b,\ell_{b-1},k_{b-2},\tilde{a}^n,s_b^n\big)\indi{1}_{\{u^n(m_b,k_{1,b-2},\ell_{b-1},k_{2,b-2})=\tilde{u}^n\}\cap\{v^n(j_b)=\tilde{v}^n\}}\nonumber\\
    &\Squad\times P_{X|US}^{\otimes n}\left(\tilde{x}^n|u^n(m_b,k_{1,b-2},\ell_{b-1},k_{2,b-2}),s_b^n\right)W_{YZ|SX}^{\otimes n}\big(y_b^n,z_b^n|s_b^n,\tilde{x}^n\big)\nonumber\\
    &\Squad\times\frac{1}{\abs{\calK_1}\abs{\calK_2}}\indi{1}_{\big\{\ell_b=\Psi_L\big(v^n(j_b)\big)\big\}\bigcap\big\{k_b=\Psi_K\big(v^n(j_b)\big)\big\}}.\label{eq:Encoding_Joint_C}
\end{align}
\subsection{Covert Analysis}
\begin{figure*}
\centering
\includegraphics[width=6.0in]{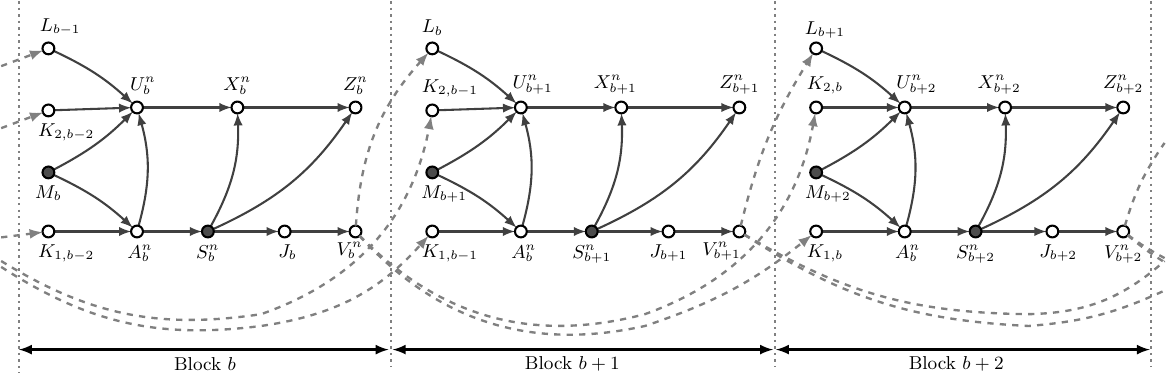}
\caption{Functional dependency graph for all the involved \acp{RV}.}
\label{fig:Dependency_C}
\end{figure*}
To prove that our coding scheme is covert, we show that $\bbE_{C_N}\bbD\left(P_{Z^N|C_N}||Q_Z^{\otimes N}\right)\xrightarrow[]{N\to\infty}0$, where
\begin{subequations}\label{eq:Distribution_Conditions_C}
\begin{align}
    Q_Z(\cdot)&\triangleq\sum_{a\in\calA}\sum_{v\in\calV}\sum_{u\in\calU}\sum_{s\in\calS}\sum_{x\in\calX}P_A(a)P_V(v)P_{U|A}(u|a)Q_{S|AV}(s|a,v)P_{X|US}(x|u,s)W_{Z|XS}(\cdot|x,s),\label{eq:QZ_Defi_C}
\end{align}such that,
\begin{align}
    \sum_{v\in\calV}P_{V}(v)Q_{S|AV}(s|a,v)=Q_{S|A}(\cdot|a).\label{eq:Dist_Condition_C}
\end{align}
\end{subequations}
Then we choose the \acp{PMF} $P_A$, $P_{V|S}$, $P_{U|A}$, and $P_{X|US}$ such that $Q_0=Q_Z$. Now, for $b\in[B]$
\begin{subequations}\label{eq:Dep_Graph_equ2_C}
\begin{align}
   \bbI\left(Z_b^n;Z_{b+1}^{B,n}\right)&\le\bbI\left(Z_b^n;K_{b-1},L_b,K_b,Z_{b+1}^{B,n}\right)\nonumber\\
   &=\bbI\left(Z_b^n;K_{b-1},L_b,K_b\right),\label{eq:Dep_Graph_equ_C}
\end{align}where \eqref{eq:Dep_Graph_equ_C} follows since $Z_b^n-\left(K_{b-1},L_b,K_b\right)-Z_{b+1}^{B,n}$ forms a Markov chain, as illustrated in the functional dependency graph in Fig.~\ref{fig:Dependency_C}. Now,
\begin{align}
    \bbI\left(Z_b^n;K_{b-1},L_b,K_b\right)&=\bbD\left(P_{Z_b^nK_{b-1}L_bK_b|\calC_b^{(n)}}\big|\big|P_{Z_b^n|\calC_b^{(n)}}P_{K_{b-1}L_bK_b|\calC_b^{(n)}}\right)\nonumber\\
    &\le\bbD\left(P_{Z_b^nK_{b-1}L_bK_b|\calC_b^{(n)}}\big|\big|Q_Z^{\otimes n}P_{K_{b-1}}^UP_{L_b}^UP_{K_b}^U\right),\label{eq:Dep_Graph_equ21_C}
\end{align}where $P_{K_{b-1}}^UP_{L_b}^UP_{K_b}^U$ is the uniform distribution over $\Big[2^{nR_K}\Big]\times\left[2^{n\tilde{R}_L}\right]\times\left[2^{n\tilde{R}_K}\right]$ and \eqref{eq:Dep_Graph_equ21_C} holds since
\begin{align}
    &\bbD\left(P_{Z_b^nK_{b-1}L_bK_b|\calC_b^{(n)}}\big|\big|P_{Z_b^n|\calC_b^{(n)}}P_{K_{b-1}L_bK_b|\calC_b^{(n)}}\right)=\bbD\left(P_{Z_b^nK_{b-1}L_bK_b|\calC_b^{(n)}}\big|\big|Q_Z^{\otimes n}P_{K_{b-1}}^UP_{L_b}^UP_{K_b}^U\right)\nonumber\\
    &\quad-\bbD\left(P_{Z_b^n|\calC_b^{(n)}}\big|\big|Q_Z^{\otimes n}\right)-\bbD\left(P_{K_{b-1}L_bK_b|\calC_b^{(n)}}\big|\big|P_{K_{b-1}}^UP_{L_b}^UP_{K_b}^U\right).\label{eq:Dep_Graph_equ22_C}
\end{align}
\end{subequations}
Therefore, combining \eqref{eq:Dep_Graph_equ2_C} with the expansion in  \eqref{eq:Cov_Comm_Constraint_NC}, by substituting $Q_0$ with $Q_Z$ and considering the expectation over the random codebook construction, results to
\begin{align}
    \bbE_{C_N}\kd{P_{Z^N|C_N}||}{Q_0^{\otimes N}}&\le 2\sum_{b=1}^B\bbE_{C_b^{(n)}}\bbD\left(P_{Z_b^nK_{b-1}L_bK_b|C_b^{(n)}}\big|\big|Q_Z^{\otimes n}P_{K_{b-1}}^UP_{L_b}^UP_{K_b}^U\right).\label{eq:Cov_Cons_General_C}
\end{align}To bound the \ac{RHS} of \eqref{eq:Cov_Cons_General_C}, using \eqref{eq:Reverse_Pinsker} in Lemma~\ref{lemma:KLD_TV}, it is sufficient to bound,
\begin{align}
    &\ToV\left(P_{Z_b^nK_{b-1}L_bK_b|\calC_b^{(n)}},Q_Z^{\otimes n}P_{K_{b-1}}^UP_{L_b}^UP_{K_b}^U\right)\nonumber\\
    &\le\ToV\left(P_{Z_b^nK_{b-1}L_bK_b|\calC_b^{(n)}},\Gamma_{Z_b^nK_{b-1}L_bK_b|\calC_b^{(n)}}\right)+\ToV\left(\Gamma_{Z_b^nK_{b-1}L_bK_b|\calC_b^{(n)}},Q_Z^{\otimes n}P_{K_{b-1}}^UP_{L_b}^UP_{K_b}^U\right),\label{eq:General_KLD_C}
\end{align}where the inequality follows from the triangle inequality.

Note that, from \eqref{eq:Encoding_Ideal_PMF_C} and \eqref{eq:Encoding_Joint_C} we have the following marginal distributions,
\begin{align}
    &\Gamma_{Z_b^nK_{b-1}L_bK_b|\calC_b^{(n)}}(z^n,k_{b-1},\ell_b,k_b)=\sum\limits_{m_b}\sum\limits_{k_{1,b-2}}\sum\limits_{j_b}\sum\limits_{\ell_{b-1}}\sum\limits_{k_{2,b-2}}\sum\limits_{s_b^n}\frac{Q_{S|AV}^{\otimes n}\Big(s^n|a^n(m_b,k_{1,b-2}),v^n(j_b)\Big)}{2^{n(R+R_J+R_L+2R_K)}}\times\nonumber\\
    &\qquad W_{Z|SU}^{\otimes n}\Big(z_b^n|s_b^n,u^n(m_b,k_{1,b-2},\ell_{b-1},k_{2,b-2})\Big)\indi{1}_{\big\{\ell_b=\Psi_L\big(v^n(j_b)\big)\big\}\bigcap\big\{k_b=\Psi_K\big(v^n(j_b)\big)\big\}}\nonumber\\
    &=\sum\limits_{m_b}\sum\limits_{k_{1,b-2}}\sum\limits_{j_b}\sum\limits_{\ell_{b-1}}\sum\limits_{k_{2,b-2}}\frac{1}{2^{n(R+R_J+R_L+2R_K)}}\times\nonumber\\
    &W_{Z|AVU}^{\otimes n}\Big(z_b^n|a^n(m_b,k_{1,b-2}),v^n(j_b),u^n(m_b,k_{1,b-2},\ell_{b-1},k_{2,b-2})\Big)\indi{1}_{\big\{\ell_b=\Psi_L\big(v^n(j_b)\big)\big\}\bigcap\big\{k_b=\Psi_K\big(v^n(j_b)\big)\big\}},\label{eq:Encoding_GZ_C}
\end{align}where
\begin{subequations}
\begin{align}
    W_{Z|SU}(z|s,u)&\triangleq\sum_{y\in\calY}\sum_{x\in\calX}P_{X|US}(x|u,s)W_{YZ|SX}(y,z|s,x),\nonumber\\
    W_{Z|AVU}(z|a,v,u)&\triangleq\sum_{y\in\calY}\sum_{x\in\calX}\sum_{s\in\calS}Q_{S|AVU}(s|a,v,u)P_{X|US}(x|u,s)W_{YZ|SX}(y,z|s,x).\nonumber
\end{align}
\end{subequations}To bound the first term on the \ac{RHS} of \eqref{eq:General_KLD_C} we have
\begin{align}
    &\bbE_{C_b^{(n)}}\ToV\left(P_{Z_b^nK_{b-1}L_bK_b|C_b^{(n)}},\Gamma_{Z_b^nK_{b-1}L_bK_b|C_b^{(n)}}\right)\nonumber\\
    &\le\bbE_{C_b^{(n)}}\ToV\left(P_{M_bK_{1,b-2}L_{b-1}K_{2,b-2}A^nS_b^nJ_bU^nV^nX^nY_b^nZ_b^nK_{b-1}L_bK_b|C_b^{(n)}},\right.\nonumber\\
    &\quad\left.\Gamma_{M_bK_{1,b-2}L_{b-1}K_{2,b-2}A^nS_b^nJ_bU^nV^nX^nY_b^nZ_b^nK_{b-1}L_bK_b|C_b^{(n)}}\right)\nonumber\\
    &\mathop=\limits^{(a)}\bbE_{C_b^{(n)}}\ToV\left(P_{M_bL_{b-1}K_{1,b-2}K_{2,b-2}A^nS_b^n|C_b^{(n)}},\Gamma_{M_bL_{b-1}K_{1,b-2}K_{2,b-2}A^nS_b^n|C_b^{(n)}}\right)\nonumber\\
    &\mathop=\limits^{(b)}\bbE_{C_b^{(n)}}\ToV\left(P_{A^nS_b^n|M_b=1,L_{b-1}=1,K_{1,b-2}=1,K_{2,b-2}=1,C_b^{(n)}},\Gamma_{A^nS_b^n|M_b=1,L_{b-1}=1,K_{1,b-2}=1,K_{2,b-2}=1,C_b^{(n)}}\right)\nonumber\\
    &\mathop=\limits^{(c)}\bbE_{C_{A_b}^{(n)}}\left[\bbE_{C_{U_b}^{(n)}C_{V_b}^{(n)}|C_{A_b}^{(n)}}\left[\ToV\left(P_{A^nS_b^n|M_b=1,L_{b-1}=1,K_{1,b-2}=1,K_{2,b-2}=1,C_b^{(n)}},\right.\right.\right.\nonumber\\
    &\quad\left.\left.\left.\Gamma_{A^nS_b^n|M_b=1,L_{b-1}=1,K_{1,b-2}=1,K_{2,b-2}=1,C_b^{(n)}}\right)\left|C_{A_b}^{(n)}\right.\right]\right],\label{eq:General_TV_C}
\end{align}where
\begin{itemize}
    \item[$(a)$] follows since
    \begin{subequations}
\begin{align}
    &P_{U^n|M_bL_{b-1}K_{1,b-2}K_{2,b-2}A^nS_b^n\calC_b^{(n)}}=\indi{1}_{\{U^n=u^n(M_b,K_{1,b-2},L_{b-1},K_{2,b-2})\}}=\nonumber\\
    &\qquad\Gamma_{U^n|M_bL_{b-1}K_{1,b-2}K_{2,b-2}A^nS_b^n\calC_b^{(n)}}\label{eq:PGC_3},\\
    &P_{J_b|M_bL_{b-1}K_{1,b-2}K_{2,b-2}A^nS_b^nU_b^n\calC_b^{(n)}}=g_{\text{\tiny{LE}}}\big(J_b|M_b,K_{1,b-2},L_{b-1},K_{2,b-2},a^n(M_b,K_{1,b-2}),S_b^n\big)=\nonumber\\
    &\qquad\Gamma_{J_b|M_bL_{b-1}K_{1,b-2}K_{2,b-2}A^nS_b^nU_b^n\calC_b^{(n)}}\label{eq:PGC_4},\\
    &P_{V^n|M_bL_{b-1}K_{1,b-2}K_{2,b-2}A^nS_b^nU^nJ_b\calC_b^{(n)}}=\indi{1}_{\{V^n=v^n(J_b)\}}=\Gamma_{V^n|M_bL_{b-1}K_{1,b-2}K_{2,b-2}A^nS_b^nU^nJ_b\calC_b^{(n)}}\label{eq:PGC_5},\\
    &P_{X^n|M_bL_{b-1}K_{1,b-2}K_{2,b-2}A^nS_b^nJ_bU^nV^n\calC_b^{(n)}}=P_{X|US}^{\otimes n}=\Gamma_{X^n|M_bL_{b-1}K_{1,b-2}K_{2,b-2}A^nS_b^nU^nJ_bV^n\calC_b^{(n)}}\label{eq:PGC_6},\\
    &P_{Y^nZ^n|M_bL_{b-1}K_{1,b-2}K_{2,b-2}A^nS_b^nU^nJ_bV^nX^n\calC_b^{(n)}}=W_{YZ|SX}^{\otimes n}=\nonumber\\
    &\qquad\Gamma_{Y^nZ^n|M_bL_{b-1}K_{1,b-2}K_{2,b-2}A^nS_b^nU^nJ_bV^nX^n\calC_b^{(n)}}\label{eq:PGC_7},\\
    &P_{L_bK_b|M_bL_{b-1}K_{1,b-2}K_{2,b-2}A^nS_b^nI_bJ_bU^nV^nX^nY^nZ^n\calC_b^{(n)}}=\indi{1}_{\big\{\ell_b=\Psi_L\big(v^n(j_b)\big)\big\}\bigcap\big\{k_b=\Psi_K\big(v^n(j_b)\big)\big\}}=\nonumber\\
    &\qquad\Gamma_{L_bK_b|M_bL_{b-1}K_{1,b-2}K_{2,b-2}A^nS_b^nI_bJ_bU^nV^nX^nY^nZ^n\calC_b^{(n)}}\label{eq:PGC_8},
\end{align}
\end{subequations}while \eqref{eq:PGC_3} follows since
\begin{align}
    &\Gamma_{J_b|M_bL_{b-1}K_{1,b-2}K_{2,b-2}A^nS_b^n\calC_b^{(n)}}\nonumber\\
    &=\frac{\Gamma_{J_bM_bL_{b-1}K_{1,b-2}K_{2,b-2}A^nS_b^n|\calC_b^{(n)}}}{\Gamma_{M_bL_{b-1}K_{1,b-2}K_{2,b-2}A^nS_b^n|\calC_b^{(n)}}}\nonumber\\
    &=\frac{Q_{S|AVU}^{\otimes n}\big(s^n|a^n(m_b,k_{1,b-2}),v^n(j_b),u^n(m_b,k_{1,b-2},\ell_{b-1},k_{2,b-2})\big)}{\sum_{j_b'}Q_{S|AVU}^{\otimes n}\big(s^n|a^n(m_b,k_{1,b-2}),v^n(j'_b),u^n(m_b,k_{1,b-2},\ell_{b-1},k_{2,b-2})\big)}\nonumber\\
    &=g_{\text{\tiny{LE}}}\big(j_b|m_b,\ell_{b-1},k_{b-2},a^n(m_b,k_{1,b-2}),s_b^n\big);
\end{align}
    \item[$(b)$] follows from since
    \begin{align}
        P_{M_bL_{b-1}K_{1,b-2}K_{2,b-2}|\calC_b^{(n)}}&=\frac{1}{\abs{\calM}\abs{\calL}\abs{\calK_1}\abs{\calK_2}}=\Gamma_{M_bL_{b-1}K_{1,b-2}K_{2,b-2}|\calC_b^{(n)}},\nonumber
    \end{align}
    the codebook $C_N$ is independent of $(M_b,L_{b-1},K_{1,b-2},K_{2,b-2})$, and the symmetry of the codebook construction \ac{wrt} $M_b$, $L_{b-1}$, $K_{1,b-2}$, and $K_{2,b-2}$;
    \item[$(c)$] follows from the law of total expectation.
    \end{itemize}
Now fix the codebook $C_{A_b}^{(n)}=\calC_{A_b}^{(n)}$ and consider the following quantity,
\begin{align}
    &\bbE_{C_{U_b}^{(n)},C_{V_b}^{(n)}|C_{A_b}^{(n)}=\calC_{A_b}^{(n)}}\left[\ToV\left(P_{A^nS_b^n|M_b=1,L_{b-1}=1,K_{b-2}=1,C_b^{(n)}},\Gamma_{A^nS_b^n|M_b=1,L_{b-1}=1,K_{b-2}=1,C_b^{(n)}}\right)\left|C_{A_b}^{(n)}=\calC_{A_b}^{(n)}\right.\right]\nonumber\\
    &\quad\mathop=\limits^{(a)}\bbE_{C_{U_b}^{(n)},C_{V_b}^{(n)}|C_{A_b}^{(n)}=\calC_{A_b}^{(n)}}\left[\ToV\left(P_{S_b^n|M_b=1,L_{b-1}=1,K_{b-2}=1,a^n(1,1),C_b^{(n)}},\right.\right.\nonumber\\
    &\left.\left.\qquad\Gamma_{S_b^n|M_b=1,L_{b-1}=1,K_{b-2}=1,a^n(1,1),C_b^{(n)}}\right)\left|C_{A_b}^{(n)}=\calC_{A_b}^{(n)}\right.\right]\nonumber\\
    &\quad=\bbE_{C_{U_b}^{(n)},C_{V_b}^{(n)}|C_{A_b}^{(n)}=\calC_{A_b}^{(n)}}\left[\ToV\left(Q_{S|A}^{\otimes n}\big(\cdot|a^n(1,1)\big),\Gamma_{S^n|a^n(1,1),C_b^{(n)}}\right)\left|C_{A_b}^{(n)}=\calC_{A_b}^{(n)}\right.\right],\label{eq:First_Level_TV_C}
\end{align}where $(a)$ follows since,
\begin{align}
    P_{A^n|M_bL_{b-1}K_{1,b-2}K_{2,b-2}\calC_b^{(n)}}&=\indi{1}_{\{A^n=a^n(M_b,K_{1,b-2})\}}=\Gamma_{A^n|M_bL_{b-1}K_{1,b-2}K_{2,b-2}\calC_b^{(n)}}.\nonumber
\end{align}
By \cite[Corollary~VII.5]{Cuff13}, the \ac{RHS} of \eqref{eq:First_Level_TV_C} vanishes when $n$ grows if
    \begin{align}
    R_J&>\bbI(V;S|A).\label{eq:Conditional_SCL_C}
\end{align}
Since for every $C_{A_b}^{(n)}=\calC_{A_b}^{(n)}$ the \ac{RHS} of \eqref{eq:General_TV_C} vanishes when $n$ grows and the total variation distance is non-negative, the expectation over $C_{A_b}^{(n)}$, in \eqref{eq:General_TV_C}, also vanishes when $n\to\infty$. 
To bound the second term on the \ac{RHS} of \eqref{eq:General_KLD_C}, by using Pinker's inequality in 
Lemma~\ref{lemma:KLD_TV} and considering \eqref{eq:Encoding_GZ_C}, we have
\begin{align}
    &\bbE_{C_b^{(n)}}\bbD\left(\Gamma_{Z_b^nK_{b-1}L_bK_b|\calC_b^{(n)}}||Q_Z^{\otimes n}P_{K_{b-1}}^UP_{L_b}^UP_{K_b}^U\right)\nonumber\\
    &=\bbE_{C_b^{(n)}}\left[\sum\limits_{(z^n,k_{b-1},\ell_b,k_b)}\Gamma_{Z_b^nK_{b-1}L_bK_b|\calC_b^{(n)}}(z^n,k_{b-1},\ell_b,k_b)\log\frac{\Gamma_{Z_b^nK_{b-1}L_bK_b|\calC_b^{(n)}}(z^n,k_{b-1},\ell_b,k_b)}{Q_Z^{\otimes n}(z^n)P_{K_{b-1}}^U(k_{b-1})P_{L_b}^U(\ell_b)P_{K_b}^U(k_b)}\right]\nonumber\\
    &=\bbE_{C_b^{(n)}}\left[\sum\limits_{(z^n,k_{b-1},\ell_b,k_b)}\sum\limits_{m_b}\sum\limits_{k_{1,b-2}}\sum\limits_{\ell_{b-1}}\sum\limits_{k_{2,b-2}}\sum\limits_{j_b}\frac{1}{2^{n(R+R_J+R_L+2R_K)}}\times\right.\nonumber\\
    &\quad W_{Z|AVU}^{\otimes n}\Big(z_b^n|a^n(m_b,k_{1,b-2}),v^n(j_b),u^n(m_b,k_{1,b-2},\ell_{b-1},k_{2,b-2})\Big)\indi{1}_{\big\{\ell_b=\Psi_L\big(v^n(j_b)\big)\big\}\bigcap\big\{k_b=\Psi_K\big(v^n(j_b)\big)\big\}}\nonumber\\    
    &\quad\log\left[\frac{1}{2^{n(R+R_J+R_L+R_K-\tilde{R}_L-\tilde{R}_K)}Q_Z^{\otimes n}(z)}\right.\times\nonumber\\
    &\quad\sum\limits_{\tilde{m}_b}\sum\limits_{\tilde{k}_{1,b-2}}\sum\limits_{\tilde{\ell}_{b-1}}\sum\limits_{\tilde{k}_{2,b-2}}\sum\limits_{\tilde{j}_b}W_{Z|AVU}^{\otimes n}\Big(z_b^n|A^n(\tilde{m}_b,\tilde{k}_{1,b-2}),V^n(\tilde{j}_b),U^n(\tilde{m}_b,\tilde{k}_{1,b-2},\tilde{\ell}_{b-1},\tilde{k}_{2,b-2})\Big)\times\nonumber\\
    &\quad\left.\left.\indi{1}_{\big\{\ell_b=\Psi_L\big(V^n(\tilde{j}_b)\big)\big\}\bigcap\big\{k_b=\Psi_K\big(V^n(\tilde{j}_b)\big)\big\}}\right]\right]\nonumber\\
    &\mathop\le\limits^{(a)}\frac{1}{2^{n(R+R_J+R_L+2R_K)}}\sum\limits_{(z^n,k_{b-1},\ell_b,k_b)}\sum\limits_{m_b}\sum\limits_{k_{1,b-2}}\sum\limits_{\ell_{b-1}}\sum\limits_{k_{2,b-2}}\sum\limits_{j_b}\sum\limits_{(a^n,v^n,u^n)}\nonumber\\
    &\quad\Gamma_{ZAVU}^{\otimes n}\Big(z_b^n,a^n(m_b,k_{1,b-2}),v^n(j_b),u^n(m_b,k_{1,b-2},\ell_{b-1},k_{2,b-2})\Big)\times\nonumber\\
    &\quad\bbE_{\Psi_L\big(v^n(j_b)\big)}\left[\indi{1}_{\big\{\ell_b=\Psi_L\big(v^n(j_b)\big)\big\}}\right]\bbE_{\Psi_K\big(v^n(j_b)\big)}\left[\indi{1}_{\big\{k_b=\Psi_K\big(v^n(j_b)\big)\big\}}\right]\times\nonumber\\
    &\quad\log\bbE_{\mathop {\backslash (m_b,k_{1,b-2},\ell_{b-1},k_{2,b-2},j_b),}\limits_{\backslash(\Psi_L (v^n(j_b)),\Psi_K (v^n(j_b)))} }\left[\frac{1}{2^{n(R+R_J+R_L+R_K-\tilde{R}_L-\tilde{R}_K)}Q_Z^{\otimes n}(z)}\times\right.\nonumber\\
    &\quad\sum\limits_{\tilde{m}_b}\sum\limits_{\tilde{k}_{1,b-2}}\sum\limits_{\tilde{\ell}_{b-1}}\sum\limits_{\tilde{k}_{2,b-2}}\sum\limits_{\tilde{j}_b}W_{Z|AVU}^{\otimes n}\Big(z_b^n|A^n(\tilde{m}_b,\tilde{k}_{1,b-2}),V^n(\tilde{j}_b),U^n(\tilde{m}_b,\tilde{k}_{1,b-2},\tilde{\ell}_{b-1},\tilde{k}_{2,b-2})\Big)\times\nonumber\\
    &\left.\quad\indi{1}_{\big\{\ell_b=\Psi_L\big(V^n(\tilde{j}_b)\big)\big\}\bigcap\big\{k_b=\Psi_K\big(V^n(\tilde{j}_b)\big)\big\}}\right]\nonumber\\
    &\mathop\le\limits^{(b)}\frac{1}{2^{n(R+R_J+R_L+2R_K)}}\sum\limits_{(z^n,k_{b-1},\ell_b,k_b)}\sum\limits_{m_b}\sum\limits_{k_{1,b-2}}\sum\limits_{\ell_{b-1}}\sum\limits_{k_{2,b-2}}\sum\limits_{j_b}\sum\limits_{(a^n,v^n,u^n)}\nonumber\\
    &\quad\Gamma_{ZAVU}^{\otimes n}\Big(z_b^n,a^n(m_b,k_{1,b-2}),v^n(j_b),u^n(m_b,k_{1,b-2},\ell_{b-1},k_{2,b-2})\Big)\frac{1}{2^{n(\tilde{R}_L+\tilde{R}_K)}}\times\nonumber\\
    &\quad\log\frac{1}{2^{n(R+R_J+R_L+R_K-\tilde{R}_L-\tilde{R}_K)}Q_Z^{\otimes n}(z)}\bbE_{\mathop {\backslash (m_b,k_{1,b-2},\ell_{b-1},k_{2,b-2},j_b),}\limits_{\backslash(\Psi_L (v^n(j_b)),\Psi_K (v^n(j_b)))} }\Big[\nonumber\\
    &\quad W_{Z|AVU}^{\otimes n}\Big(z_b^n|a^n(m_b,k_{1,b-2}),v^n(j_b),u^n(m_b,k_{1,b-2},\ell_{b-1},k_{2,b-2})\Big)+\nonumber\\
    &\quad\sum\limits_{\tilde{j}_b\ne j_b}W_{Z|AVU}^{\otimes n}\Big(z_b^n|a^n(m_b,k_{1,b-2}),V^n(\tilde{j}_b),u^n(m_b,k_{1,b-2},\ell_{b-1},k_{2,b-2})\Big)\times\nonumber\\
    &\quad\indi{1}_{\big\{\ell_b=\Psi_L\big(V^n(\tilde{j}_b)\big)\big\}\bigcap\big\{k_b=\Psi_K\big(V^n(\tilde{j}_b)\big)\big\}}+\nonumber\\
    &\quad\sum\limits_{(\tilde{\ell}_{b-1},\tilde{k}_{2,b-2})\ne(\ell_{b-1},k_{2,b-2})}W_{Z|AVU}^{\otimes n}\Big(z_b^n|a^n(m_b,k_{1,b-2}),v^n(j_b),U^n(m_b,k_{1,b-2},\tilde{\ell}_{b-1},\tilde{k}_{2,b-2})\Big)+\nonumber\\
    &\quad\sum\limits_{(\tilde{\ell}_{b-1},\tilde{k}_{2,b-2},\tilde{j}_b)\ne(\ell_{b-1},k_{2,b-2},j_b)}W_{Z|AVU}^{\otimes n}\Big(z_b^n|a^n(m_b,k_{1,b-2}),V^n(\tilde{j}_b),U^n(m_b,k_{1,b-2},\tilde{\ell}_{b-1},\tilde{k}_{2,b-2})\Big)\times\nonumber\\
    &\quad\indi{1}_{\big\{\ell_b=\Psi_L\big(V^n(\tilde{j}_b)\big)\big\}\bigcap\big\{k_b=\Psi_K\big(V^n(\tilde{j}_b)\big)\big\}}+\nonumber\\
    &\quad\sum\limits_{(\tilde{m}_b,\tilde{k}_{1,b-2})\ne (m_b,k_{1,b-2})}\Squad\sum\limits_{(\tilde{\ell}_{b-1},\tilde{k}_{2,b-2})}W_{Z|AVU}^{\otimes n}\Big(z_b^n|A^n(\tilde{m}_b,\tilde{k}_{1,b-2}),v^n(j_b),U^n(\tilde{m}_b,\tilde{k}_{1,b-2},\tilde{\ell}_{b-1},\tilde{k}_{2,b-2})\Big)+\nonumber\\
    &\quad\sum\limits_{(\tilde{m}_b,\tilde{k}_{1,b-2},\tilde{j}_b)\ne (m_b,k_{1,b-2},j_b)}\Squad\sum\limits_{(\tilde{\ell}_{b-1},\tilde{k}_{2,b-2})}\hspace{-2mm}W_{Z|AVU}^{\otimes n}\Big(z_b^n|A^n(\tilde{m}_b,\tilde{k}_{1,b-2}),V^n(\tilde{j}_b),U^n(\tilde{m}_b,\tilde{k}_{1,b-2},\tilde{\ell}_{b-1},\tilde{k}_{2,b-2})\Big)\times\nonumber\\
    &\quad\indi{1}_{\big\{\ell_b=\Psi_L\big(V^n(\tilde{j}_b)\big)\big\}\bigcap\big\{k_b=\Psi_K\big(V^n(\tilde{j}_b)\big)\big\}}\Big]\nonumber\\
    &\le\frac{1}{2^{n(R+R_J+R_L+2R_K+\tilde{R}_L+\tilde{R}_K)}}\sum\limits_{(z^n,k_{b-1},\ell_b,k_b)}\sum\limits_{m_b}\sum\limits_{k_{1,b-2}}\sum\limits_{\ell_{b-1}}\sum\limits_{k_{2,b-2}}\sum\limits_{j_b}\sum\limits_{(a^n,v^n,u^n)}\nonumber\\
    &\quad\Gamma_{ZAVU}^{\otimes n}\Big(z_b^n,a^n(m_b,k_{1,b-2}),v^n(j_b),u^n(m_b,k_{1,b-2},\ell_{b-1},k_{2,b-2})\Big)\times\nonumber\\
    &\quad\log\frac{1}{2^{n(R+R_J+R_L+R_K-\tilde{R}_L-\tilde{R}_K)}Q_Z^{\otimes n}(z)}\left[W_{Z|AVU}^{\otimes n}\Big(z_b^n|a^n(m_b,k_{1,b-2}),v^n(j_b),u^n(m_b,k_{1,b-2},\ell_{b-1},k_{2,b-2})\Big)\right.+\nonumber\\
    &\quad\sum\limits_{\tilde{j}_b\ne j_b}W_{Z|AU}^{\otimes n}\Big(z_b^n|a^n(m_b,k_{1,b-2}),u^n(m_b,k_{1,b-2},\ell_{b-1},k_{2,b-2})\Big)2^{-n(\tilde{R}_L+\tilde{R}_K)}+\nonumber\\
    &\quad\sum\limits_{(\tilde{\ell}_{b-1},\tilde{k}_{2,b-2})\ne(\ell_{b-1},k_{2,b-2})}W_{Z|AV}^{\otimes n}\Big(z_b^n|a^n(m_b,k_{1,b-2}),v^n(j_b)\Big)+\nonumber\\
    &\quad\sum\limits_{(\tilde{\ell}_{b-1},\tilde{k}_{2,b-2},\tilde{j}_b)\ne(\ell_{b-1},k_{2,b-2},j_b)}W_{Z|A}^{\otimes n}\Big(z_b^n|a^n(m_b,k_{1,b-2})\Big)2^{-n(\tilde{R}_L+\tilde{R}_K)}+\nonumber\\
    &\quad\left.\sum\limits_{(\tilde{m}_b,\tilde{k}_{1,b-2})\ne (m_b,k_{1,b-2})}\Squad\sum\limits_{(\tilde{\ell}_{b-1},\tilde{k}_{2,b-2})}W_{Z|V}^{\otimes n}\Big(z_b^n|v^n(j_b)\Big)+1\right]\nonumber\\
    &\le\frac{1}{2^{n(R+R_J+R_L+2R_K+\tilde{R}_L+\tilde{R}_K)}}\sum\limits_{(z^n,k_{b-1},\ell_b,k_b)}\sum\limits_{m_b}\sum\limits_{k_{1,b-2}}\sum\limits_{\ell_{b-1}}\sum\limits_{k_{2,b-2}}\sum\limits_{j_b}\sum\limits_{(a^n,v^n,u^n)}\nonumber\\
    &\quad\Gamma_{ZAVU}^{\otimes n}\Big(z_b^n,a^n(m_b,k_{1,b-2}),v^n(j_b),u^n(m_b,k_{1,b-2},\ell_{b-1},k_{2,b-2})\Big)\times\nonumber\\
    &\quad\log\frac{1}{2^{n(R+R_J+R_L+R_K-\tilde{R}_L-\tilde{R}_K)}Q_Z^{\otimes n}(z)}\Big[\nonumber\\
    &\quad W_{Z|AVU}^{\otimes n}\Big(z_b^n|a^n(m_b,k_{1,b-2}),v^n(j_b),u^n(m_b,k_{1,b-2},\ell_{b-1},k_{2,b-2})\Big)+\nonumber\\
    &\quad 2^{n(R_J-\tilde{R}_L-\tilde{R}_K)}W_{Z|AU}^{\otimes n}\Big(z_b^n|a^n(m_b,k_{1,b-2}),u^n(m_b,k_{1,b-2},\ell_{b-1},k_{2,b-2})\Big)+\nonumber\\
    &\quad2^{n(R_L+R_{K_2})}W_{Z|AV}^{\otimes n}\Big(z_b^n|a^n(m_b,k_{1,b-2}),v^n(j_b)\Big)+\nonumber\\
    &\quad\left.2^{n(R_L+R_{K_2}+R_J-\tilde{R}_L-\tilde{R}_K)}W_{Z|A}^{\otimes n}\Big(z_b^n|a^n(m_b,k_{1,b-2})\Big)+2^{n(R+R_{K_1}+R_L+R_{K_2})}W_{Z|V}^{\otimes n}\Big(z_b^n|v^n(j_b)\Big)+1\right]\nonumber\\
    &\le\frac{1}{2^{n(R+R_J+R_L+2R_K+\tilde{R}_L+\tilde{R}_K)}}\sum\limits_{(z^n,k_{b-1},\ell_b,k_b)}\sum\limits_{m_b}\sum\limits_{k_{1,b-2}}\sum\limits_{\ell_{b-1}}\sum\limits_{k_{2,b-2}}\sum\limits_{j_b}\sum\limits_{(a^n,v^n,u^n)}\nonumber\\
    &\quad\Gamma_{ZAVU}^{\otimes n}\Big(z_b^n,a^n(m_b,k_{1,b-2}),v^n(j_b),u^n(m_b,k_{1,b-2},\ell_{b-1},k_{2,b-2})\Big)\times\nonumber\\
    &\quad\log\left[\frac{W_{Z|AVU}^{\otimes n}\Big(z_b^n|a^n(m_b,k_{1,b-2}),v^n(j_b),u^n(m_b,k_{1,b-2},\ell_{b-1},k_{2,b-2})\Big)}{2^{n(R+R_J+R_L+R_K-\tilde{R}_L-\tilde{R}_K)}Q_Z^{\otimes n}(z)}\right.+\nonumber\\
    &\quad\frac{W_{Z|AU}^{\otimes n}\Big(z_b^n|a^n(m_b,k_{1,b-2}),u^n(m_b,k_{1,b-2},\ell_{b-1},k_{2,b-2})\Big)}{2^{n(R+R_L+R_K)}Q_Z^{\otimes n}(z)}+\nonumber\\
    &\quad\left.\frac{W_{Z|AV}^{\otimes n}\Big(z_b^n|a^n(m_b,k_{1,b-2}),v^n(j_b)\Big)}{2^{n(R+R_J+R_{K_1}-\tilde{R}_L-\tilde{R}_K)}Q_Z^{\otimes n}(z)}+\frac{W_{Z|A}^{\otimes n}\Big(z_b^n|a^n(m_b,k_{1,b-2})\Big)}{2^{n(R+R_{K_1})}Q_Z^{\otimes n}(z)}+\frac{W_{Z|V}^{\otimes n}\Big(z_b^n|v^n(j_b)\Big)}{2^{n(R_J-\tilde{R}_L-\tilde{R}_K)}Q_Z^{\otimes n}(z)}+1\right]\nonumber\\
    &\mathop=\limits^{(c)}\Delta_1+\Delta_2\label{eq:Bounding_General_KLD_C}
\end{align}where
\begin{itemize}
    \item[$(a)$] follows from Jensen's inequality;
    \item[$(b)$] follows by expanding the summation in the argument of the $\log$ function and considering $\indi{1}_{\{\cdot\}}\le1$;
    \item[$(c)$] follows by defining $\Delta_1$ and $\Delta_2$ as
    \end{itemize}
\begin{align}
    \Delta_1&\triangleq\frac{1}{2^{n(R+R_J+R_L+2R_K+\tilde{R}_L+\tilde{R}_K)}}\sum\limits_{(k_{b-1},\ell_b,k_b)}\sum\limits_{m_b}\sum\limits_{k_{1,b-2}}\sum\limits_{\ell_{b-1}}\sum\limits_{k_{2,b-2}}\sum\limits_{j_b}\nonumber\\
    &\sum\limits_{(z^n,a^n(m_b,k_{1,b-2}),v^n(j_b),u^n(m_b,k_{1,b-2},\ell_{b-1},k_{2,b-2}))\in\calT_\epsilon^{(n)}}\hspace{-25mm}\Gamma_{ZAVU}^{\otimes n}\Big(z_b^n,a^n(m_b,k_{1,b-2}),v^n(j_b),u^n(m_b,k_{1,b-2},\ell_{b-1},k_{2,b-2})\Big)\times\nonumber\\
    &\quad\log\left[\frac{W_{Z|AVU}^{\otimes n}\Big(z_b^n|a^n(m_b,k_{1,b-2}),v^n(j_b),u^n(m_b,k_{1,b-2},\ell_{b-1},k_{2,b-2})\Big)}{2^{n(R+R_J+R_L+R_K-\tilde{R}_L-\tilde{R}_K)}Q_Z^{\otimes n}(z)}\right.+\nonumber\\
    &\quad\frac{W_{Z|AU}^{\otimes n}\Big(z_b^n|a^n(m_b,k_{1,b-2}),u^n(m_b,k_{1,b-2},\ell_{b-1},k_{2,b-2})\Big)}{2^{n(R+R_L+R_K)}Q_Z^{\otimes n}(z)}+\nonumber\\
    &\quad\left.\frac{W_{Z|AV}^{\otimes n}\Big(z_b^n|a^n(m_b,k_{1,b-2}),v^n(j_b)\Big)}{2^{n(R+R_J+R_{K_1}-\tilde{R}_L-\tilde{R}_K)}Q_Z^{\otimes n}(z)}+\frac{W_{Z|A}^{\otimes n}\Big(z_b^n|a^n(m_b,k_{1,b-2})\Big)}{2^{n(R+R_{K_1})}Q_Z^{\otimes n}(z)}+\frac{W_{Z|V}^{\otimes n}\Big(z_b^n|v^n(j_b)\Big)}{2^{n(R_J-\tilde{R}_L-\tilde{R}_K)}Q_Z^{\otimes n}(z)}+1\right]\nonumber\\
    &\le\frac{1}{2^{n(R+R_J+R_L+2R_K+\tilde{R}_L+\tilde{R}_K)}}\sum\limits_{(k_{b-1},\ell_b,k_b)}\sum\limits_{m_b}\sum\limits_{k_{1,b-2}}\sum\limits_{\ell_{b-1}}\sum\limits_{k_{2,b-2}}\sum\limits_{j_b}\nonumber\\
    &\sum\limits_{(z^n,a^n(m_b,k_{1,b-2}),v^n(j_b),u^n(m_b,k_{1,b-2},\ell_{b-1},k_{2,b-2}))\in\calT_\epsilon^{(n)}}\hspace{-25mm}\Gamma_{ZAVU}^{\otimes n}\Big(z_b^n,a^n(m_b,k_{1,b-2}),v^n(j_b),u^n(m_b,k_{1,b-2},\ell_{b-1},k_{2,b-2})\Big)\times\nonumber\\
    &\quad\log\left[\frac{2^{-n(1-\epsilon)\bbH(Z|A,V,U)}}{2^{n(R+R_J+R_L+R_K-\tilde{R}_L-\tilde{R}_K)}2^{-n(1+\epsilon)\bbH(Z)}}\right.+\frac{2^{-n(1-\epsilon)\bbH(Z|A,U)}}{2^{n(R+R_L+R_K)}2^{-n(1+\epsilon)\bbH(Z)}}+\nonumber\\
    &\quad\left.\frac{2^{-n(1-\epsilon)\bbH(Z|A,V)}}{2^{n(R+R_J+R_{K_1}-\tilde{R}_L-\tilde{R}_K)}2^{-n(1+\epsilon)\bbH(Z)}}+\frac{2^{-n(1-\epsilon)\bbH(Z|A)}}{2^{n(R+R_{K_1})}2^{-n(1+\epsilon)\bbH(Z)}}+\frac{2^{-n(1-\epsilon)\bbH(Z|V)}}{2^{n(R_J-\tilde{R}_L-\tilde{R}_K)}2^{-n(1+\epsilon)\bbH(Z)}}+1\right],\label{eq:Si1_C}\\
    \Delta_2&\triangleq\frac{1}{2^{n(R+R_J+R_L+2R_K+\tilde{R}_L+\tilde{R}_K)}}\sum\limits_{(k_{b-1},\ell_b,k_b)}\sum\limits_{m_b}\sum\limits_{k_{1,b-2}}\sum\limits_{\ell_{b-1}}\sum\limits_{k_{2,b-2}}\sum\limits_{j_b}\nonumber\\
    &\sum\limits_{(z^n,a^n(m_b,k_{1,b-2}),v^n(j_b),u^n(m_b,k_{1,b-2},\ell_{b-1},k_{2,b-2}))\notin\calT_\epsilon^{(n)}}\hspace{-25mm}\Gamma_{ZAVU}^{\otimes n}\Big(z_b^n,a^n(m_b,k_{1,b-2}),v^n(j_b),u^n(m_b,k_{1,b-2},\ell_{b-1},k_{2,b-2})\Big)\times\nonumber\\
    &\quad\log\left[\frac{W_{Z|AVU}^{\otimes n}\Big(z_b^n|a^n(m_b,k_{1,b-2}),v^n(j_b),u^n(m_b,k_{1,b-2},\ell_{b-1},k_{2,b-2})\Big)}{2^{n(R+R_J+R_L+R_K-\tilde{R}_L-\tilde{R}_K)}Q_Z^{\otimes n}(z)}\right.+\nonumber\\
    &\quad\frac{W_{Z|AU}^{\otimes n}\Big(z_b^n|a^n(m_b,k_{1,b-2}),u^n(m_b,k_{1,b-2},\ell_{b-1},k_{2,b-2})\Big)}{2^{n(R+R_L+R_K)}Q_Z^{\otimes n}(z)}+\nonumber\\
    &\quad\left.\frac{W_{Z|AV}^{\otimes n}\Big(z_b^n|a^n(m_b,k_{1,b-2}),v^n(j_b)\Big)}{2^{n(R+R_J+R_{K_1}-\tilde{R}_L-\tilde{R}_K)}Q_Z^{\otimes n}(z)}+\frac{W_{Z|A}^{\otimes n}\Big(z_b^n|a^n(m_b,k_{1,b-2})\Big)}{2^{n(R+R_{K_1})}Q_Z^{\otimes n}(z)}+\frac{W_{Z|V}^{\otimes n}\Big(z_b^n|v^n(j_b)\Big)}{2^{n(R_J-\tilde{R}_L-\tilde{R}_K)}Q_Z^{\otimes n}(z)}+1\right]\nonumber\\
    &\le2\abs{A}\abs{U}\abs{V}\abs{Z}e^{-n\epsilon^2\mu_{A,V,U,Z}}n\log\left(\frac{5}{\mu_Z}+1\right),
\end{align}where $\mu_{A,U,V,Z}=\min\limits_{(a,u,v,z)\in(\calA,\calU,\calV,\calZ)}\Gamma_{AUVZ}(a,u,v,z)$ and $\mu_Z=\min\limits_{z\in\calZ}\Gamma_{Z}(z)$.
When $n\to\infty$ then $\Delta_2\to 0$, and $\Delta_1\to 0$ when
\begin{subequations}\label{eq:resol_C_2}
\begin{align}
    R+R_J+R_L+R_K-\tilde{R}_L-\tilde{R}_K&>\bbI(A,V,U;Z),\\
    R+R_L+R_K&>\bbI(A,U;Z),\\
    R+R_J+R_{K_1}-\tilde{R}_L-\tilde{R}_K&>\bbI(A,V;Z),\\
    R+R_{K_1}&>\bbI(A;Z),\\
    R_J-\tilde{R}_L-\tilde{R}_K&>\bbI(V;Z).
\end{align}
\end{subequations}

\subsection{Decoding}
The following lemma is essential to analyze the probability of error.
\begin{lemma}[Typicality]
\label{lemma:Typicaity_C}
If $(R_K,R_L,R_J)\in\bbR_+^3$ satisfy the constraint in \eqref{eq:Conditional_SCL_C} then for any $m_b\in\calM$ and $\epsilon>0$ we have
\begin{align}
    \bbE_{C_b^{(n)}}\bbP_P\big[\big(A^n(m_b,k_{1,b-2}),S_b^n,V^n(J_b)\big)\notin\calT_\epsilon^{(n)}\big]\xrightarrow[]{n\to\infty}0.\label{eq:lemma_typicality_C}
\end{align}
\end{lemma}The proof of Lemma~\ref{lemma:Typicaity_C} is similar to the proof of  Lemma~\ref{lemma:Typicaity}.
After receiving $Y_b^n$, given the shared key $k_{b-2}\triangleq(k_{1,b-2},k_{2,b-2})$, the decoder looks for the smallest value of $(\hat{m}_b,\hat{\ell}_{b-1},\hat{j}_b)$ such that $\left(A^n(\hat{m}_b,k_{1,b-2}),U^n\big(\hat{m}_b,k_{1,b-2},\hat{\ell}_{b-1},k_{2,b-2}\big),V^n(\hat{j}_b),Y_b^n\right)\in\calT_{\epsilon}^{(n)}(P_{AUVY})$ and choose $(\hat{M}_b,\hat{L}_{b-1},\hat{J}_b)=(1,1,1)$ if such a $(\hat{m}_b,\hat{\ell}_{b-1},\hat{j}_b)$ does not exist. Let,
\begin{subequations}
    \begin{align}
    \calE&\triangleq\left\{M\ne\hat{M}\right\},\label{eq:General_C_Pe}\\
    \calE_b&\triangleq\left\{M_b\ne\hat{M}_b\right\},\label{eq:General_C_Pe_b}\\
    \calE_{1,b}&\triangleq\left\{\left(A^n(\hat{m}_b,k_{1,b-2}),U^n\big(\hat{m}_b,k_{1,b-2},\hat{\ell}_{b-1},k_{2,b-2}\big),S_b^n\right)\notin\calT_{\epsilon}^{(n)}(P_{AUS})\right\},\label{eq:General_C_Enc}\\
    \calE_{2,b}&\triangleq\left\{\left(A^n(\hat{m}_b,k_{1,b-2}),U^n\big(\hat{m}_b,k_{1,b-2},\hat{\ell}_{b-1},k_{2,b-2}\big),Y_b^n\right)\notin\calT_{\epsilon}^{(n)}(P_{AUY})\right\},\label{eq:General_C_Dec1}\\
    \calE_{3,b}&\triangleq\left\{\left(A^n(\hat{m}_b,k_{1,b-2}),U^n\big(\hat{m}_b,k_{1,b-2},\hat{\ell}_{b-1},k_{2,b-2}\big),Y_b^n\right)\in\calT_{\epsilon}^{(n)}(P_{AUY}),\Squad\text{for some}\right.\nonumber\\
    &\qquad\qquad\left.\big(\hat{m}_b,\hat{\ell}_b\big)\ne (M_b,L_b)\right\}.\label{eq:General_C_Dec2}
\end{align}
\end{subequations}Now, we bound the probability of error as
\begin{align}
    \bbP(\calE)=\bbP\left(\bigcup_{b=1}^B\calE_b\right)\le\sum\limits_{b=1}^B\bbP(\calE_b),\label{eq:Total_PE_C}
\end{align}where the inequality follows from the union bound. Next for each $b\in[B]$ we bound $\bbP(\calE_b)$ as follows,
\begin{align}
    \bbP(\calE_b)\le\bbP(\calE_{1,b})+\bbP(\calE_{1,b}^c\cap\calE_{2,b})+\bbP(\calE_{1,b}^c\cap\calE_{2,b}^c\cap\calE_{3,b}),\label{eq:Union_Bound_C}
\end{align}from Lemma~\ref{lemma:Typicaity_C} the first term on the \ac{RHS} of \eqref{eq:Union_Bound_C} vanishes when $n$ grows to infinity, the second on the \ac{RHS} of \eqref{eq:Union_Bound_C} vanishes when $n$ grows by the law of large numbers, by the union bound and \cite[Theorem~1.3]{Kramer_Book}, the last term on the \ac{RHS} of \eqref{eq:Union_Bound_C} vanishes when $n$ grows if,
\begin{align}
    R+R_L<\bbI(U;Y).\label{eq:Dec_Constraint_C}
\end{align}
Next, we bound the error probability for the key generation. Let $L_{b-1}$ and $J_{b-1}$ denote the indices that are chosen by the transmitter, and $\hat{L}_{b-1}$ and $\hat{J}_{b-1}$ denote the estimate of these indices at the receiver. The receiver by decoding $U^n$ at the end of block $b$ knows the index $\hat{L}_{b-1}$. 
Let,
\begin{subequations}
\begin{align}
    \tilde{\calE}&\triangleq\left\{\left(V^n(\hat{J}_{b-1}),S_{b-1}^n,A_{b-1}^n,U_{b-1}^n,Y_{b-1}^n\right)\notin\calT_\epsilon^{(n)}\right\},\label{eq:Key_Gen_Err1_C}\\
    \tilde{\calE}_1&\triangleq\left\{\left(V^n(j_{b-1}),A_{b-1}^n,S_{b-1}^n\right)\notin\calT_{\tilde{\epsilon}}^{(n)},\quad\text{for all}\Squad j_{b-1}\in\left[2^{nR_J}\right]\right\},\label{eq:Key_Gen_Err2_C}\\
    \tilde{\calE}_2&\triangleq\left\{\left(V^n(J_{b-1}),S_{b-1}^n,A_{b-1}^n,U_{b-1}^n,Y_{b-1}^n\right)\notin\calT_{\tilde{\epsilon}}^{(n)}\right\},\label{eq:Key_Gen_Err3_C}\\
    \tilde{\calE}_3&\triangleq\left\{\left(V^n(\hat{j}_{b-1}),A_{b-1}^n,U_{b-1}^n,Y_{b-1}^n\right)\in\calT_{\tilde{\epsilon}}^{(n)},\Squad\text{for some}\Squad\hat{j}_b\in\calB(\hat{L}_{b-1}),\hat{j}_b\ne J_{b-1}\right\}.\label{eq:Key_Gen_Err4_C}
\end{align}where $\epsilon>\tilde{\epsilon}>0$. 
\end{subequations}Now by the union bound, 
\begin{align}
    \bbP\big(\tilde{\calE}\big)\le\bbP\big(\tilde{\calE}_1\big)+\bbP\big(\tilde{\calE}_1^c\cap\tilde{\calE}_2\big)+\bbP\big(\tilde{\calE}_3\big).\label{eq:Key_Gen_UB_C}
\end{align}By Lemma~\ref{lemma:Typicaity_C} the first term on the \ac{RHS} of \eqref{eq:Key_Gen_UB_C} vanishes when $n$ grows to infinity if \eqref{eq:Conditional_SCL_C} holds, similar to \cite[Sec.~11.3.1]{ElGamalKim} the second and the third terms on the \ac{RHS} of \eqref{eq:Key_Gen_UB_C} go to zero when $n\to\infty$ if,
\begin{subequations}\label{eq:WZ_Dec_C}
\begin{align}
    R_J&>\bbI(V;S|A),\label{eq:WZ_1_C}\\
    R_J-R_L&<\bbI(V;A,U,Y),\label{eq:WZ_2_C}
\end{align}
\end{subequations}
Applying Fourier-Motzkin elimination procedure, \cite{ElGamalKim,FMEIT}, to eliminate $(R_J,R_L,R_{K_1},R_{K_2},R_K,\tilde{R}_L,\tilde{R}_K)$, in \eqref{eq:Conditional_SCL_C}, \eqref{eq:resol_C_2}, \eqref{eq:Dec_Constraint_C}, and \eqref{eq:WZ_Dec_C} and considering $\tilde{R}_L+\tilde{R}_K\ge R_L+R_K$ and $R_K=R_{K_1}+R_{K_2}$ leads to
\begin{subequations}\label{eq:FME_Output_C}
\begin{align}
R&>\bbI(A;Z)+\bbI(V;Z)-\bbI(V;A,U,Y),\label{eq:FME_Output_C_1}\\
R&>\bbI(A,V;Z)-\bbI(V;A,U,Y),\label{eq:FME_Output_C_2}\\
0&>\bbI(V;Z)-\bbI(V;A,U,Y),\label{eq:FME_Output_C_3}\\
0&>\bbI(A,U;Z)+\bbI(V;Z)-\bbI(A,U;Y)-\bbI(V;A,U,Y),\label{eq:FME_Output_C_4}\\
0&>\bbI(A,V,U;Z)-\bbI(A,U;Y)-\bbI(V;A,U,Y),\label{eq:FME_Output_C_5}\\
R&<\bbI(A,U;Y),\label{eq:FME_Output_C_6}\\
R&<-\bbI(V;S|A)+\bbI(A,U;Y)+\bbI(V;A,U,Y),\label{eq:FME_Output_C_7}
\end{align}
where, since $\bbI(A,U,V;Z)\ge\bbI(A,U;Z)+\bbI(V;Z)$, \eqref{eq:FME_Output_C_4} is redundant because of \eqref{eq:FME_Output_C_5}, and since $\bbI(A,V;Z)\ge\bbI(A;Z)+\bbI(V;Z)$, \eqref{eq:FME_Output_C_1} is redundant because of \eqref{eq:FME_Output_C_2}. Therefore, one can rewrite the rate constraints in \eqref{eq:FME_Output_C} as the rate constraints in Theorem~\ref{thm:Acievability_KG_C}. 
\end{subequations}

\section{Proof of Theorem~\ref{thm:Converse_C}}
\label{proof:thm:Converse_C}
Consider any sequence of codes with length $N$ for channels with \ac{ADSI} when the state is available causally at the transmitter such that $P_e^{(N)}\le\epsilon_N$, $\bbD\left(P_{Z^N}||Q_0^{\otimes N}\right)\le\tilde{\epsilon}$, and $\bar{R}_K/N\triangleq\eta_N$, where $\epsilon_N\xrightarrow[]{N\to\infty}0$ and $\eta_N\xrightarrow[]{N\to\infty}0$. 
\subsection{Epsilon Rate}We first define $\calF^{(\epsilon)}_{\text{U-C}}$ for $\epsilon>0$ which expands the rate defined in \eqref{eq:Converse_AD_C} as
\begin{subequations}\label{eq:Converse_AD_epsilon_C}
\begin{align}
  \calF^{(\epsilon)}_{\text{U-C}} = \left.\begin{cases}R\geq 0: \exists P_{ASUVXYZ}\in\calG^{(\epsilon)}_{\text{U-C}}:\\
  R\le\bbI(U;Y)+\epsilon\\
\end{cases}\right\},\label{eq:Converse_A_epsilon_C}
\end{align}
where
\begin{align}
  \calG^{(\epsilon)}_{\text{U-C}}\triangleq \left.\begin{cases}P_{ASUVXYZ}:\\
P_{ASUVXYZ}=P_AQ_{S|A}P_{U|A}P_{V|US}P_{X|US}W_{YZ|XS}\\
\bbI(U;Y)\ge\bbI(V;Z)-\epsilon\\
\bbD\big(P_Z||Q_0\big)\le\epsilon\\
\end{cases}\right\}.\label{eq:Converse_D_epsilon_C}
\end{align}
\end{subequations}Now we show that any achievable rate $R$ belongs to $\calF^{(\epsilon)}_{\text{U-C}}$, i.e., $R\in\calF^{(\epsilon)}_{\text{U-C}}$. 
For any $\gamma>0$, and $\epsilon_N>0$ we have,
\begin{align}
    NR&=\bbH(M)\nonumber\\
    &=\bbH(M|K)\nonumber\\
    &\mathop\le\limits^{(a)}\bbI\big(M;Y^N|K\big)+N\epsilon_N\nonumber\\
    &=\sum\limits_{t=1}^N\big[\bbI\big(M;Y_t|K,Y^{t-1}\big)\big]+N\epsilon_N\nonumber\\
    &\le\sum\limits_{t=1}^N\big[\bbI\big(M,K,A_t,Y^{t-1};Y_t\big)\big]+N\epsilon_N\nonumber\\
    &\mathop\le\limits^{(b)}\sum\limits_{t=1}^N\big[\bbI\big(M,K,A_t,S^{t-1};Y_t\big)\big]+N\epsilon_N\nonumber\\
    &\mathop=\limits^{(c)}\sum\limits_{t=1}^N\bbI(U_t;Y_t)+N\epsilon_N\nonumber\\
    &=N\sum\limits_{t=1}^N\bbP(T=t)\bbI(U_T;Y_T|T=t)+N\epsilon_N\nonumber\\
    &=N\bbI(U_T;Y_T|T)+N\epsilon_N\nonumber\\
    &\le N\bbI(U_T,T;Y_T)+N\epsilon_N\nonumber\\
    &\mathop=\limits^{(d)} N\bbI(U;Y)+N\epsilon_N\nonumber\\
    &\mathop\le\limits^{(e)} N\bbI(U;Y)+N\delta\label{eq:Traditianl_Converse_C}
\end{align}where
\begin{itemize}
    \item[$(a)$] follows from Fano's inequality;
    \item[$(b)$] follows since $(M,K,A_t,Y^{t-1})-(M,K,A_t,S^{t-1})-Y_t$ forms a Markov chain, note that $V_t-(M,K,A_t,S^{t-1})-Y_t$, where $V_t\triangleq(M,K,Z^{t-1})$, also forms a Markov chain;
    \item[$(c)$] follows by defining $U_t\triangleq\big(M,K,A_t,S^{t-1}\big)$;
    \item[$(d)$] follows by defining $U\triangleq(U_T,T)$ and $Y\triangleq Y_T$;
    \item[$(e)$] follows by defining $\delta\triangleq\max\{\epsilon_N,\eta_N,\gamma\}$.
\end{itemize}
We can also lower bound the rate $NR+\bar{R}_K^{(N)}$ as follows,
\begin{align}
    NR+\bar{R}_K^{(N)}&=\bbH(M,K)\nonumber\\
    &\ge\bbI\big(M,K;Z^N\big)\nonumber\\
    &=\sum\limits_{t=1}^N\big[\bbI\big(M,K;Z_t|Z^{t-1}\big)\big]\nonumber\\
    &\mathop\ge\limits^{(a)}\sum\limits_{t=1}^N\big[\bbI\big(M,K,Z^{t-1};Z_t\big)\big]-\tilde{\epsilon}\nonumber\\
    &\mathop=\limits^{(b)}\sum\limits_{t=1}^N\big[\bbI(V_t;Z_t)\big]-\tilde{\epsilon}\nonumber\\
    &=N\sum\limits_{t=1}^N\big[\bbP(T=t)\bbI(V_T;Z_T|T=t)\big]-\tilde{\epsilon}\nonumber\\
    &=N\bbI(V_T;Z_T|T)-\tilde{\epsilon}\nonumber\\
    &\mathop\ge\limits^{(c)} N\bbI(V_T,T;Z_T)-2\tilde{\epsilon}\nonumber\\
    &\mathop=\limits^{(d)} N\bbI(V;Z)-2\tilde{\epsilon}\label{eq:Non-Resolvability1_Converse_C}
\end{align}where
\begin{itemize}
    \item[$(a)$] and $(c)$ follow from Lemma~\ref{lemma:iid_Time_indep};
    \item[$(b)$] follows by defining  $V_t\triangleq\big(M,K,Z^{t-1}\big)$;
    \item[$(d)$] follows by defining $V\triangleq(V_T,T)$ and  $Z\triangleq Z_T$.
\end{itemize}
For any $\gamma>0$, selecting $N$ large enough ensures that,
\begin{align}
    R+\frac{\bar{R}_K^{(N)}}{N}\ge\bbI(V;Z)-2\gamma.
\end{align}Hence,
\begin{align}
    R&\ge\bbI(V;Z)-2\gamma-\frac{\bar{R}_K}{N}\nonumber\\
    &=\bbI(V;Z)-2\gamma-\eta_N\nonumber\\
    &\ge\bbI(V;Z)-3\delta,\label{eq:Non-Resolvability_Converse_C}
\end{align}where the last inequality follows by defining $\delta\triangleq\max\{\epsilon_N,\eta_N,\gamma\}$.

To prove that $\bbD(P_Z||Q_0)\le\epsilon$, for $N$ large enough we have
\begin{align}
    \bbD(P_Z||Q_0)&=\bbD(P_{Z_T}||Q_0)=\bbD\left(\frac{1}{N}\sum\limits_{t=1}^NP_{Z_t}\Big|\Big|Q_0\right)\le\frac{1}{N}\sum\limits_{t=1}^N\bbD\left(P_{Z_t}\Big|\Big|Q_0\right)\nonumber\\
    &\le\frac{1}{N}\bbD\left(P_{Z^N}\Big|\Big|Q_0^{\otimes N}\right)\le\frac{\tilde{\epsilon}}{N}\le\gamma\le\delta.
\end{align}Combining \eqref{eq:Traditianl_Converse_C} and \eqref{eq:Non-Resolvability_Converse_C} proves that $\forall\epsilon_N,\tilde{\epsilon},\eta_N$, $R\le\max\{R:R\in\calF^{(\epsilon)}_{\text{U-C}}\}$. Hence,
\begin{align}
    R\le\max\left\{R:R\in\bigcap\limits_{\epsilon>0}\calF^{(\epsilon)}_{\text{U-C}}\right\}.
\end{align}
\subsection{Proof for Continuity at Zero}The proof follows similar lines as the proof for continuity at zero in \cite[Appendix~F]{Keyless22}.
\begin{remark}[Converse for More General Channels]
\label{remark:Causal_General_Converse}
For the channels of the form $W_{YZ|XSA}$ we can derive an upper bound similar to the bound that we have in \eqref{eq:Traditianl_Converse_C} except that the auxiliary \ac{RV} $U\triangleq\left(M,K,Y^{t-1},S^{t-1}\right)$. Therefore, we get the same upper bound as that in Theorem~\ref{thm:Converse_C} without the condition $\bbI(U;Y)\ge\bbI(V;Z)$ in \eqref{eq:Converse_D_C} and the joint probability distribution will be $P_AQ_{S|A}P_{U|A}P_{X|ASU}W_{YZ|ASX}$. 
\end{remark}

\section{Proof of Theorem~\ref{thm:Gaussian_Converse}}
\label{proof:thm:Converse_AWGN}
Consider any sequence of codes with length $N$ for channels with \ac{ADSI} when the state is available non-causally at the transmitter such that $P_e^{(N)}\le\epsilon_N$, $\bbD\left(P_{Z^N}||Q_0^{\otimes N}\right)\le\tilde{\epsilon}$, and $R_K/N\le\eta_N$, where $\epsilon_N\xrightarrow[]{n\to\infty}0$ and $\eta_N\xrightarrow[]{n\to\infty}0$. 
Let $\tilde{P}_X\triangleq\bbE[X^2]\le P_X$, $\tilde{P}_A\triangleq\bbE[A^2]\le P_A$, $\Lambda_{XA}\triangleq\bbE[XA]$, $\Lambda_{XS}\triangleq\bbE[XN_S]$, and therefore
\begin{subequations}\label{eq:Second_Moment_Variances}
\begin{align}
    \bar{\sigma}_Y^2&=\bbE(Y^2)=\tilde{P}_A+\tilde{P}_X+T+\sigma_Y^2+2\Lambda_{XA}+2\Lambda_{XS},\label{eq:Variance_Y}\\
    \bar{\sigma}_Z^2&=\bbE(Z^2)=\tilde{P}_A+\tilde{P}_X+T+\sigma_Z^2+2\Lambda_{XA}+2\Lambda_{XS},\label{eq:Variance_Z}\\
    \Sigma_1&=\bbE\left(N_S\left[A\Squad Y\right]\right)=\left[0\Squad \Lambda_{XS}+T\right]=\begin{bmatrix}0&\sqrt{T\tilde{P}_X}\rho_{XS}+T\end{bmatrix},\label{eq:Variance_1}\\
    \Sigma_2&=\bbE\left(N_S\left[A\Squad Y\right]^\intercal\right)=\left[0\Squad \Lambda_{XS}+T\right]^\intercal=\begin{bmatrix}0\\\sqrt{T\tilde{P}_X}\rho_{XS}+T\end{bmatrix},\label{eq:Variance_2}\\
    \Sigma_3&=\bbE\left(\left[A\Squad Y\right]^\intercal\left[A\Squad Y\right]\right)=\begin{bmatrix}\tilde{P}_A&\sqrt{\tilde{P}_A\tilde{P}_X}\rho_{XA}+\tilde{P}_A\\\sqrt{\tilde{P}_A\tilde{P}_X}\rho_{XA}+\tilde{P}_A&\tilde{P}_A+\tilde{P}_X+T+\sigma_Y^2+2\sqrt{\tilde{P}_A\tilde{P}_X}\rho_{XA}+2\sqrt{T\tilde{P}_X}\rho_{XS}\end{bmatrix},\label{eq:Variance_4}\\
    \sigma_{N_S|A,Y}^2&=T-\Sigma_1\Sigma_3^{-1}\Sigma_2=\frac{T\tilde{P}_X\left(1-\rho_{XS}^2-\rho_{XA}^2\right)+T\sigma_Y^2}{\tilde{P}_X\left(1-\rho_{XS}^2-\rho_{XA}^2\right)+\left(\sqrt{T}+\sqrt{\tilde{P}_X}\rho_{XS}\right)^2+\sigma_Y^2}
    \label{eq:Variance_SgAY}
\end{align}
\end{subequations}Now since $Q_0\sim\calN(0,T+\sigma_Z^2)$ the covertness constraint $P_Z=Q_0$, implies that
\begin{align}
    \bar{\sigma}_Z^2=T+\sigma_Z^2\Rightarrow\tilde{P}_X+\tilde{P}_A+T+\sigma_Z^2+2\Lambda_{XA}+2\Lambda_{XS}=T+\sigma_Z^2.\label{eq:Covertness_Constarint_Gaussian_1}
\end{align}Considering $\Lambda_{XS}=\sqrt{T\tilde{P}_X}\rho_{XS}$ and $\Lambda_{XA}=\sqrt{\tilde{P}_X\tilde{P}_A}\rho_{XA}$, the covertness constraint in \eqref{eq:Covertness_Constarint_Gaussian_1} can be rewritten as
\begin{align}
    &\tilde{P}_X+\tilde{P}_A+2\sqrt{\tilde{P}_X\tilde{P}_A}\rho_{XA}+2\sqrt{T\tilde{P}_X}\rho_{XS}=0.\label{eq:Covertness_Constarint_Gaussian}
\end{align}Algebraic manipulation shows that this constraint can be written as
\begin{align}
\tilde{P}_X\left(1-\rho_{XA}^2-\rho_{XS}^2\right)+\left(\sqrt{T}+\sqrt{\tilde{P}_X}\rho_{XS}\right)^2+\left(\sqrt{\tilde{P}_A}+\sqrt{\tilde{P}_X}\rho_{XA}\right)^2-T=0.\label{eq:Useful_ident_Partial_y_2}
\end{align} 
Now from Theorem~\ref{thm:Converse_NC} we have,
\begin{align}
    R&\le\bbI(A,U;Y)-\bbI(U;S|A)\nonumber\\
    &=\bbI(A;Y)+\bbI(U;Y|A)-\bbI(U;S|A)\nonumber\\
    &\le\bbI(A;Y)+\bbI(U;Y,S|A)-\bbI(U;S|A)\nonumber\\
    &=\bbI(A;Y)+\bbI(U;Y|A,S)\nonumber\\
    &\le\bbI(A;Y)+\bbI(U,X;Y|A,S)\nonumber\\
    &\mathop=\limits^{(a)}\bbI(A;Y)+\bbI(X;Y|A,S)\nonumber\\
    &=\bbI(A;Y)+\bbI(X;Y|A,N_S)\nonumber\\
    &=\dent(Y)-\dent(Y|A)+\dent(Y|A,N_S)-\dent(Y|X,A,N_S)\nonumber\\
    &\mathop=\limits^{(b)}\dent(Y)-\bbI(N_S;Y|A)-\dent(N_Y)\nonumber\\
    &=\dent(Y)+\dent(N_S|A,Y)-\dent(N_S|A)-\dent(N_Y)\nonumber\\
    &\mathop=\limits^{(c)}\dent(Y)+\dent(N_S|A,Y)-\dent(N_S)-\dent(N_Y)\nonumber\\
    &\mathop\le\limits^{(d)}\frac{1}{2}\log\left(\frac{\bar{\sigma}_Y^2\sigma_{N_S|A,Y}^2}{T\sigma_Y^2}\right)\nonumber\\
    &\mathop=\limits^{(e)}\frac{1}{2}\log\left(\frac{\bar{\sigma}_Y^2(T-\Sigma_1\Sigma_3^{-1}\Sigma_2)}{T\sigma_Y^2}\right)\nonumber\\
    &\mathop=\limits^{(f)}\frac{1}{2}\log\left(1+\frac{\tilde{P}_X\left(1-\rho_{XS}^2-\rho_{XA}^2\right)}{\sigma_Y^2}\right)+\frac{1}{2}\log\left(1+\frac{\left(\sqrt{\tilde{P}_A}+\sqrt{\tilde{P}_X}\rho_{XA}\right)^2}{\tilde{P}_X\left(1-\rho_{XS}^2-\rho_{XA}^2\right)+\left(\sqrt{T}+\sqrt{\tilde{P}_X}\rho_{XS}\right)^2+\sigma_Y^2}\right)\nonumber\\
    &=\frac{1}{2}\log\left(\frac{\left(\tilde{P}_X\left(1-\rho_{XS}^2-\rho_{XA}^2\right)+\sigma_Y^2\right)}{\sigma_Y^2}\left[1+\frac{\left(\sqrt{\tilde{P}_A}+\sqrt{\tilde{P}_X}\rho_{XA}\right)^2}{\tilde{P}_X\left(1-\rho_{XS}^2-\rho_{XA}^2\right)+\sigma_Y^2+\left(\sqrt{T}+\sqrt{\tilde{P}_X}\rho_{XS}\right)^2}\right]\right)\nonumber\\
    &\mathop\le\limits^{(g)}\frac{1}{2}\log\left(\frac{\left(\tilde{P}_X\left(1-\rho_{XS}^2\right)+\sigma_Y^2\right)}{\sigma_Y^2}\left[1+\frac{\tilde{P}_A}{\tilde{P}_X\left(1-\rho_{XS}^2\right)+\sigma_Y^2+\left(\sqrt{T}+\sqrt{\tilde{P}_X}\rho_{XS}\right)^2}\right]\right)\label{eq:Gaussian_Upper}
\end{align}where
\begin{itemize}
    \item[$(a)$] follows since $U-(X,S)-Y$ forms a Markov-chain;
    \item[$(b)$] follows from \eqref{eq:system_Model_Final};
    \item[$(c)$] follows since $A$ and $N_S$ are independent;
    \item[$(d)$] follows from the maximum differential entropy lemma \cite{ElGamalKim};
    \item[$(e)$] follows from \eqref{eq:Variance_SgAY};
    \item[$(f)$] follows from \eqref{eq:Variance_Y} and \eqref{eq:Variance_SgAY};
    \item[$(g)$] follows by setting $\rho_{XA}=0$ since
    \begin{align}
    &\frac{1}{2}\log\left(\frac{\left(\tilde{P}_X\left(1-\rho_{XS}^2-\rho_{XA}^2\right)+\sigma_Y^2\right)}{\sigma_Y^2}\left[1+\frac{\left(\sqrt{\tilde{P}_A}+\sqrt{\tilde{P}_X}\rho_{XA}\right)^2}{\tilde{P}_X\left(1-\rho_{XS}^2-\rho_{XA}^2\right)+\sigma_Y^2+\left(\sqrt{T}+\sqrt{\tilde{P}_X}\rho_{XS}\right)^2}\right]\right)\nonumber\\
    &=\frac{1}{2}\log\left(\frac{\left(\tilde{P}_X\left(1-\rho_{XS}^2-\rho_{XA}^2\right)+\sigma_Y^2\right)}{\sigma_Y^2}\right.\times\nonumber\\
    &\qquad\left.\left[\frac{\tilde{P}_X\left(1-\rho_{XS}^2-\rho_{XA}^2\right)+\sigma_Y^2+\left(\sqrt{T}+\sqrt{\tilde{P}_X}\rho_{XS}\right)^2+\left(\sqrt{\tilde{P}_A}+\sqrt{\tilde{P}_X}\rho_{XA}\right)^2}{\tilde{P}_X\left(1-\rho_{XS}^2-\rho_{XA}^2\right)+\sigma_Y^2+\left(\sqrt{T}+\sqrt{\tilde{P}_X}\rho_{XS}\right)^2}\right]\right)\nonumber\\
    &=\frac{1}{2}\log\left(\frac{\left(\tilde{P}_X\left(1-\rho_{XS}^2-\rho_{XA}^2\right)+\sigma_Y^2\right)}{\sigma_Y^2}\left[\frac{T+\sigma_Y^2}{\tilde{P}_X\left(1-\rho_{XS}^2-\rho_{XA}^2\right)+\sigma_Y^2+\left(\sqrt{T}+\sqrt{\tilde{P}_X}\rho_{XS}\right)^2}\right]\right);\label{eq:equivalent_Exp}
    \end{align}where the last equality follows from the covertness constraint \eqref{eq:Useful_ident_Partial_y_2}.
\end{itemize}
\begin{figure*}
\centering
\includegraphics[width=10cm]{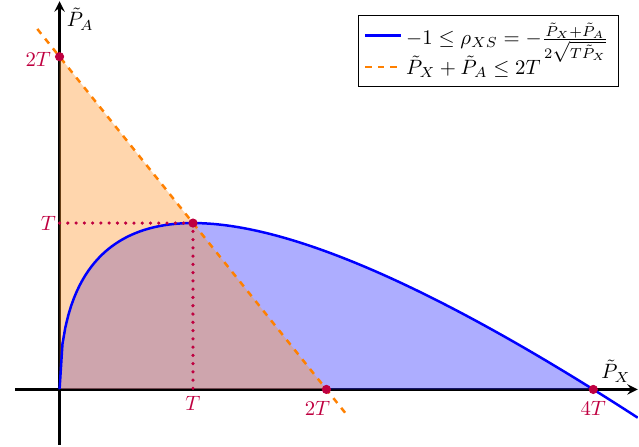}
\caption{The blue region corresponds to $-1\le\rho_{XS}=-\frac{\tilde{P}_X+\tilde{P}_A}{2\sqrt{T\tilde{P}_X}}\le0$, note that since $\tilde{P}_X\ge0$, $\tilde{P}_A\ge0$, and $T\ge0$ the condition $\rho_{XS}=-\frac{\tilde{P}_X+\tilde{P}_A}{2\sqrt{T\tilde{P}_X}}\le0$ is satisfied; the orange region corresponds to the optimal $\tilde{P}_X+\tilde{P}_A\le2T$, note that since we aim to achieve a higher rate with the least amounts of power inputs it is not optimal to go beyond $\tilde{P}_X+\tilde{P}_A=2T$.}
\label{fig:Admis_XY}
\end{figure*}
Now, one can show that the argument of the $\log$ function in \eqref{eq:equivalent_Exp} is decreasing in $\rho_{XA}^2$ and therefore the optimal choice for $\rho_{XA}$ is $0$. Now, setting $\rho_{XA}=0$ the covertness constraint \eqref{eq:Covertness_Constarint_Gaussian} reduces to $\tilde{P}_X+\tilde{P}_A+2\sqrt{T\tilde{P}_X}\rho_{XS}=0$, where $-1\le\rho_{XS}\le0$, which can be written as  
\begin{align}
    \rho_{XS}=-\frac{\tilde{P}_X+\tilde{P}_A}{2\sqrt{T\tilde{P}_X}}.\label{eq:Covertness_Constarint_Gaussian_2}
\end{align}Substituting \eqref{eq:Covertness_Constarint_Gaussian_2} in \eqref{eq:Gaussian_Upper} leads to,
\begin{align}
    R\le\frac{1}{2}\log\left(\frac{\left(T+\sigma_Y^2\right)\left(4T\tilde{P}_X-(\tilde{P}_X+\tilde{P}_A)^2+4T\sigma_Y^2\right)}{4T\sigma_Y^2\left(T+\sigma_Y^2-\tilde{P}_A\right)}\right).\label{eq:Gaussian_Upper_2}
\end{align}
Now, taking the derivative \ac{wrt} $\tilde{P}_X$ and $\tilde{P}_A$ one can show that \eqref{eq:Gaussian_Upper_2} is maximized when $\tilde{P}_X+\tilde{P}_A=2T$. However, since we must have $-1\le\rho_{XS}=-\frac{\tilde{P}_X+\tilde{P}_A}{2\sqrt{T\tilde{P}_X}}\le0$, every $\tilde{P}_X\ge0$ and $\tilde{P}_A\ge0$ that satisfy $\tilde{P}_X+\tilde{P}_A=2T$ is not admissible (see Fig.~\ref{fig:Admis_XY}). Not that since $\tilde{P}_X>0$, $\tilde{P}_A>0$, and $T>0$ the condition $\rho_{XS}=-\frac{\tilde{P}_X+\tilde{P}_A}{2\sqrt{T\tilde{P}_X}}\le0$ is satisfied. 
When $\tilde{P}_A+\tilde{P}_X\le2T$ and $-1\le\rho_{XS}=-\frac{\tilde{P}_X+\tilde{P}_A}{2\sqrt{T\tilde{P}_X}}\le0$ the argument of the $\log$ function in \eqref{eq:Gaussian_Upper_2} is increasing in $\tilde{P}_X$ and $\tilde{P}_A$, therefore
\begin{align}
    &\argmax\limits_{\substack{0\le\tilde{P}_X\le P_X,\\ 0\le\tilde{P}_A\le P_A:\\-1\le-\frac{\tilde{P}_X+\tilde{P}_A}{2\sqrt{T\tilde{P}_X}}}}\frac{\left(T+\sigma_Y^2\right)\left(4T\tilde{P}_X-(\tilde{P}_X+\tilde{P}_A)^2+4T\sigma_Y^2\right)}{4T\sigma_Y^2\left(T+\sigma_Y^2-\tilde{P}_A\right)}\nonumber\\
    &=\begin{cases}
(\tilde{P}_X=P_X,\tilde{P}_A=P_A)&\text{when}\Squad P_A\le2\sqrt{TP_X}-P_X\Squad\text{and}\Squad P_X+P_A<2T\\
(\tilde{P}_X=P_X,\tilde{P}_A=2\sqrt{TP_X}-P_X)&\text{when}\Squad P_A>2\sqrt{TP_X}-P_X\Squad\text{and}\Squad P_X<T\\
(\tilde{P}_X=P'_X,\tilde{P}_A=P'_A):P'_X+P'_A=2T&\text{when}\Squad P_X+P_A\ge2T\Squad\text{and}\Squad P_X\ge T\\
\end{cases}\label{eq:armax_sol}
\end{align}
Substituting \eqref{eq:armax_sol} into \eqref{eq:Gaussian_Upper_2} leads to
\begin{align}
\begin{cases}
R\le\frac{1}{2}\log\left(\frac{\left(T+\sigma_Y^2\right)\left(4TP_X-(P_X+P_A)^2+4T\sigma_Y^2\right)}{4T\sigma_Y^2\left(T+\sigma_Y^2-P_A\right)}\right)&\text{when}\Squad P_A\le2\sqrt{TP_X}-P_X\Squad\text{and}\Squad P_X+P_A<2T\\
R\le\frac{1}{2}\log\left(\frac{T+\sigma_Y^2}{T+\sigma_Y^2+P_X-2\sqrt{TP_X}}\right)&\text{when}\Squad P_A>2\sqrt{TP_X}-P_X\Squad\text{and}\Squad P_X<T\\
R\le\frac{1}{2}\log\left(1+\frac{T}{\sigma_Y^2}\right)&\text{when}\Squad P_X+P_A\ge2T\Squad\text{and}\Squad P_X\ge T\\
\end{cases}\label{eq:Constaints_Region_6}
\end{align}

\end{appendices}

\bibliographystyle{IEEEtran}
\bibliography{IEEEabrv,bibfile}

\end{document}